\documentclass{amsart}

\pdfoutput=1 % to ensure pdf latex processing

\usepackage{mathrsfs}	% for mathscr font
\usepackage{amsmath}	% allows equation numbering within section
\usepackage{graphicx}
\usepackage{color}	% makes highlights yellow
\usepackage{appendix}	
\usepackage{hyperref}	% makes references which are clickable in a pdf
\usepackage{enumitem}
\usepackage{subcaption}

\usepackage[backend=bibtex,maxbibnames=99,sorting=nyt]{biblatex}
 
\newcommand{\de}{\ensuremath{\partial}}
\newcommand{\dee}{\ensuremath{\textrm{d}}}

\newcommand{\fd}[2]{\ensuremath{ \frac{\dee #1}{\dee #2}}}
\newcommand{\fdf}[1]{\ensuremath{ \frac{\dee}{\dee #1}}}

\newcommand{\inty}[4]{\ensuremath{ \int_{#1}^{#2} \! #3 \, \dee#4 }}

\newcommand{\field}[1]{\mathbb{#1}}
\newcommand{\ip}[2]{\ensuremath{ \left< #1 | #2 \right> } }

\theoremstyle{plain}
\newtheorem{theorem}{Theorem}[section]
\newtheorem{definition}[theorem]{Definition}

\newtheorem{property}[theorem]{Property}
\newtheorem{remark}[theorem]{Remark}
\newtheorem{lemma}[theorem]{Lemma}
\newtheorem{corollary}[theorem]{Corollary}
\newtheorem{proposition}[theorem]{Proposition}
\newtheorem{example}[theorem]{Example}
\newtheorem*{problem}{Problem}
\newtheorem{difficulty}[]{Difficulty}

\numberwithin{equation}{section}
\let\oldhat\hat
\renewcommand{\hat}[1]{\oldhat{\boldsymbol{#1}}}
\bibliography{CIMP-final.bib}

\allowdisplaybreaks

\title[Propagation through a one-dimensional band crossing]{Wavepackets in inhomogeneous periodic media: propagation through a one-dimensional band crossing} 

\author{Alexander Watson and Michael I. Weinstein}

\begin{document}

\begin{abstract}
We consider a model of an electron in a crystal moving under the influence of an external electric field: Schr\"{o}dinger's equation in one spatial dimension with a potential which is the sum of a periodic function $V$ and a smooth function $W$. We assume that the period of $V$ is much shorter than the scale of variation of $W$ and denote the ratio of these scales by $\epsilon$. We consider the dynamics of \emph{semiclassical wavepacket} asymptotic (in the limit $\epsilon \downarrow 0$) solutions which are spectrally localized near to a \emph{crossing} of two Bloch band dispersion functions of the periodic operator $- \frac{1}{2} \de_z^2 + V(z)$. We show that the dynamics is qualitatively different from the case where bands are well-separated: at the time the wavepacket is incident on the band crossing, a second wavepacket is `excited' which has \emph{opposite} group velocity to the incident wavepacket. We then show that our result is consistent with the solution of a `Landau-Zener'-type model. 
\end{abstract}

\maketitle

%\tableofcontents

\section{Introduction} \label{sec:introduction}
\noindent In this work we study the non-dimensionalized, semi-classically scaled, time-dependent Schr\"{o}dinger equation for $\psi^\epsilon (x,t):\field{R} \times [0,\infty) \rightarrow \field{C}$:
\begin{equation} \label{eq:original_equation}
\begin{split}
	&i \epsilon \de_t \psi^\epsilon = - \frac{1}{2} \epsilon^2 \de_x^2 \psi^\epsilon + V\left(\frac{x}{\epsilon}\right)\psi^\epsilon + W(x) \psi^\epsilon\ \equiv\ H^\epsilon\ \psi^\epsilon	\\
	&\psi^\epsilon(x,0) = \psi^\epsilon_0(x).
\end{split}
\end{equation}
Here, $\epsilon$ is a positive real parameter which we assume to be small. We assume throughout that the function $V$ is smooth and 1-periodic so that:
\begin{equation} \label{eq:potential_periodic}
	V(z + 1) = V(z) \text{ for all $z \in \field{R}$},
\end{equation}
and that $W$ is smooth with all derivatives uniformly bounded (this assumption may be relaxed; see Remark 1.2 of \cite{2017WatsonWeinsteinLu}). Equation (\ref{eq:original_equation}) is the independent-particle approximation in condensed matter physics \cite{ashcroft_mermin} for the dynamics of an electron in a crystal described by periodic potential $V$, under the influence of an external electric field generated by a `slowly varying' potential $W$.

Let $E_n$ denote the $n$th Bloch band dispersion function of the periodic operator $- \frac{1}{2} \de_z^2 + V(z)$. It is known that \cite{carles_sparber,2017WatsonWeinsteinLu} for any uniformly \emph{isolated}, or non-degenerate, band $E_n$ (see Figure \ref{fig:lowest_bands_cos_2pix}) there exists a family of explicit asymptotic solutions of (\ref{eq:original_equation}) known as \emph{semiclassical wavepackets} which, for any fixed positive integer $N$, approximate exact solutions up to `Ehrenfest time' $t \sim \ln 1/\epsilon$ up to errors of order $\epsilon^{N}$ in $L^2_x$. The center of mass and average quasi-momentum of these solutions evolve (up to errors of size $o(1)$) along classical trajectories generated by the `Bloch band' Hamiltonian $\mathcal{H}_n := E_n(p)+W(q)$. We refer to such an asymptotic solution as a \emph{wavepacket associated with the band $E_n$}. The `Ehrenfest' time-scale of validity of the asymptotics is known to be the general limiting time-scale of  validity of wavepacket, or coherent state, approximations \cite{schubert_vallejos_toscano}. These results generalize to $d$-dimensional analogs of \eqref{eq:original_equation} \cite{2017WatsonWeinsteinLu,carles_sparber}, see also \cite{1991Gerard,1997GerardMarkowichMauserPoupaud,1996PoupaudRinghofer}.

\begin{figure}
\includegraphics[scale=.4]{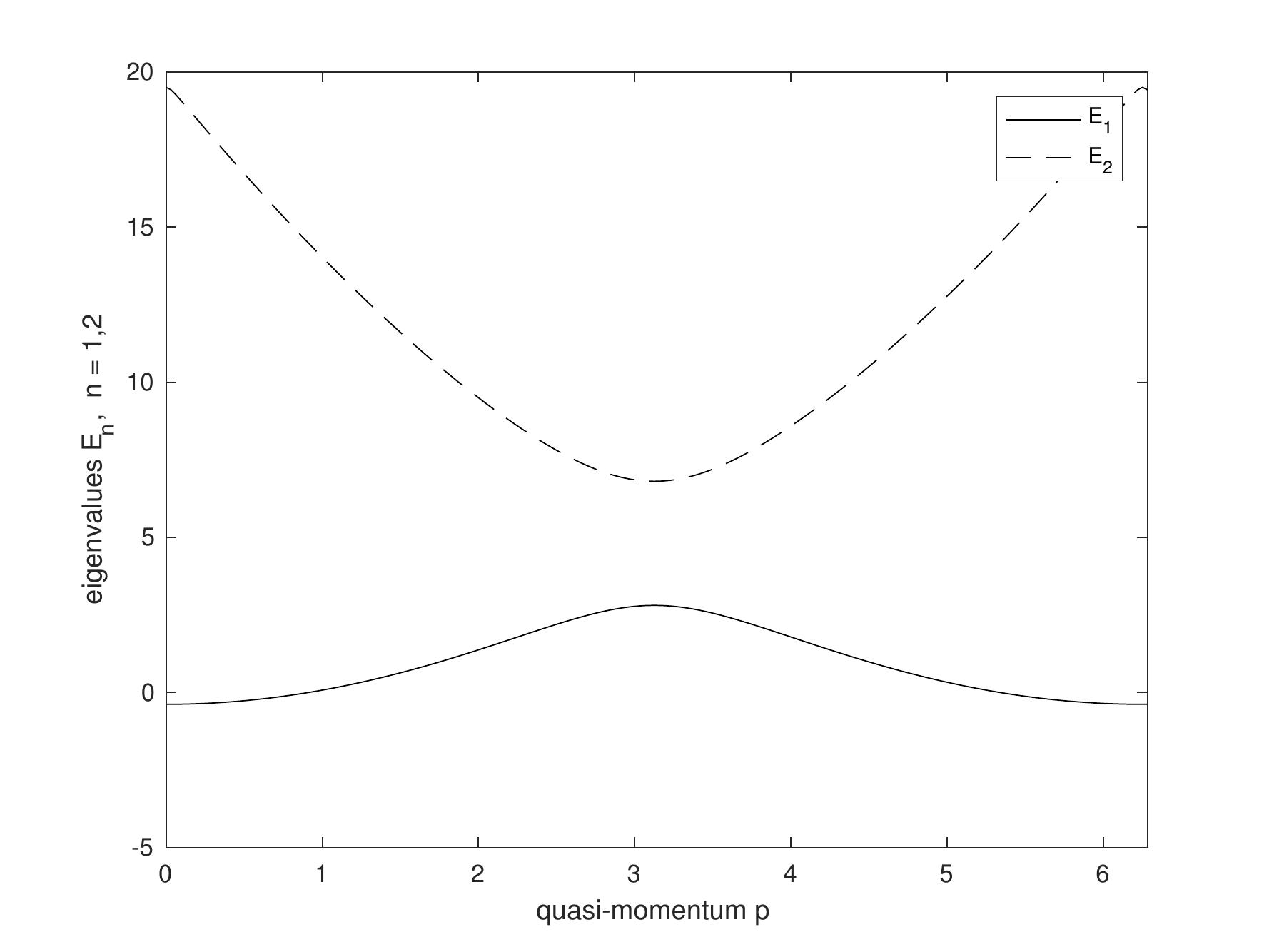}
\caption{ Plot of the two lowest Bloch band dispersion functions $E_1(p), E_2(p)$ for the 1-periodic potential  $V(z) = 4\cos(2\pi z)$. Both bands are \emph{isolated} from each other and all other bands: for all $p \in [0,2\pi]$, $G(E_2(p)) > 0$ and $G(E_1(p)) > 0$ where $G(E_n(p))$ is the spectral band gap function \eqref{eq:nth_spectral_gap_function}. Consequently, the maps $p \mapsto E_1(p), E_2(p)$ are smooth. }
\label{fig:lowest_bands_cos_2pix}
\end{figure}

In this work we consider the following question concerning the dynamics of wave-packets in a situation where two Bloch bands are \emph{not} isolated: 
\begin{problem}
Consider equation (\ref{eq:original_equation}) with initial conditions given by a wavepacket associated with a band $E_n$, for which the classical trajectories associated with $\mathcal{H}_n$ pass through 
 a point in phase space, $(q^*,p^*)$,  where the Bloch band $E_n$ is degenerate; see Figure \ref{fig:crossing_bands}. How are the dynamics different from the isolated band case?\end{problem}
\noindent More precisely, suppose that two bands $E_n(p), E_{n+1}(p)$ touch at a quasi-momentum $p^*$ in the Brillouin zone, but are otherwise non-degenerate in a neighborhood of $p^*$ (see Figure \ref{fig:crossing_bands}). We study a wavepacket associated with the band $E_n$ initially localized in phase space on a classical trajectory $(q(t),p(t))$ generated by $\mathcal{H}_n$ which encounters the crossing after some finite time $t^*$: for some $t^* > 0$, $\lim_{t \uparrow t^*} p(t) = p^*$\footnote{The scenario we consider is \emph{non-generic}: for generic periodic potentials in one spatial dimension, all bands are isolated \cite{reed_simon_4,kuchment,1976Simon}. Non-trivial periodic potentials exhibiting band crossings may be explicitly displayed (see Section \ref{examples}).}.

\begin{figure} 
\includegraphics[scale=.4]{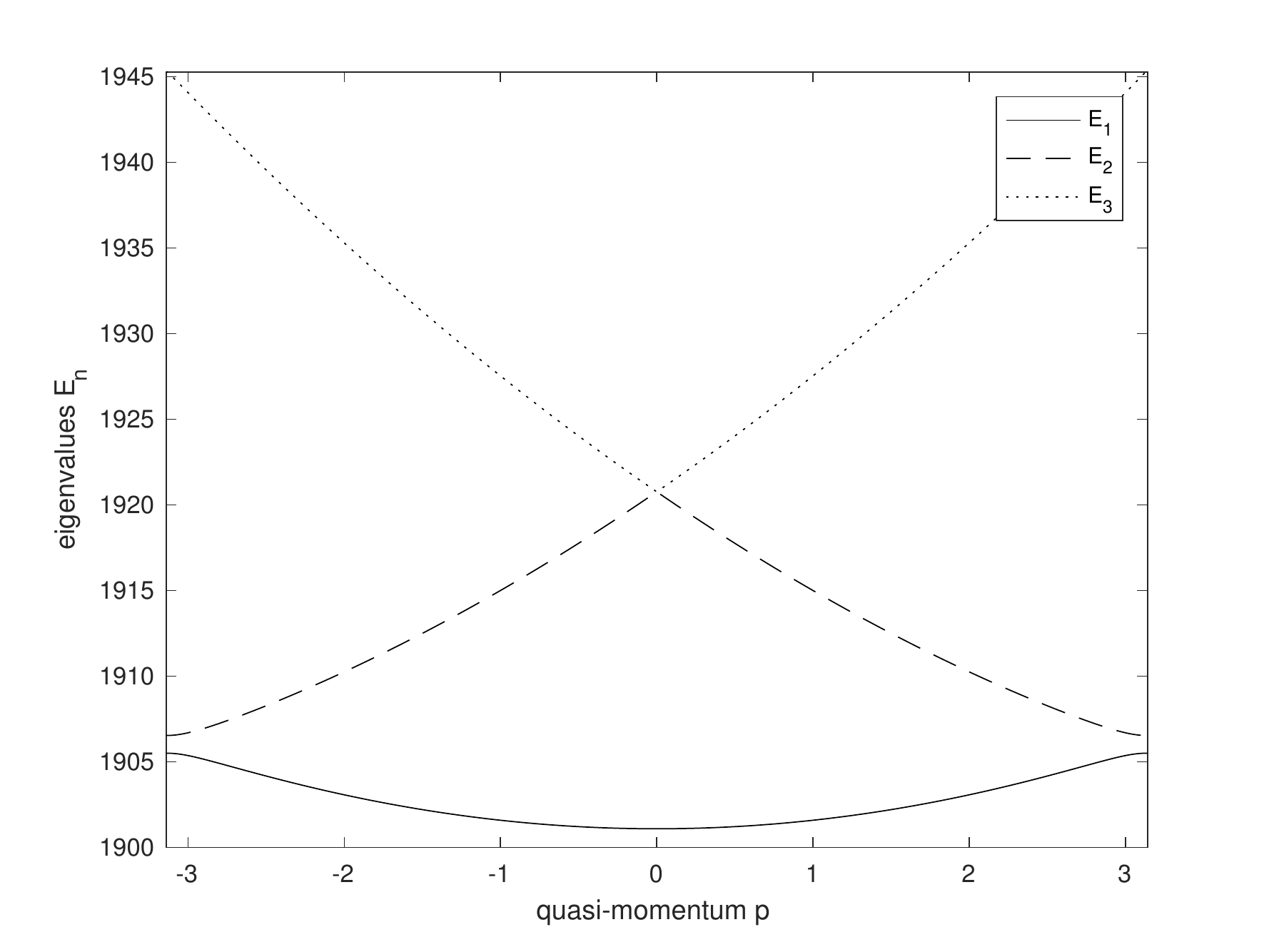}
\caption{Plot of the three lowest Bloch band dispersion functions for $V(z) = \wp_{1/2,i\omega'}(z + i\omega')$, the `one-gap' potential (see Example \ref{ex:n_gap_potentials}), with $\omega' = .8$. The band $E_1(p)$ is isolated over the whole Brillouin zone $[-\pi,\pi]$, but the bands $E_2(p), E_3(p)$ are \emph{degenerate} at $p = 0$. For this choice of potential, for \emph{all} integers $n \geq 2$ the band $E_n(p)$ is degenerate with the band $E_{n+1}(p)$ at either $p = 0$ or $p = \pi$. The one and only  gap is shown.}
\label{fig:crossing_bands}
\end{figure}

We now give a rough statement of our main results; a more precise statement is given in Section \ref{sec:statement_of_results}. Assume that an `incident' wavepacket is \emph{driven} through the crossing so that $\lim_{t \uparrow t^*} \dot{p}(t) = \lim_{t \uparrow t^*} \de_q W(q(t)) \neq 0$. For a precise set up, see the Band Crossing Scenario (Property \ref{classical_path_assumption}). Our main results address: 
\begin{enumerate}
\item {\bf Quantifying the breakdown of `single-band' description as $t \uparrow t^*$; Theorem \ref{th:limit_of_single_band_approximation}}: Fix any positive integer, $N$. For $t \ll t^*$,  the solution of (\ref{eq:original_equation}) can be represented as a wavepacket associated with the band $E_n$ with errors which are ${O}((\sqrt\epsilon)^{N})$ in $L^2(\mathbb{R})$. As $t \uparrow t^*$, this `single-band' description fails to capture the dynamics of the PDE to any order in $\sqrt\epsilon$ higher than order $(\sqrt\epsilon)^0=1$, since it does not incorporate an excited wave associated with the band $E_{n+1}$ whose norm grows to be of the order $\sqrt{\epsilon}$ as  $t$ approaches $t^*$ on the non-adiabatic time-scale $s=(t-t^*)/\sqrt\epsilon$. 
\item {\bf  Coupling of degenerate bands and excitation of a reflected wave-packet; Theorem \ref{th:band_crossing_theorem}:} For $t \sim t^*$ and for $t \gg t^*$ the solution of (\ref{eq:original_equation}) is well-approximated by the sum of two semiclassical wavepackets: a `transmitted' wavepacket associated with the band $E_{n+1}$ with $L^2$-norm of order $1$ and a `reflected' wavepacket associated with the band $E_n$ with $L^2$-norm of order $\sqrt\epsilon$ (Figure \ref{fig:centers_of_mass}). The size of the error terms is $o(\sqrt\epsilon)$ in $L^2(\mathbb{R})$.  The expansion is constructed via a rigorous matched-asymptotic analysis in which the `transmitted'  and `reflected' wave-packets evolve on the emergent non-adiabatic time-scale $s=(t-t^*)/\sqrt\epsilon$. 
\end{enumerate}
In Appendix \ref{app:consistency} we show that these results are consistent with the solution of an appropriate `Landau-Zener'-type model.
\begin{figure}
\includegraphics[scale=.4]{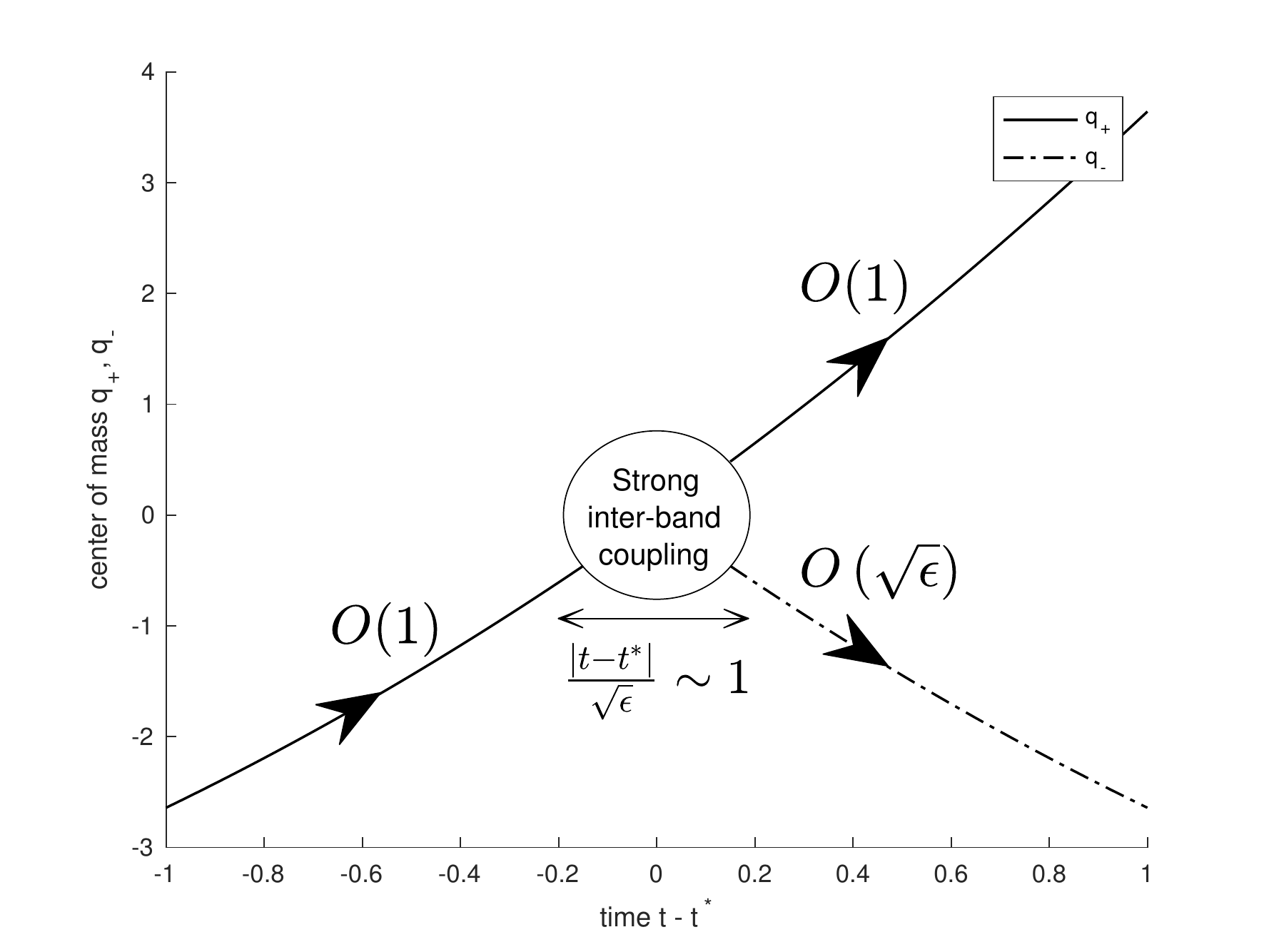}
\caption{Schematic of  center of mass location versus time of the `incident/transmitted' wavepacket $q_+(t)$ and the `excited/reflected' wavepacket $q_-(t)$, which satisfy \eqref{eq:+_band_classical_system} and \eqref{eq:-_band_classical_system}. As $t$ approaches $t^*$ the average quasi-momentum of the incident wavepacket is \emph{degenerate}, {\it i.e.} $p_+(t^*) = p^*$ where $E_+(p^*) = E_-(p^*)$. Inter-band coupling is non-negligible, occurs over the emergent non-adiabatic time-scale $s := \frac{t - t^*}{\sqrt\epsilon}=\mathcal{O}(1)$  and  excites a  second wavepacket. The size in $L^2$ of the `excited' wavepacket is smaller than that of the `incident' wavepacket by a factor of $\sqrt{\epsilon}$ and proportional to the `coupling coefficient' $\ip{\chi_-(\cdot;p^*)}{\de_p \chi_+(\cdot;p^*)}$. Here $E_\pm(p), \chi_\pm(z;p)$ refer to the \emph{smooth continuations} of the band eigenpairs $E_n(p), E_{n+1}(p), \chi_n(z;p), \chi_{n+1}(z;p)$ through the crossing (see Property \ref{band_crossing_assumption} and Figure \ref{fig:crossing_bands_smooth}). Such continuations always exist at one-dimensional band crossings (Theorem \ref{th:all_crossings_smooth}).}
\label{fig:centers_of_mass}
\end{figure}

The proofs of Theorems \ref{th:limit_of_single_band_approximation} and \ref{th:band_crossing_theorem} rely on the existence of \emph{smooth continuations} of the Bloch band dispersion functions $E_n, E_{n+1}$ through the crossing point $p^*$; see Property \ref{band_crossing_assumption} and Figure \ref{fig:crossing_bands_smooth}.  Such continuations always exist in one spatial dimension and more generally at arbitrary codimension 1 eigenvalue band crossings (see Theorem \ref{th:all_crossings_smooth}, and \cite{hagedorn} for the general case). The key tools in our proof are the method of matched asymptotic expansions (following Hagedorn \cite{hagedorn}) combined with the basic result of Carles and Sparber \cite{carles_sparber} on semiclassical wavepacket solutions in the presence of a periodic background potential.  

Our proof does not readily generalize to cases where no smooth continuation of eigenvalue bands exists; for example at `conical', or `Dirac', points which are codimension 2 eigenvalue band crossings \cite{fefferman_weinstein_diracpoints,fefferman_weinstein,fefferman_leethorp_weinstein,fefferman_leethorp_weinstein_2,2017FeffermanLee-ThorpWeinstein,2017FeffermanLee-ThorpWeinstein_2}. The dynamics of semiclassical wavepackets at such crossings was studied by Hagedorn in the context of the Born-Oppenheimer approximation of molecular physics \cite{hagedorn}. Adapting his methods to the present context is the subject of ongoing work. A model of the dynamics at a `conical' codimension 2 Bloch band degeneracy was derived in \cite{2016Fermanian-KammererMehats}. 
\begin{figure}
\includegraphics[scale=.4]{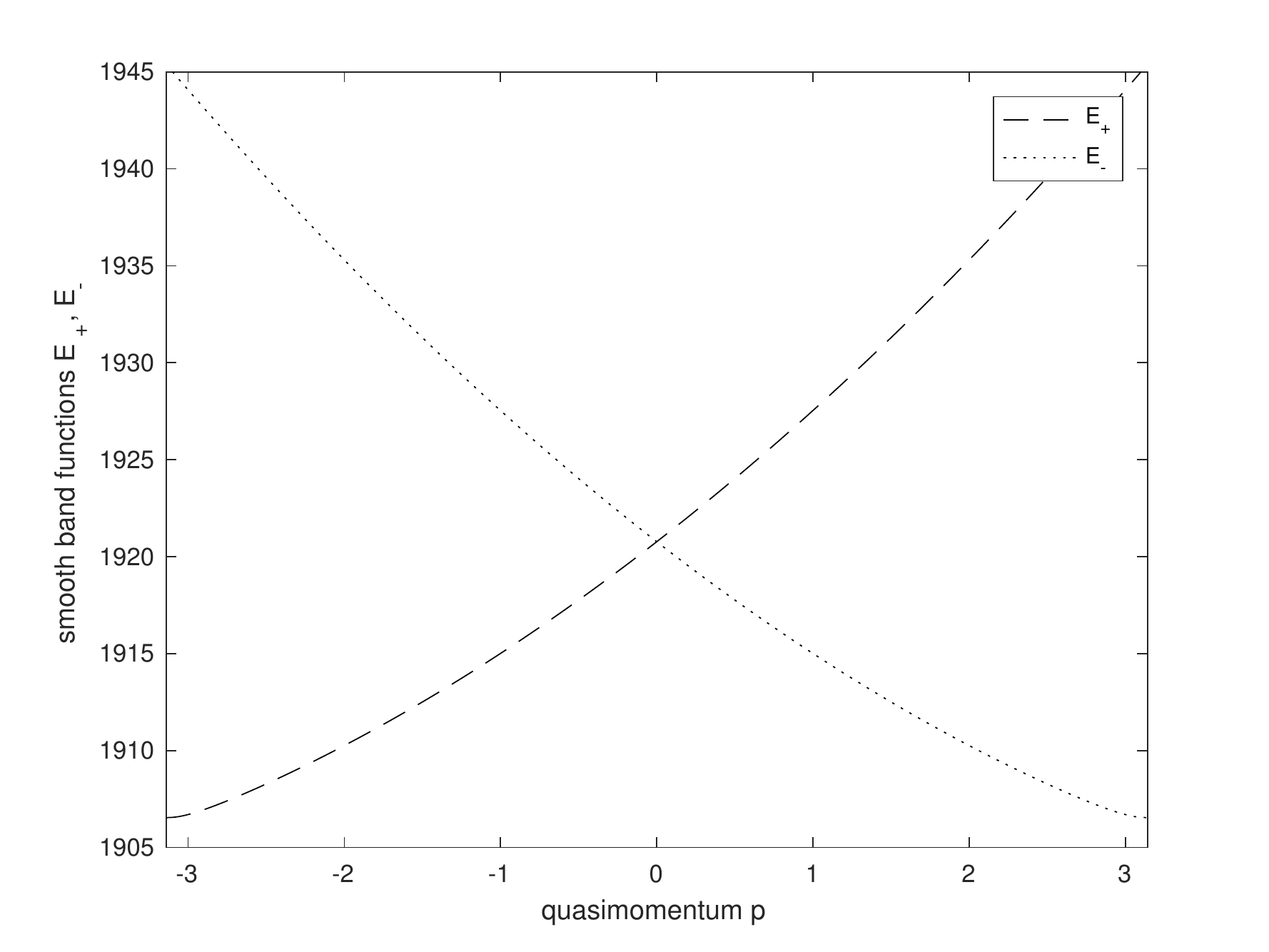}
\caption{Plot of the maps $E_+(p), E_-(p)$ defined by \eqref{eq:def_smooth_bands} with $n = 2$ and where $E_2(p), E_3(p)$ are the second and third lowest Bloch band dispersion functions when $V(z) = \wp_{1/2,i\omega'}(z + i\omega')$, the `one-gap' potential with $\omega' = .8$. The lowest three Bloch bands of this potential are shown in Figure \ref{fig:crossing_bands}.} 
\label{fig:crossing_bands_smooth}
\end{figure}

Quantum dynamics at an eigenvalue band crossing was studied by Landau \cite{1932Landau} and Zener \cite{1932Zener} in the 1930s. Their simplified model, involving an explicitly time-dependent Hamiltonian, captures many of the essential features of the dynamics in a neighborhood of the crossing.

For a review of their work, see \cite{Nakamura}. A rigorous proof of their main result, the `Landau-Zener' formula for the probability of an inter-band transition, was given by Hagedorn \cite{1991Hagedorn} for the case of an `avoided' crossing. Validity of the formula for true crossings was then proved by Joye \cite{1994Joye}. In Appendix \ref{app:consistency} we show that our main theorem (Theorem \ref{th:band_crossing_theorem}) is consistent with the solution of an appropriate `Landau-Zener'-type simplified model of the present problem. 

The propagation of Wigner measures through crossings (both true and avoided) in the context of the Born-Oppenheimer approximation has been studied by Fermanian-Kammerer, Lasser, and others \cite{2002Fermanian-KammererGerard,2003Fermanian-KammererGerard,2005LasserTeufel,2008Fermanian-KammererLasser,2013ChaiJinLi,2013Fermanian-KammererGerardLasser,2014ChaiJinLiMorandi,2017Fermanian-KammererLasser}. The dynamics of semiclassical wavepackets at an avoided crossing in this context was studied by Hagedorn and Joye \cite{hagedorn_joye}. A discussion of crossing phenomena from the perspective of normal forms and microlocal analysis was given by Colin de Verdiere et al. \cite{1999ColindeVerdiere,2003ColinDeVerdiere,2004ColinDeVerdiere}.

\subsection{Notation} \label{sec:notation}
\begin{itemize}
\item We introduce a natural set of function spaces which encode both smoothness and spatial decay. For every $l \in \field{N}$:
\begin{equation} \label{eq:sigma_l_spaces}
	\Sigma^l(\field{R}) := \left\{ f \in L^2(\field{R}) : \| f \|_{\Sigma^l} := \sum_{|\alpha| + |\beta| \leq l} \| y^\alpha (- i \de_y)^\beta f(y) \|_{L^2_y} < \infty \right\}.
\end{equation}
\item The space of Schwartz functions $\mathcal{S}(\field{R})$ is the space of functions defined as:
\begin{equation}
	\mathcal{S}(\field{R}) := \cap_{l \in \field{N}} \Sigma^l(\field{R}). 
\end{equation}
\item We will refer throughout to the space of $L^2$-integrable functions which are $1-$periodic:
\begin{equation}
	L^2_{per} := \left\{ f \in L^2_{loc}(\field{R}) : f(z+1) = f(z) \text{ at almost every $z \in \field{R}$} \right\}.
\end{equation}
\item For functions of period $1$, the  Brillouin zone $\mathcal{B}$ may be chosen to be any real interval of length $2\pi$. To arrange that the band degeneracy
 occurring at quasi-momentum $p^*=\pi$ is located at 
an interior point of $\mathcal{B}$, we fix $\mathcal{B} := [0,2\pi]$. 
\item Conventions for  $L^2$-inner products and induced norms:
\begin{equation}
	\ip{ f }{ g }_{L^2(\mathcal{D})} := \inty{\mathcal{D}}{}{ \overline{f(x)} g(x) }{x}, \quad \| f \|_{L^2(\mathcal{D})} := \ip{f}{f}_{L^2(\mathcal{D})}^{1/2} 
\end{equation}
For brevity, when $\mathcal{D} = \field{R}$ we omit the domain of integration:
\begin{equation}
	\ip{f}{g}_{L^2} := \inty{\field{R}}{}{ \overline{f(x)} g(x) }{x}, \quad \| f \|_{L^2} := \ip{f}{f}_{L^2}^{1/2},
\end{equation}
and when $\mathcal{D} = [0,1]$ we omit all subscripts:
\begin{equation}
	\ip{f}{g} := \inty{[0,1]}{}{ \overline{f(x)} g(x) }{x}, \quad \| f \| := \ip{f}{f}^{1/2}.
\end{equation}
\item We make the following convention for the Fourier transform and its inverse:
\begin{equation}
	\mathcal{F}_x \{ f(x) \}(\xi) := \inty{- \infty}{\infty}{ e^{- 2 \pi i \xi x} f(x) }{x}, \quad \mathcal{F}^{-1}_{\xi} \{ g \}(x) := \inty{- \infty}{\infty}{ e^{2 \pi i \xi x} g(\xi) }{\xi}.
\end{equation}
\end{itemize}

\subsection*{Acknowledgements} The authors wish to thank George Hagedorn, Jianfeng Lu, and Christof Sparber for stimulating discussions. This research was supported in part by National Science Foundation Grant Nos. DMS-1412560, DMS-1620418 and Simons Foundation Math + X Investigator Award \#376319 (Michael I. Weinstein).

\section{Review of Floquet-Bloch theory and the isolated band theory of  wavepackets} \label{sec:statement_of_results}
\subsection{Floquet-Bloch theory}
In order to state our results we require some background on the spectral theory of the Schr\"{o}dinger operator:
\begin{equation} \label{eq:periodic_operator}
	H := - \frac{1}{2} \de^2_z + V(z)
\end{equation}
where $V$ is 1-periodic (see \cite{kuchment, reed_simon_4} for proofs and details). Consider the family of self-adjoint eigenvalue problems parameterized by the real parameter $p$: 
\begin{equation} \label{eq:quasi_periodic_eigenvalue_problem}
\begin{split}
	&H \Phi(z;p) = E(p) \Phi(z;p)	\\	
	&\Phi(z + 1;p) = e^{i p} \Phi(z;p) \text{ for all } z \in \field{R}.
\end{split}
\end{equation}
Because of the explicit $2 \pi$-periodicity of the boundary condition, there is no loss of generality in restricting our attention to $p \in \mathcal{B}$, where $\mathcal{B}$ is any real interval of length $2 \pi$. $\mathcal{B}$ is usually fixed to be $[-\pi,\pi]$ or $[0,2\pi]$ and referred to as the Brillouin zone. The eigenvalue problem \eqref{eq:quasi_periodic_eigenvalue_problem} is equivalent (by the transformation $\Phi(z;p) = e^{i p z} \chi(z;p)$) to the family of self-adjoint eigenvalue problems with 1-periodic boundary conditions:
\begin{equation} \label{eq:reduced_eigenvalue_problem}
\begin{split}
	&H(p) \chi(z;p) = E(p) \chi(z;p)	\\	
	&\chi(z + 1;p) = \chi(z;p) \text{ for all } z \in \field{R} \\
	&H(p) := \frac{1}{2} ( p - i \de_z )^2 + V(z).
\end{split}
\end{equation}
For fixed $p$, the spectrum of the equivalent operators \eqref{eq:quasi_periodic_eigenvalue_problem} and \eqref{eq:reduced_eigenvalue_problem} is real and discrete and the eigenvalues can be ordered with multiplicity: 
\begin{equation}
	E_1(p) \leq E_2(p) \leq ... \leq E_n(p) \leq ... \, .
\end{equation}
The maps $p \mapsto E_n(p)$, for $p$ varying over $\mathcal{B}$, are known as the spectral band functions. Their graphs are Lipschitz continuous and are called the dispersion curves of $H$. The set of all dispersion curves as $p$ varies over $\mathcal{B}$ is called the band structure of $H$ (\ref{eq:periodic_operator}). The associated normalized eigenfunctions $\chi_n(z;p)$ of \eqref{eq:reduced_eigenvalue_problem} are a basis of the space:
\begin{equation}
	L^2_{per} := \left\{ f \in L^2_{loc} : f(z + 1) = f(z) \text{ at almost every $z \in \field{R}$} \right\},
\end{equation} 
and any function in $L^2(\field{R})$ may be expressed as a superposition of \emph{Bloch waves}:
\begin{equation} \label{eq:bloch_waves}
	\left\{ \Phi_n(z;p) = e^{i p z} \chi_n(z;p) : n \in \field{N}, p \in \mathcal{B} \right\}. 
\end{equation}
The $L^2$-spectrum of the operator (\ref{eq:periodic_operator}) is the union of the real intervals swept out by the spectral band functions $E_n(p)$:
\begin{equation}
	\sigma(H)_{L^2(\field{R}^d)} = \cup_{n \in \field{R}} \left\{ E_n(p) : p \in \mathcal{B} \right\}.
\end{equation}
We define a measure of the spectral gap or separation at quasimomentum $p \in \mathcal{B}$ between $E_n$ and all other spectral band functions satisfying (\ref{eq:reduced_eigenvalue_problem}):
\begin{equation} \label{eq:nth_spectral_gap_function}
	G(E_n(p)) := \min_{m \neq n} | E_n(p) - E_{m}(p) |. 
\end{equation}
We make the following definitions:
\begin{definition} \label{def:definition_isolated_degenerate}
Let $E_n(p)$ denote an eigenvalue band of either of the equivalent eigenvalue problems \eqref{eq:quasi_periodic_eigenvalue_problem}, \eqref{eq:reduced_eigenvalue_problem} and let $\tilde{p} \in \mathcal{B}$. If:
\begin{equation}
	G(E_n(\tilde{p})) > 0,
\end{equation}
then we will say that $E_n(p)$ is \emph{isolated} at $\tilde{p}$. If:
\begin{equation}
	G(E_n(\tilde{p})) = 0,
\end{equation}
then we will say that $E_n(p)$ is involved in a \emph{Bloch band degeneracy} or \emph{band crossing} at $\tilde{p}$.
\end{definition}

\subsection{Isolated band theory}
\begin{property}[Isolated Band Property] \label{isolated_band_assumption}
Let $E_n$ denote a band dispersion function satisfying (\ref{eq:reduced_eigenvalue_problem}) for $p \in \mathcal{B}$. Let  $t_0 < t_1\le\infty$ and $q_{0}, p_{0} \in \field{R} \times \mathcal{B}$ be such that the equations of motion of the classical Hamiltonian $\mathcal{H}_n(p,q) := E_n(p) + W(q)$:
\begin{align} 
	\label{eq:classical_system}&\dot{q}(t) = \de_p E_n(p(t)), &&\dot{p}(t) = - \de_q W(q(t)) 	\\
	\nonumber&q(t_0) = q_0 &&p(t_0) = p_{0}
\end{align}
have a unique smooth solution $(q(t),p(t))$ for $t \in [t_0,t_1)$ such that $E_n$ is isolated along the trajectory $(q(t),p(t))$ for $t \in [t_0,t_1)$; i.e: 
\begin{equation}
\label{min-gap}
	M(t_0,t_1) := \inf_{t \in [t_0,t_1)} G(E_n(p(t))) > 0,
\end{equation}
where $G(E_n(p))$ is defined by (\ref{eq:nth_spectral_gap_function}). 
\end{property} 
For arbitrary constant $S_0 \in \field{R}$ we let $S(t)$ denote the classical action along the path $(q(t),p(t))$: 
\begin{equation} \label{eq:action_integral}
	S(t) = S_0 + \inty{t_0}{t}{ p(t') \de_{p} E_n(p(t')) - E_n(p(t')) - W(q(t')) }{t'}.
\end{equation} 
For arbitrary $a^0_0(y) \in \mathcal{S}(\field{R})$, let $a^0(y,t)$ denote the unique solution of Schr\"{o}dinger's equation with a time-dependent harmonic oscillator Hamiltonian depending parametrically on the classical trajectory $(q(t),p(t))$ with initial data specified at $t_0$ by $a^0_0(y)$:
\begin{equation} 
\begin{split} \label{eq:envelope_equation} 
	&i \de_t a^0(y,t) = \mathscr{H}(t) a^0(y,t),	\\
	&\mathscr{H}(t) := \frac{1}{2} \de_p^2 E_n(p(t)) (- i \de_y)^2 + \frac{1}{2} \de^2_q W(q(t)) y^2 + \de_q W(q(t)) \mathcal{A}_n(p(t)),	\\
	&a^0(y,t_0) = a^0_0(y).
\end{split}
\end{equation}
Here, $p\in\mathcal{B}\mapsto\mathcal{A}_n(p)$ denotes the $n$-th band Berry connection (see Section \ref{sec:notation} for conventions regarding inner products and norms):  
\begin{equation} \label{eq:berry_connection}
	\mathcal{A}_{n}(p) := i \ip{\chi_n(\cdot;p))}{\de_{p}\chi_n(\cdot;p)}.
\end{equation}
Since the $\chi_n(z;p)$ are assumed normalized: for all $p\in\field{R}$, $\| \chi_n(\cdot,p) \| = 1$,
%\begin{equation}
%	\text{for all } p \in \field{R}, \text{ } , 
%\end{equation}
it follows that  $\mathcal{A}_n(p)$ is real-valued. The term $\de_q W(q(t)) \mathcal{A}_n(p(t))$ in \eqref{eq:envelope_equation} depends only on $t$ and leads to an overall phase shift in the solution of \eqref{eq:envelope_equation} known as Berry's phase.
\begin{remark} \label{rem:adiabatic_gauge}
Given any path $p(t)$ through parameter space which does not self-intersect (i.e. $p(t_1) = p(t_2) \implies t_1 = t_2$ for all $t_1, t_2$) it is possible to choose phases of the eigenfunctions $\chi_n(z;p)$ in such a way that the Berry connection \eqref{eq:berry_connection} is zero when evaluated along the curve $p(t)$ for all $t$. This choice is known as the \emph{adiabatic gauge}. See Proposition 3.1 of \cite{hagedorn}, for example. In dimensions greater than $1$, the integral of the Berry connection along a closed path through parameter space is not, in general, zero and is physically observable \cite{berry}.
\end{remark} 

%It is useful to define the following short-hand notation for the Hamiltonian operator which appears in \eqref{eq:original_equation}:
%\begin{equation} \label{eq:def_H}
%	H^\epsilon := - \frac{1}{2} \epsilon^2 \Delta_x + V\left(\frac{x}{\epsilon}\right) + W(x).
%\end{equation}

We now state a mild refinement of the result of Carles-Sparber \cite{carles_sparber} which we find more 
directly applicable: 
\begin{theorem}[Order $1$ wave-packet] \label{th:isolated_band_theorem_1}
Let $(q(t),p(t))$ denote the classical trajectory generated by the Hamiltonian $\mathcal{H}_n(p,q)=E_n(p)+W(q)$,
where $p\mapsto E_n(p)$ denotes the $n^{th}$ spectral band function for the periodic Schr\"{o}dinger operator  $-\frac12 \partial_z^2+V(z)$.
Assume that band $E_n$ satisfies the Isolated Band Property \ref{isolated_band_assumption} along the trajectory 
 $(q(t),p(t))$ for $t\in[t_0,t_1)$, {\it i.e.}  $M(t_0,t_1)>0$; see \eqref{min-gap}. 
 
 Let $S(t)$ be as in \eqref{eq:action_integral} with $S_0 \in \field{R}$ and $a^0(y,t)$ be the unique solution of \eqref{eq:envelope_equation} with initial data $a_0^0(y)\in\mathcal{S}(\mathbb{R})$. 
  
   Then, for sufficiently small $\epsilon > 0$ the following holds. Let $\psi^\epsilon(x,t)$ denote the unique solution of the initial value problem (\ref{eq:original_equation}) with approximate `Bloch wavepacket' initial data given at $t = t_0$:
\begin{equation} 
\begin{split}
	&i \epsilon \de_t \psi^\epsilon = H^\epsilon \psi^\epsilon 		\\
	&\psi^\epsilon(x,t_0) = \epsilon^{-1/4}\ e^{i S_0/\epsilon} e^{i p_0 (x - q_0) / \epsilon}\ a^0_0\left(\frac{x - q_0}{\sqrt{\epsilon}}\right)\ \chi_n\left(\frac{x}{\epsilon};p_0\right). 
\end{split}
\end{equation} 
For $t \in [t_0,t_1)$, the solution evolves as a modulated `Bloch wavepacket' plus a corrector $\eta^\epsilon(x,t)$:
\begin{equation} \label{eq:wavepacket_evolution}
	\psi^\epsilon(x,t) = \epsilon^{-1/4}\ e^{i S(t)/\epsilon}\ e^{i p(t) (x - q(t)) / \epsilon}\ a^0\left(\frac{x - q(t)}{\sqrt{\epsilon}},t\right)\ \chi_n\left(\frac{x}{\epsilon};p(t)\right)\ +\ \eta^\epsilon(x,t)
\end{equation}
where the leading order term is of order 1 in $L^2(\field{R})$ and the corrector $\eta^\epsilon$ satisfies:
\begin{equation} \label{eq:bound_on_eta_1}
	\| \eta^\epsilon(\cdot,t) \|_{L^2} \leq C e^{c ( t-t_0)} \sqrt{\epsilon},\qquad t_0\le t< t_1.
\end{equation}
The constants $C > 0, c > 0$ depend on $M(t_0,t_1)$  and the initial data specified at $t_0$, are independent of 
$\epsilon$ and do not depend otherwise on $t_0$ and $ t_1$. Moreover, the constant $C$ in \eqref{eq:bound_on_eta_1} satisfies $C\uparrow\infty$ as $M(t_0,t_1)\downarrow0$.

In particular,  if  $M(t_0,\infty)>0$ then 
\begin{equation} \label{eq:bound_on_ehrenfest_timescale}
	\sup_{t \in [t_0,\tilde{C} \ln 1/\epsilon]} \| \eta^\epsilon(\cdot,t) \|_{L^2} = o(1), 
\end{equation}
where $\tilde{C}$ is any constant such that $\tilde{C} < \frac{1}{2 c}$.
\end{theorem}
\begin{remark}\label{ehrenfest-time}
 The timescale $t \sim \ln 1/\epsilon$ is known as `Ehrenfest time' and is known to be the general limit of applicability of wavepacket, or coherent state, approximations (see \cite{schubert_vallejos_toscano} and references therein).
\end{remark}
It is convenient at this point to introduce a short-hand notation for the leading order ($O(1)$ in $L^2$) `Bloch wavepacket' asymptotic solution associated with the band $E_n$ with centering along the classical trajectory $(q(t), p(t))$ and envelope function $a^0(y,t)$ \eqref{eq:wavepacket_evolution}:
\begin{equation} \label{eq:WP_notation_0}
\begin{split}
	&\text{WP}^{0,\epsilon}[S(t),q(t),p(t),a^0(y,t),\chi_n(z;p(t))](x,t) := 	\\
	&\epsilon^{-1/4}\ e^{i S(t)/\epsilon}\ e^{i p(t) (x - q(t)) / \epsilon}\ a^0\left(\frac{x - q(t)}{\sqrt{\epsilon}},t\right)\  \chi_n\left(\frac{x}{\epsilon};p(t)\right).
\end{split}
\end{equation}
In our analysis we require a refinement of Theorem 1.1 of \cite{2017WatsonWeinsteinLu} where it was demonstrated how to compute corrections to the asymptotic solution (\ref{eq:wavepacket_evolution}) in order to improve the error bound \eqref{eq:bound_on_eta_1} by a factor of $\sqrt{\epsilon}$. 

For any $a^1_0(y) \in \mathcal{S}(\field{R})$, let $a^1(y,t)$ denote the unique solution of the following inhomogeneous Schr\"{o}dinger equation with initial data specified at $t_0$ by $a^1_0(y)$ \emph{driven} by the solution $a^0(y,t)$ of \eqref{eq:envelope_equation}: 
\begin{equation} \label{eq:first_order_envelope_equation}
\begin{split}
	&\left(\ i \de_t\ -\ \mathscr{H}(t)\ \right) a^1(y,t) = \mathscr{I}(t) a^0(y,t), 	\\
	&\mathscr{I}(t) := \frac{1}{6} \de_p^3 E_n(p(t)) (- i \de_y)^3 + \frac{1}{6} \de_q^3 W(q(t)) y^3 	\\
	&+ \de_q W(q(t)) \de_p \mathcal{A}_n(p(t)) (- i \de_y) + \de^2_q W(q(t)) \mathcal{A}_n(p(t)) y, 	\\
	&a^1(y,t_0) = a^1_0(y).
\end{split}
\end{equation} 
Again, $\mathcal{A}_n(p)$ denotes the Berry connection, displayed in \eqref{eq:berry_connection}. We next introduce a convenient short-hand notation for the `Bloch wavepacket' asymptotic solution associated with the band $E_n$ including a first-order in $\sqrt\epsilon$ correction to $\text{WP}^{0,\epsilon}$ in \eqref{eq:WP_notation_0}:
\begin{equation} \label{eq:WP_notation_1}
\begin{split}
	&\text{WP}^{1,\epsilon}[S(t),q(t),p(t),a^0(y,t),a^1(y,t),\chi_n(z;p(t))](x,t) := 	\\
	&\epsilon^{-1/4} e^{i S(t) / \epsilon} e^{i p(t) (x - q(t)) / \epsilon} \left\{ a^0 \left( \frac{x - q(t)}{\sqrt{\epsilon}},t\right) \chi_n\left(\frac{x}{\epsilon};p(t)\right) \right.	\\
	&\left. + \sqrt{\epsilon} \left[ a^1 \left(\frac{x - q(t)}{\sqrt{\epsilon}},t\right) \chi_n\left(\frac{x}{\epsilon};p(t)\right) + (- i \de_y) a^0\left(\frac{x - q(t)}{\sqrt{\epsilon}},t\right) \de_p \chi_n\left(\frac{x}{\epsilon};p(t)\right) \right] \right\}.
\end{split}	
\end{equation}
Then, we have the following mild generalization of the result of Theorem 1.1 in \cite{2017WatsonWeinsteinLu}:
\begin{theorem}[Order 1 wave-packet with order $\sqrt\epsilon$ correction] \label{th:isolated_band_theorem_2}
Assume the same setting as in Theorem  \ref{th:isolated_band_theorem_1}, in particular that the Isolated  Band Property \ref{isolated_band_assumption} holds along the trajectory $(p(t),q(t))$ of the classical Hamiltonian $\mathcal{H}_n=E_n(p)+W(q)$ for $t\in[t_0,t_1)$, where $t_0<t_1\le\infty$. 

Let $S(t)$ be as in (\ref{eq:action_integral}) with initial action $S(0)=S_0\in\field{R}$.  Let $a^0(y,t)$ as in (\ref{eq:envelope_equation}) and  $a^1(y,t)$ be as in \eqref{eq:first_order_envelope_equation} with $a^0_0(y)$ and $ a^1_0(y) \in \mathcal{S}(\field{R})$. 

Then, for sufficiently small $\epsilon > 0$, we have that the unique solution $\psi^\epsilon(x,t)$ of the initial value problem (\ref{eq:original_equation}) with approximate `Bloch wavepacket' initial data with corrections proportional to $\sqrt{\epsilon}$ given at $t = t_0$:
\begin{equation} \label{eq:approximate_bloch_wp_data}
\begin{split}
	&i \epsilon \de_t \psi^\epsilon = H^\epsilon \psi^\epsilon	\\
	&\psi^\epsilon(x,t_0) = \text{\emph{WP}}^{1,\epsilon}[S_0,q_0,p_0,a^0_0(y),a^1_0(y),\chi_n(z;p_0)](x) + O_{L^2_x}(\epsilon)
\end{split}
\end{equation} 
evolves as a modulated `Bloch wavepacket' plus a corrector $\eta^\epsilon(x,t)$:
\begin{equation}
	\psi^\epsilon(x,t) = \text{\emph{WP}}^{1,\epsilon}[S(t),q(t),p(t),a^0(y,t),a^1(y,t),\chi_n(z;p(t))](x,t) + \eta^\epsilon(x,t)
\end{equation}
where the corrector $\eta^\epsilon$ satisfies, for $t \in [t_0,t_1)$, the bound:
\begin{equation}
	\| \eta^\epsilon(\cdot,t) \|_{L^2} \leq C e^{c t} \epsilon
\end{equation}
where the constants $C > 0, c > 0$ are as stated in Theorem \ref{th:isolated_band_theorem_1}.

Furthermore, it follows that if $M(t_0,\infty)>0$, then we have the following error bound on the Ehrenfest time-scale:
\begin{equation}
	\sup_{t \in [0,\tilde{C} \ln 1/\epsilon]} \| \eta^\epsilon(\cdot,t) \|_{L^2} = o(\sqrt{\epsilon}), 
\end{equation}
where $\tilde{C}$ is any constant such that $\tilde{C} < \frac{1}{2 c}$, cf. \eqref{eq:bound_on_ehrenfest_timescale}. 
\end{theorem}
\begin{remark} \label{rem:high_order_WP}
By a natural extension of the methods of \cite{carles_sparber} and \cite{2017WatsonWeinsteinLu} one may derive, for any integer $k \geq 0$, `{}$k$th-order Bloch wavepacket' approximate solutions:
\begin{equation}
	\text{\emph{WP}}^{k,\epsilon}[S(t),q(t),p(t),a^0(y,t),a^1(y,t),a^2(y,t),...,a^k(y,t),\chi_n(z;p(t))](x,t)
\end{equation}
such that the exact solution $\psi^\epsilon(x,t)$ of \eqref{eq:original_equation} with `$k$-th order Bloch wavepacket' initial data:
\begin{equation}
	\psi^\epsilon_0(x) = \text{\emph{WP}}^{k,\epsilon}[S_0,q_0,p_0,a^0_0(y),a^1_0(y),a^2_0(y),...,a^k_0(y),\chi_n(z;p_0)](x)
\end{equation}
satisfies:
\begin{equation}
\begin{split}
	&\psi^\epsilon(x,t) = \text{\emph{WP}}^{k,\epsilon}[S(t),q(t),p(t),a^0(y,t),a^1(y,t),a^2(y,t),...,a^k(y,t),\chi_n(z;p(t))](x,t) \\
	&\hphantom{ \psi^\epsilon(x,t) = } + o_{L^2_x}(\epsilon^{k/2}) 
\end{split}
\end{equation}
up to `Ehrenfest time' $t \sim \ln 1/\epsilon$. Note that each function $\text{\emph{WP}}^{k,\epsilon}[...](x,t)$ depends on $k+1$ envelope functions $a^0(y,t), a^1(y,t), a^2(y,t),...,a^k(y,t),$ each of which satisfies a suitable Schr\"{o}dinger equation driven by the $k$ previously defined envelope functions. Hence $a^2(y,t)$ satisfies a Schr\"{o}dinger equation driven by $a^0(y,t)$ and $a^1(y,t)$ and so on. 
\end{remark}

\section{Statement of results on dynamics at band crossings} 

\subsection{Linear band crossings}\label{LBC}
We next give a  precise discussion of the character of one-dimensional Bloch band degeneracies (band crossings). The following property describes a \emph{linear band crossing}, illustrated in Figures \ref{fig:crossing_bands} and \ref{fig:crossing_bands_smooth}. In Theorem \ref{th:all_crossings_smooth} we assert that all band crossings in one dimension are of this type:
\begin{property}[Linear band crossing] \label{band_crossing_assumption} 
Let $E_{n}(p), E_{n+1}(p)$ denote two spectral band functions satisfying \eqref{eq:reduced_eigenvalue_problem} for $p \in \mathcal{B}$. Let $p^*\in U$, where $U$ is an open subset of $\mathcal{B}$, be  such that:
\begin{enumerate}[label=\emph{(A\arabic*)}]
\item The bands $E_{n}$ and $E_{n+1}$ are degenerate at $p^*$, and this degeneracy is unique in $U$: \label{item:crossing}
\begin{equation}
\begin{split}
	&E_n(p^*) = E_{n+1}(p^*)	\\
	&\text{if } \tilde{p}^* \in U \text{ and } E_n(\tilde{p}^*) = E_{n+1}(\tilde{p}^*), \text{ then } \tilde{p}^* = p^*. 
\end{split}
\end{equation}
\item \label{other_bands_isolated_away} The bands $E_{n+1}, E_n$ are uniformly isolated from the rest of the spectrum for all $p \in U$, i.e. there exists a positive constant $M > 0$ such that:
\begin{equation} \label{eq:crossing_bands_isolated_from_others}
	\min_{p \in \overline{U}} \; \min_{m \notin \{ n, n+1 \} } \left\{ | E_{m}(p) - E_{n+1}(p) |, |E_n(p) - E_m(p)| \right\} \geq M > 0 
\end{equation}
\item The maps: 
\begin{equation} \label{eq:def_smooth_bands}
\begin{split}
	&p \mapsto (E_+(p),\chi_+(z;p)) := \begin{cases} (E_{n}(p),\chi_{n}(z;p)) &\text{for } p \in U \text{ and } p < p^* \\ (E_{n+1}(p),\chi_{n+1}(z;p)) &\text{for } p \in U \text{ and } p \geq p^* \end{cases} 	\\
	&p \mapsto (E_-(p),\chi_-(z;p)) := \begin{cases} (E_{n+1}(p),\chi_{n+1}(z;p)) &\text{for } p \in U \text{ and } p < p^* \\ (E_{n}(p),\chi_{n}(z;p)) &\text{for } p \in U \text{ and } p \geq p^* \end{cases} 	
\end{split}
\end{equation}
are smooth for all $p \in U$. 
\item The bands $E_+, E_-$ satisfy $\de_p E_+(p^*) > 0, \text{ } \de_p E_-(p^*) < 0$ and in particular:
\begin{equation} \label{eq:linear_splitting}
	\de_p E_+(p^*) - \de_p E_-(p^*) = 2 \de_p E_+(p^*) > 0. 	
\end{equation} 
\end{enumerate} 
\end{property}

\noindent {\bf Caveat Lector!} {\it In \eqref{eq:def_smooth_bands}, the notation $+$ and $-$ refers to the slope of the smooth band functions at the crossing point: $\de_p E_+(p^*) > 0, \de_p E_-(p^*) < 0$. This is not to be confused with an ordering of the bands themselves. Indeed, with our conventions we have: }
\begin{equation}
	\text{\it for $p \in U$ and $p < p^*$:  } E_+(p) = E_n(p) < E_{n+1}(p) = E_-(p).
\end{equation}
It is useful to view the functions $E_+(p), E_-(p)$ as \emph{smooth continuations} of the band functions $E_n(p), E_{n+1}(p)$ from the interval $\{ p \in U : p\le p^*\}$ to the interval $\{ p \in U : p > p^*\}$. We will refer to any band crossing satisfying Property \ref{band_crossing_assumption} as a \emph{linear crossing}. In one spatial dimension, \emph{all} band crossings are linear. Moreover, crossings can only occur at 0 or $\pi$ (modulo $2 \pi$):
\begin{theorem} \label{th:all_crossings_smooth}
Let $E_n(p), E_{n+1}(p)$ denote spectral band functions satisfying \eqref{eq:reduced_eigenvalue_problem} for $p \in \mathcal{B}$, and let $p^* \in \mathcal{B}$ be such that $E_n(p^*) = E_{n+1}(p^*)$. Then: 
\begin{enumerate}
\item $p^* = 0$ or $\pi$ (modulo $2 \pi$). 
\item There exists an open interval $U$ containing $p^*$ such that hypotheses (A1)-(A4) of Property \ref{band_crossing_assumption} hold.
\end{enumerate}
\end{theorem}
The proof of Theorem \ref{th:all_crossings_smooth} follows from the fact that the eigenvalue problem \eqref{eq:quasi_periodic_eigenvalue_problem} has the form of an ordinary differential equation with real and periodic coefficients and may therefore be put into the form of Hill's equation, for which a well-developed theory exists \cite{MagnusWinkler,eastham,reed_simon_4}. For a self-contained proof, see Appendix B.1 of \cite{thesis}. 
\begin{corollary} \label{cor:away_from_crossing_regular}
Let $E_n(p), E_{n+1}(p)$ denote spectral band functions in one dimension which cross at some $p^* \in \mathcal{B}$. Let $P^\perp_\pm(p)$ denote the projection onto the orthogonal complement in $L^2_{per}$ of the functions $\chi_+(z;p), \chi_-(z;p)$, defined for $p \in U$ by \eqref{eq:def_smooth_bands}. Then: 
\begin{equation} \label{eq:corollary_assertion}
	\| ( H(p) - E_\sigma(p) )^{-1} P^\perp_\pm(p) \|_{L_{per}^2 \rightarrow H_{per}^2} \leq \frac{1}{M} , \qquad \sigma = \pm, \; p \in \overline{U}.
\end{equation}
where $M > 0$ is the constant appearing in \eqref{eq:crossing_bands_isolated_from_others}. 
\end{corollary}
The bound \eqref{eq:corollary_assertion} follows immediately from \eqref{eq:crossing_bands_isolated_from_others}. When we consider the dynamics of wavepackets associated with $E_n(p)$ or $E_{n+1}(p)$ and spectrally localized close to $p^*$, the gap condition \eqref{eq:crossing_bands_isolated_from_others} and Corollary \ref{cor:away_from_crossing_regular} will allow us to bound contributions to the solution from all bands other than $E_n(p)$ and $E_{n+1}(p)$ uniformly through the crossing time, see Appendix D of \cite{2017WatsonWeinsteinLu} for details. 
\begin{remark}\label{Dirac-Points}
Theorem \ref{th:all_crossings_smooth} does not generalize to spatial dimensions larger than one. Indeed, at so-called `conical' or `Dirac' points, which occur in the spectral band structure of two-dimensional periodic Schr\"{o}dinger operators with \emph{honeycomb lattice} symmetry, the local band structure is the union of Lipschitz surfaces \cite{fefferman_weinstein_diracpoints,fefferman_weinstein} and the map $p \mapsto \chi_n(z;p)$ from the Brillouin zone to the Bloch eigenfunctions is discontinuous \cite{fefferman_weinstein}.
\end{remark}

\subsection{Examples of potentials with linear band crossings}\label{examples}

\begin{example}[$m-$ gap potentials] \label{ex:n_gap_potentials}
Let $\omega_1, \omega_3 \in \field{C}$ with $\text{\emph{Im }}(\omega_3/\omega_1) \neq 0$. Define $\wp_{\omega_1, \omega_3}(z)$, the \emph{Weierstrass elliptic function with periods $2 \omega_1, 2 \omega_3$} by:
\begin{equation}
	\wp_{\omega_1,\omega_3}(z) := \frac{1}{z^2} + \sum_{\substack{(m,n) \in \field{Z} \times \field{Z}, \\ (m,n) \neq (0,0)}} \frac{1}{(z - 2 m \omega_1 - 2 n \omega_3)^2} - \frac{1}{(2 m \omega_1 + 2 n \omega_3)^2}.
\end{equation}
The function $\wp_{\omega_1,\omega_3}(z)$ is doubly-periodic and even:
\begin{equation} \label{eq:p_symmetries}
\begin{split}
	&\wp_{\omega_1,\omega_3}(z + 2 \omega_1) = \wp_{\omega_1,\omega_3}(z + 2 \omega_3) = \wp_{\omega_1,\omega_3}(z)	\\
	&\wp_{\omega_1,\omega_3}(-z) = \wp_{\omega_1,\omega_3}(z),
\end{split}
\end{equation}
and has poles of degree two at the points $\Omega_{m,n} = 2 m \omega_1 + 2 n \omega_3$ for all $(m,n) \in \field{Z} \times \field{Z}$. If $\omega_1 = \omega$, $\omega_3 = i \omega'$ with $\omega, \omega' \in \field{R}$ and $\omega > 0$, then $\wp_{\omega,i \omega'}(z)$ is real for $z$ such that $\text{\emph{Re }}z \in \{0,\omega\}$ or $\text{\emph{Im }}z \in \{0,\omega'\}$ by the symmetries \eqref{eq:p_symmetries}. Now fix $\omega = 1/2$, and define for any $\omega' \in \field{R}$ with $\omega' \neq 0$ and positive integer $m$: 
\begin{equation}
	V(z) := \frac{ m (m+1) }{2} \wp_{1/2,i \omega'}(z + i \omega').	
\end{equation}
Then for $z \in \field{R}$, $V(z)$ is a real, smooth, 1-periodic function. 

The $m$ lowest Bloch band dispersion functions defined by \eqref{eq:reduced_eigenvalue_problem} for this potential are non-degenerate for all $p \in \mathcal{B}$, but for \emph{every} $n > m$, the band $E_n(p)$ has a linear crossing with the band $E_{n+1}(p)$ at $p = 0$ or $p = \pi$ \cite{MagnusWinkler}. Such potentials are known as `$m$-gap' potentials since the $L^2(\field{R})$ spectrum of the operator $- \frac{1}{2} \de_z^2 + V(z)$ in this case consists of $m + 1$ real intervals with $m$ `gaps' between them. Indeed, \emph{all} `$m$-gap' potentials, for positive integers $m$, are elliptic functions \cite{Hochstadt}. Any Weierstrass elliptic function may be written in terms of Jacobi elliptic functions; for more detail see \cite{2005Cai,WhittakerWatson,Chandrasekharan,Akhiezer,PoschelTrubowitz}. 

The lowest three bands of a `one-gap' potential are shown in Figure \ref{fig:crossing_bands}. The smooth bands at the linear crossing between the second and third bands defined by \eqref{eq:def_smooth_bands}, whose existence is ensured by Theorem \ref{th:all_crossings_smooth}, are shown in Figure \ref{fig:crossing_bands_smooth}. 
\end{example}
\begin{example}[Trivial band crossings] \label{ex:double_period_potentials}
Every Bloch band of a $1/2-$ periodic function regarded as a  $1$-periodic function is degenerate. To see this, let $V(z)$ be $1/2$-periodic. We may plot the band structure of the operator $-\frac{1}{2} \de_z^2 + V(z)$ with respect to the natural $4 \pi$-periodic Brillouin zone, which we take for concreteness to be $[0,4\pi]$. We may also treat $V(z)$ as a 1-periodic potential and plot its band structure with respect to the $2 \pi$-periodic Brillouin zone $[0,2\pi]$. But it is clear that any eigenpair of the $1/2$-periodic eigenvalue problem will also be an eigenpair of the $1$-periodic eigenvalue problem. Hence the band structure of the 1-periodic operator is nothing but the band structure of the 1/2-periodic operator `folded over' onto the shorter interval. More precisely, eigenvalues of the 1/2-periodic operator with quasi-momentum $p \in [2\pi,4\pi]$ will be eigenvalues of the 1-periodic operator with quasi-momentum $p - 2\pi$. For an example, see Figure \ref{fig:folded_bands}. The converse of this argument also holds: whenever all `odd-numbered' gaps close, then $V(z)$ is $1/2$-periodic. Analogous statements hold for $1/2^n$-periodic potentials for any natural number $n$ \cite{1946Borg,1966Hochstadt,reed_simon_4}. We will refer to such crossings as `trivial', since they may be removed by a proper choice of Brillouin zone. The amplitude of the  `excited' wave due to such crossings is zero; see Remark \ref{rem:remark_on_trivial_crossings} and Appendix \ref{app:coupling_coefficient_zero}. 
\end{example}
\begin{figure}
\begin{subfigure}{.5\textwidth} 
\includegraphics[scale=.38]{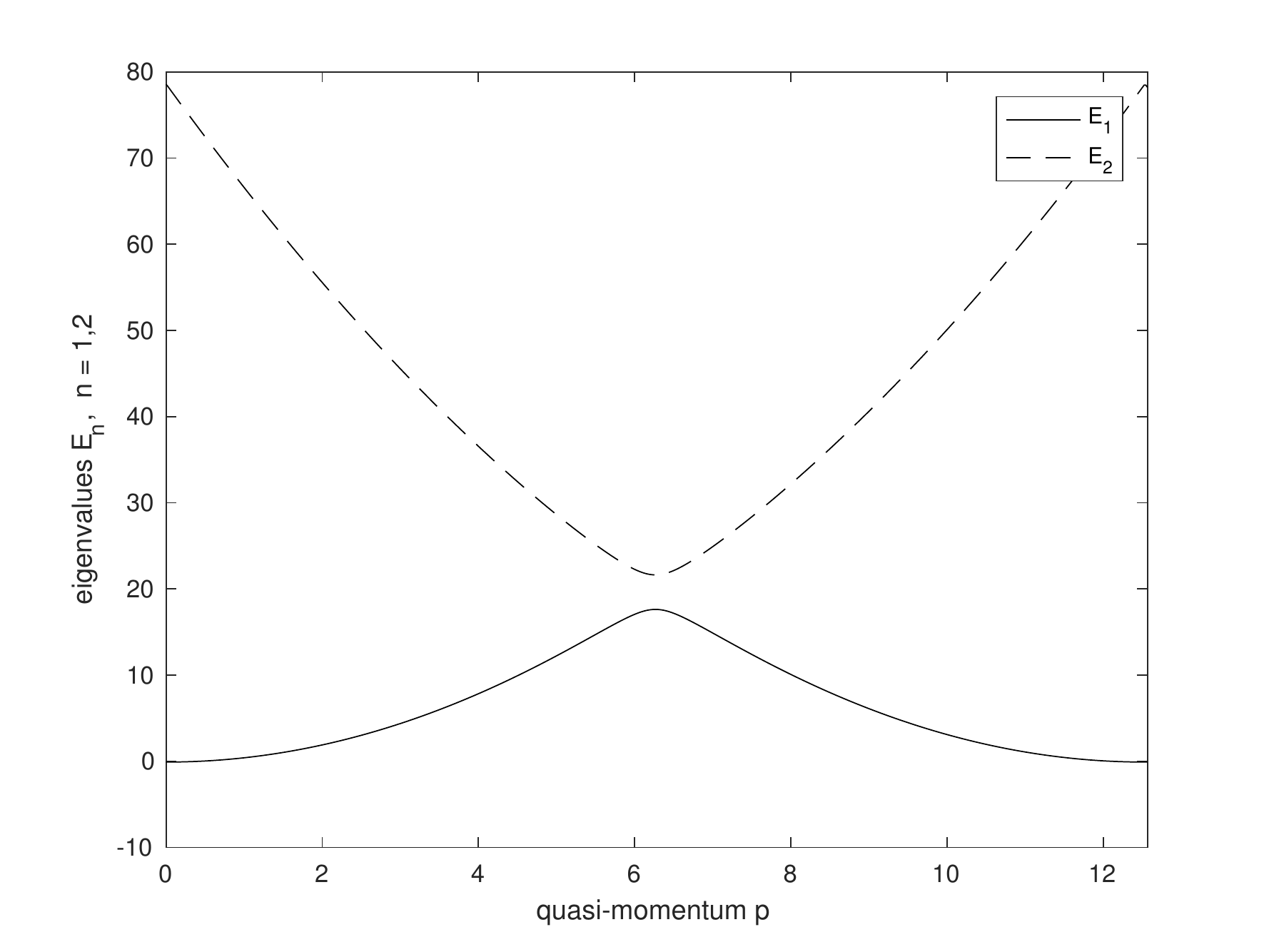} 
\subcaption{}
\end{subfigure}%
\begin{subfigure}{.5\textwidth}
\includegraphics[scale=.38]{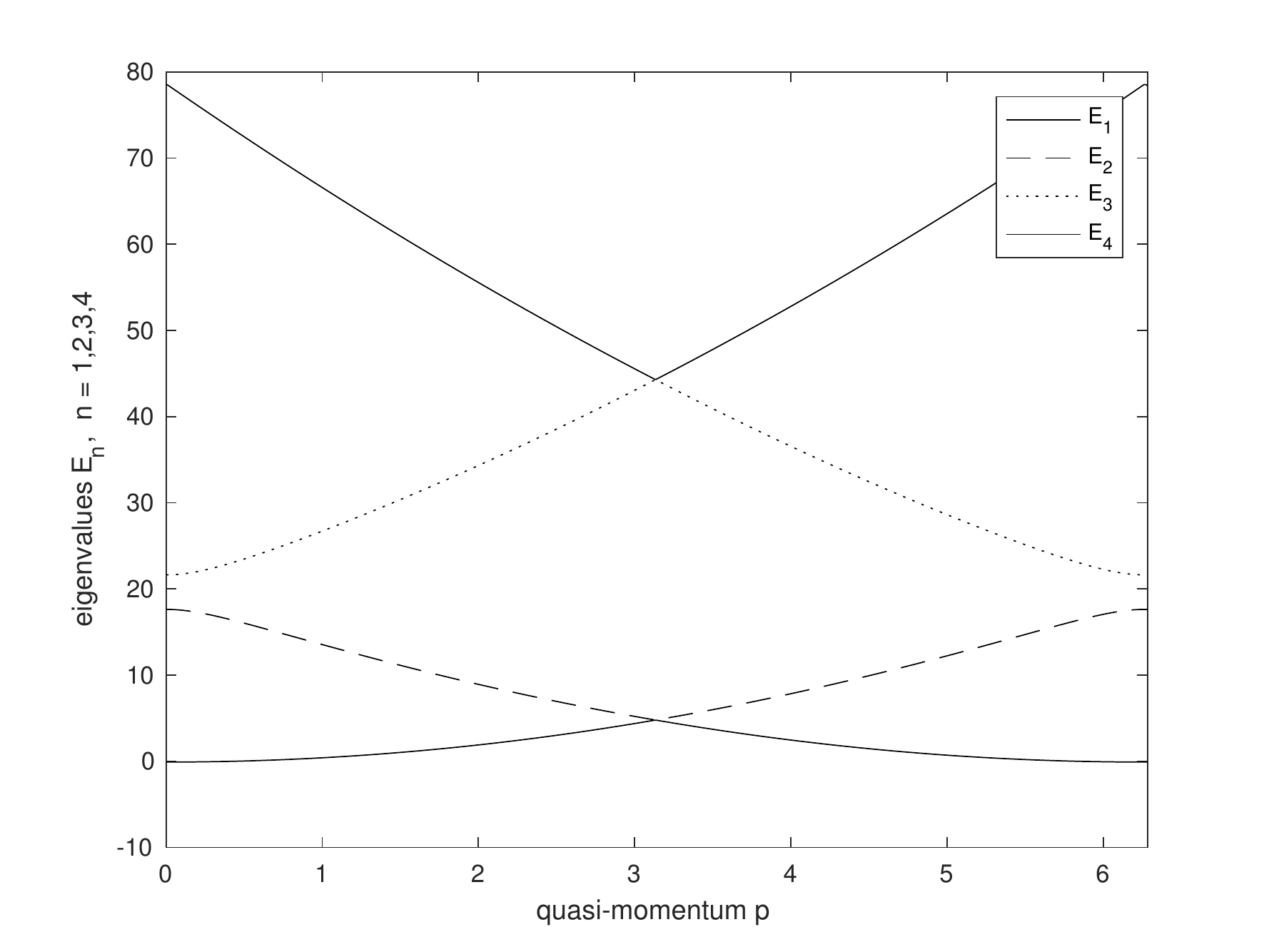} 
\subcaption{}
\end{subfigure}
\caption{Lowest Bloch bands when $V(z) = 4 \cos(4 \pi z)$, viewed as a 1/2-periodic potential and plotted over the natural Brillouin zone in this case $[0,4\pi]$ (\textsc{a}) and viewed as a 1-periodic potential and plotted over $[0,2\pi]$ (\textsc{b}). When $V(z)$ is viewed as a 1-periodic potential, \emph{every} Bloch band is degenerate at $p = \pi$.}
\label{fig:folded_bands}
\end{figure}

\subsection{Band crossing dynamics}\label{crossing-dynamics}

We now make precise the scenario of a wavepacket whose quasi-momentum is driven by the external potential $W$ towards a quasi-momentum $p^*\in\mathcal{B}$ at which there is a linear band crossing; see Property \ref{band_crossing_assumption}. 
\begin{property}[Band Crossing Scenario] \label{classical_path_assumption}
Let $E_n, E_{n+1}$ denote spectral band functions associated with the eigenvalue problem (\ref{eq:reduced_eigenvalue_problem}) for $p \in \mathcal{B}$ which have a linear crossing in the sense of Property \ref{band_crossing_assumption} at $p^*$. Let $q_0, p_0 \in \field{R} \times \mathcal{B}$ be such that $G(E_n(p_0)) > 0$ (i.e. the band $E_n(p)$ is \emph{isolated} at $p_0$: recall the definition of the spectral gap function $G$ \eqref{eq:nth_spectral_gap_function}). We assume the existence of a positive constant $t^* > 0$ such that the equations of motion of the classical Hamiltonian $\mathcal{H}_n(q,p) := E_n(p) + W(q)$:
\begin{align} \label{eq:n_band_classical_system}
	&\dot{q}(t) = \de_p E_n(p(t)) &&\dot{p}(t) = - \de_q W(q(t)) 	\\
\nonumber	&q(0) = q_0 &&p(0) = p_0
\end{align}
have a unique smooth solution $(q(t),p(t)) \subset \field{R} \times \mathcal{B}$ for all $t \in [0,t^*)$ such that the Bloch band function $E_n$ is isolated when evaluated at $p(t)$ for every $t \in [0,t^*)$:  
\begin{equation} \label{eq:first_crossing_at_t_star}
	\text{for all } t \in [0,t^*), \text{ } G(E_n(p(t))) > 0 \text{ and } \lim_{t \uparrow t^*} p(t) = p^*.
\end{equation}
Let $q^*$ denote the limit: $\lim_{t \uparrow t^*} q(t)$. We assume that the wavepacket is `driven' towards the crossing in the following sense: 
\begin{equation} \label{eq:driven_through_pstar}
	\lim_{t \uparrow t^*} \dot{p}(t) = - \de_q W(q^*) > 0.
\end{equation}
\end{property}
\begin{remark}
We choose the sign of $- \de_q W(q^*)$ in \eqref{eq:driven_through_pstar} to be positive without loss of generality. Note that it follows from \eqref{eq:driven_through_pstar} that for $t < t^*$ with $|t - t^*|$ sufficiently small, $p(t) < p^*$: i.e. the wave-packet quasi-momentum approaches $p^*$ `from the left'. As a consequence the `smooth extension' of the map $t \mapsto E_n(p(t))$ for $t \geq t^*$ makes use of $E_+(p(t))$ rather than $E_-(p(t))$; see Proposition \ref{prop:smooth_continuation}. 
\end{remark}
We aim to describe the solution of the PDE \eqref{eq:original_equation} with `Bloch wavepacket' initial data of the form: 
\begin{equation} 
	\psi^\epsilon(x,0) = \text{WP}^{1,\epsilon}[S_0,q_0,p_0,a^0_0(y),a^1_0(y),\chi_n(z;p_0)](x)\ ;
\end{equation}
in the Band Crossing Scenario (Property  \ref{classical_path_assumption}) up to errors of $o_{L^2}(\sqrt{\epsilon})$ for $t$ up to and greater than the crossing time $t^*$. Here, $a^0_0, a^1_0 \in \mathcal{S}(\field{R})$, $S_0 \in \field{R}$, and the $\text{WP}^{1,\epsilon}[...]$ notation is as in \eqref{eq:WP_notation_1}. 

Note that for $t < t^*$, \eqref{eq:first_crossing_at_t_star} implies that Property \ref{isolated_band_assumption} holds with $t_0 = 0$ and $t_1 = t$. By Theorem \ref{th:isolated_band_theorem_2} the solution $\psi^\epsilon(x,t)$ of \eqref{eq:original_equation} satisfies, for fixed $t$ and $\epsilon\downarrow0$:
\begin{equation} \label{eq:incoming_single_band}
	\psi^\epsilon(x,t) = \text{WP}^{1,\epsilon}[S(t),q(t),p(t),a^0(y,t),a^1(y,t),\chi_n(z;p(t))](x,t) + O_{L^2}(\epsilon)
\end{equation}
where $q(t),p(t),S(t),a^0(y,t),a^1(y,t)$ are as in \eqref{eq:n_band_classical_system}, \eqref{eq:action_integral}, \eqref{eq:envelope_equation}, and \eqref{eq:first_order_envelope_equation} respectively. \medskip

Two difficulties arise in estimating the error term in \eqref{eq:incoming_single_band} for $t \geq t^*$: 
\begin{difficulty} \label{dif:1}
The functions $q(t)$, $p(t)$, $S(t)$, $a^0(y,t)$, $a^1(y,t)$, $\chi_n(z;p(t))$, and $\de_p \chi_n(z;p(t))$, and therefore the function: 
\begin{equation}
	\text{\emph{WP}}^{1,\epsilon}[S(t),q(t),p(t),a^0(y,t),a^1(y,t),\chi_n(z;p(t))](x,t),
\end{equation}
are not well-defined at $t = t^*$ since the band function $E_n(p)$ and its associated eigenfunctions $\chi_n(z;p)$ are not smooth in $p$ at $p^*$. 
\end{difficulty}
\begin{difficulty} \label{dif:2}
The $L^2_x$-norm of the error in the approximation \eqref{eq:incoming_single_band} depends directly on the inverse of the spectral gap function $G(E_n(p(t)))$, which blows up as $t \uparrow t^*$ since $|G(E_n(p(t)))| \sim | E_{n+1}(p(t)) - E_n(p(t)) | \downarrow 0$. 
\end{difficulty}
\noindent We return to Difficulty \ref{dif:2} below: see Theorem \ref{th:limit_of_single_band_approximation} and Corollary \ref{cor:limit_of_single_band_approximation}.

\subsection{Resolution of Difficulty \ref{dif:1}; Smooth Continuation of Bands}
Difficulty \ref{dif:1} may be overcome by making proper use of the smooth band functions $E_+, E_-$; see  \eqref{eq:def_smooth_bands} in Property \ref{band_crossing_assumption}, Theorem \ref{th:all_crossings_smooth} and Figure \ref{fig:crossing_bands_smooth}. The following proposition shows how in the Band Crossing Scenario (Property \ref{classical_path_assumption}), we may extend the map $[0,t^*) \rightarrow \field{R} \times \mathcal{B}, t \mapsto (q(t),p(t))$ to a smooth map over an interval $[0,T]$ with $T > t^*$ using the smooth band function $E_+$:
\begin{proposition} \label{prop:smooth_continuation}
Assume the Band Crossing Scenario (Property \ref{classical_path_assumption}) with crossing occurring for $t=t^*$. Then for sufficiently small positive $\delta$ with $0 < \delta < t^*$, the equations of motion of the classical Hamiltonian $\mathcal{H}_+({q},{p}) := E_+({p}) + W({q})$ with data specified at $t^*$:
\begin{align} \label{eq:+_band_classical_system}
	&\dot{{q}}_+(t) = \de_{{p}} E_+({p}_+(t)), &&\dot{{p}}_+(t) = - \de_{{q}} W({q}_+(t)) 	\\
	\nonumber&{q}_+(t^*) = q^* &&{p}_+(t^*) = p^* 	
\end{align}
have a unique smooth solution $({q}_+(t),{p}_+(t)) \subset \field{R} \times U$ over the interval $t \in [t^* - \delta,t^* + \delta]$ which satisfies: 
\begin{equation} \label{eq:smooth_continuation}
	\text{for all } t \in [t^* - \delta,t^*), \text{ } q(t) = {q}_+(t) \text{, } p(t) = {p}_+(t).
\end{equation}
Furthermore, for sufficiently small $T \geq t^* + \delta > 0$, there exists a solution $(q_{n+1}(t),p_{n+1}(t))$ $\subset \field{R} \times \mathcal{B}$ of the equations of motion of the classical Hamiltonian $\mathcal{H}_{n+1}(q,p) := E_{n+1}(p) + W(q)$ over the interval $t \in (t^*,T]$ satisfying the limits:
\begin{align} \label{eq:n+1_system}
	&\dot{q}_{n+1}(t) = \de_{p} E_{n+1}(p_{n+1}(t)) &&\dot{p}_{n+1}(t) = - \de_{q} W(q_{n+1}(t)) 	\\
\nonumber	&\lim_{t \downarrow t^*} q_{n+1}(t) = q^* &&\lim_{t \downarrow t^*} p_{n+1}(t) = p^*
\end{align}
such that $G(E_{n+1}(p_{n+1}(t))) > 0$ for all $t \in (t^*,T]$. This solution satisfies: 
\begin{equation} \label{eq:smooth_continuation_again1}
	\text{for all } t \in (t^*,t^* + \delta], \text{ } q_+(t) = {q}_{n+1}(t) \text{, } p_+(t) = {p}_{n+1}(t).
\end{equation}
It follows from \eqref{eq:smooth_continuation} and \eqref{eq:smooth_continuation_again1} that the map:
\begin{equation} \label{eq:extended_map}
	t \mapsto (\mathfrak{q}_+(t),\mathfrak{p}_+(t)) := \begin{cases} (q(t),p(t)) &\text{for } t \in [0,t^* - \delta] \\ ({q}_+(t),{p}_+(t)) &\text{for } t \in [t^* - \delta,t^*+\delta] \\ (q_{n+1}(t),p_{n+1}(t)) &\text{for } t \in [t^*+\delta,T] \end{cases}	
\end{equation}
is smooth as a map $[0,T] \rightarrow \field{R} \times \mathcal{B}$. 
\end{proposition}
\begin{proof}
The potential $W$ is smooth by assumption, the band functions $E_n, E_{n+1}$ are smooth everywhere away from $p^*$, and the band function $E_+$ is smooth in $U$, a neighborhood of $p^*$. The proposition then follows easily from existence and uniqueness for solutions of ODEs with smooth coefficients.
\end{proof}
Corresponding to the smooth extension $t \mapsto (\mathfrak{q}_+(t),\mathfrak{p}_+(t))$ we may define the smooth extension of $\chi_n(z;p(t))$ through the crossing:
\begin{equation} \label{eq:X_+}
	t \mapsto \mathfrak{X}_+(z;\mathfrak{p}_+(t)) := \begin{cases} \chi_n(z;p(t)) &\text{for } t \in [0,t^* - \delta] \\ \chi_+(z;p_+(t)) &\text{for } t \in [t^* - \delta,t^*+\delta] \\ \chi_{n+1}(z;p_{n+1}(t)) &\text{for } t \in [t^*+\delta,T] \end{cases}.
\end{equation}
Finally, using the smooth maps $t \mapsto (\mathfrak{q}_+(t),\mathfrak{p}_+(t)) $ and $t \mapsto \mathfrak{X}_+(z;\mathfrak{p}_+(t)) $, we introduce smooth extensions of the functions $a^0(y,t)$, $a^1(y,t)$, and $S(t)$ over the whole interval $t \in [0,T]$ as follows:
\begin{definition}[Smooth extensions of $a^0(y,t)$, $a^1(y,t)$, and $S(t)$] \label{def:S_a_b_extensions}
Let: 
\begin{equation} \label{eq:defs_of_stars}
	S^* := \lim_{t \uparrow t^*} S(t), \text{ } a^{0,*}(y) := \lim_{t \uparrow t^*} a^0(y,t), \text{ \emph{and} } a^{1,*}(y) := \lim_{t \uparrow t^*} a^1(y,t). 
\end{equation}
Then let $S_+(t)$, $a^0_+(y,t)$, and $a^1_+(y,t)$ be defined for $t \in [t^* - \delta,t^* + \delta]$ by \eqref{eq:action_integral}, \eqref{eq:envelope_equation}, and \eqref{eq:first_order_envelope_equation} with $t_0 = t^*$, and where all dependence on $p(t)$, $q(t)$, $E_n(p(t))$, $\chi_n(z;p(t))$, and $W(q(t))$ is replaced by dependence on $p_+(t), q_+(t)$, $E_+(p_+(t))$, $\chi_+(z;p_+(t))$, and $W(q_+(t))$ respectively, and:
\begin{equation}
	 S_+(t^*) = S^*, \text{ } a_+^0(y,t^*) = a^{0,*}(y), \text{ \emph{and} } a_+^1(y,t^*) = a^{1,*}(y).
\end{equation}
Then let $S_{n+1}(t)$, $a^0_{n+1}(y,t)$, $a^1_{n+1}(y,t)$ be defined for $t \in (t^*,T]$ by equations \eqref{eq:action_integral}, \eqref{eq:envelope_equation}, and \eqref{eq:first_order_envelope_equation}, replacing dependence on $p(t)$, $q(t)$, $E_n(p(t))$, $\chi_n(z;p(t))$, and $W(q(t))$ by dependence on $p_{n+1}(t), q_{n+1}(t)$, $E_{n+1}(p_{n+1}(t))$, $\chi_{n+1}(z;p_{n+1}(t))$, and $W(q_{n+1}(t))$ and the limits:
\begin{equation}
	\lim_{t \downarrow t^*} S_{n+1}(t) = S^*, \text{ } \lim_{t \downarrow t^*} a_{n+1}(y,t) = a^{0,*}(y), \text{ \emph{and} } \lim_{t \downarrow t^*} a^1(y,t) = a^{1,*}(y).
\end{equation}
We denote by $\mathfrak{S}_+(t)$, $\mathfrak{a}^0_+(y,t)$, and $\mathfrak{a}^1_+(y,t)$ smooth maps defined over the whole interval $t \in [0,T]$ defined analogously to \eqref{eq:extended_map} so that, for example:
\begin{equation} \label{eq:a_+_extension}
	t \mapsto \mathfrak{a}^0_+(y,t) := \begin{cases} a^0(y,t) &\text{for } t \in [0,t^* - \delta] \\ a^0_+(y,t) &\text{for } t \in [t^* - \delta,t^*+\delta] \\ a^0_{n+1}(y,t) &\text{for } t \in [t^*+\delta,T] \end{cases}.
\end{equation}
\end{definition}

We now define a \emph{first-order wavepacket smoothly continued through the crossing}, an expression which is smooth for all $t \in [0,T]$ by (recall the definition of $\text{WP}^{1,\epsilon}$ \eqref{eq:WP_notation_1}):
\begin{equation} \label{eq:WP_1_+}
\begin{split}
	&\text{WP}^{1,\epsilon}[\mathfrak{S}_+(t),\mathfrak{q}_+(t),\mathfrak{p}_+(t),\mathfrak{a}^0_+(y,t),\mathfrak{a}^1_+(y,t),\mathfrak{X}_+(z;\mathfrak{p}_+(t))](x,t)	:= 	\\
	&:= \begin{cases} \text{{WP}}^{1,\epsilon}[S(t),q(t),p(t),a^0(y,t),a^1(y,t),\chi_n(z;p(t))](x,t) &\text{for } t \in [0,t^* - \delta] \\ \text{{WP}}^{1,\epsilon}[S_+(t),q_+(t),p_+(t),a^0_+(y,t),a^1_+(y,t),\chi_+(z;p(t))](x,t) &\text{for } t \in [t^* - \delta,t^*+\delta] \\ \text{{WP}}^{1,\epsilon}[S_{n+1}(t),q_{n+1}(t),p_{n+1}(t),a^0_{n+1}(y,t),a^1_{n+1}(y,t),\chi_{n+1}(z;p(t))](x,t) &\text{for } t \in [t^* + \delta,T] \end{cases}.
\end{split}
\end{equation}
\begin{remark}
By construction, \eqref{eq:WP_1_+} is a wavepacket associated with the band $E_n$ for $t \in [0,t^* - \delta]$, a wavepacket associated with the band $E_{n+1}$ for $t \in [t^* + \delta,T]$, and a wavepacket associated with the `smooth transition' $E_+$ for $t \in [t^* - \delta,t^* + \delta]$. 
\end{remark}
\subsection{Resolution of Difficulty \ref{dif:2}; Incorporation of the second band and a new, fast / non-adiabatic, time-scale}
We first present a result which quantifies, through a blow-up rate of the error bound, the breakdown of the single band approximation \eqref{eq:incoming_single_band} as $t \uparrow t^*$: 
\begin{theorem} \label{th:limit_of_single_band_approximation}
 Assume the Band Crossing Scenario (Property \ref{classical_path_assumption}). Assume $a^0_0(y)$ and $ a^1_0(y) \in \mathcal{S}(\field{R})$ and let $\psi^\epsilon(x,t)$ denote the unique solution of (\ref{eq:original_equation}) with `Bloch wavepacket' initial data:
\begin{equation} \label{eq:bloch_wavepacket_initial_data} 
	\psi^\epsilon(x,0) = \text{\emph{WP}}^{1,\epsilon}_n[S_0,q_0,p_0,a^0_0(y),a^1_0(y),\chi_n(z;p_0)](x). 
\end{equation}

Then for $t \in [0,t^*)$, $\psi^\epsilon(x,t)$ satisfies: 
\begin{equation} \label{eq:WP_and_corrector}
	\psi^\epsilon(x,t) = \text{\emph{WP}}^{1,\epsilon}[\mathfrak{S}_+(t),\mathfrak{q}_+(t),\mathfrak{p}_+(t),\mathfrak{a}^0_+(y,t),\mathfrak{a}^1_+(y,t),\mathfrak{X}_+(z;\mathfrak{p}_+(t))](x) + \eta^{\epsilon}(x,t)
\end{equation} 
where $\text{\emph{WP}}^{1,\epsilon}_+(x,t)$ is given by (\ref{eq:WP_1_+}). Moreover, the  corrector $\eta^{\epsilon}(x,t)$ satisfies the following bound for $0<t<t^*$, which blows up as $t\uparrow t^*$: 
\begin{equation} \label{eq:bound_on_eta_near_the_crossing}
\begin{split}
	\| \eta^\epsilon(\cdot,t) \|_{L^2} \leq &\text{  }C \left| \ip{\chi_-(\cdot;p^*)}{\de_p \chi_+(\cdot;p^*)} \right| \left( \frac{\epsilon}{|t - t^*|} + \frac{ \epsilon^{3/2} }{ |t - t^*|^2 } \right) \\ 
	&+ O\left(\epsilon, \epsilon^{3/2} \ln |t - t^*| , \frac{\epsilon^{3/2}}{|t - t^*|}\right) \\
\end{split}
\end{equation}
The constants in \eqref{eq:bound_on_eta_near_the_crossing} (explicit and implied) are
 independent of $t, \epsilon$ and are finite as long as:  
 $\partial_p E_+(p^*)-\partial_p E_-(p^*)=2\partial_p E_+(p^*)>0$ and $\partial_qW(q^*)\ne0$.
\end{theorem}
The proof of Theorem \ref{th:limit_of_single_band_approximation} is by explicit term by term estimation of the corrector $\eta^\epsilon(x,t)$. We demonstrate this procedure in Section \ref{sec:proof_of_single_band_failure_theorem} for the term which contributes the first term on the right-hand side of the estimate \eqref{eq:bound_on_eta_near_the_crossing}. This term dominates when $t = t^* - \epsilon^\xi$ where $0 < \xi < 1/2$ and hence controls the time interval of validity of the single-band ansatz (see Corollary \ref{cor:limit_of_single_band_approximation}). All other terms in $\eta^\epsilon$ may be bounded similarly, see Appendix B.3 of \cite{thesis}. 
\begin{remark}
The manner in which the nonzero constants  $\partial_p E_+(p^*)-\partial_p E_-(p^*)=2\partial_p E_+(p^*)$ and $\partial_qW(q^*)$ play a role in the bound \eqref{eq:bound_on_eta_near_the_crossing} is seen in \eqref{eq:E_plus_minus_expanded}.
\end{remark}
Theorem \ref{th:limit_of_single_band_approximation} shows that the single band ansatz, even when smoothly continued through the linear band crossing, fails to give a good approximation (error of size $o_{L^2_x}(\sqrt{\epsilon})$) to the solution $\psi^\epsilon(x,t)$ of equation \eqref{eq:original_equation} for small $|t - t^*|$. Furthermore,  since the dominant terms in the bound \eqref{eq:bound_on_eta_near_the_crossing}, $\sqrt{\epsilon} \times (\sqrt{\epsilon} / |t-t^*|), \sqrt{\epsilon} \times (\sqrt{\epsilon} / |t - t^*| )^2$,  are proportional to 
\begin{equation} \label{eq:coupling_coefficient}
	\ip{\chi_-(\cdot;p^*)}{\de_p \chi_+(\cdot;p^*)},
\end{equation}
we see that the failure of the single band wave-packet approximation is due to contributions to the solution from the other band participating in the linear crossing, $p\mapsto E_-(p)$,  growing to be of size $\sim \sqrt{\epsilon}$ when:
\begin{equation}
	|t - t^*| \sim \sqrt{\epsilon}. 
\end{equation}

\begin{remark} \label{rem:remark_on_trivial_crossings}
At `trivial' crossings, which occur when the potential $V(z)$ has minimal period $1/2$ (recall Example \ref{ex:double_period_potentials}), the `inter-band coupling coefficient' \eqref{eq:coupling_coefficient} is zero (see Appendix \ref{app:coupling_coefficient_zero}). It follows that the amplitude of the wave associated with the other band involved in the crossing `excited' (Theorem \ref{th:band_crossing_theorem}) at the crossing is also zero. This is consistent with the observation that the crossing may be removed by making the proper choice of Brillouin zone. 
\end{remark}

\begin{remark} \label{rem:remark_on_nontrivial_crossings}
When $V(z)$ is an $m-$gap potential (recall Example \ref{ex:n_gap_potentials}), the eigenfunctions $\chi_\pm(z;p)$ may be explicitly displayed (see \cite{2005Cai}, for example), and the `coupling coefficient' \eqref{eq:coupling_coefficient} may be numerically computed relatively easily. Such a computation for the lowest band crossing of the `one-gap' potential (shown in Figures \ref{fig:crossing_bands} and \ref{fig:crossing_bands_smooth}) shows that \eqref{eq:coupling_coefficient} is non-zero in this case.
\end{remark} 

The following Corollary of Theorem \ref{th:limit_of_single_band_approximation} precisely characterizes the time interval of validity of the single band ansatz: 
\begin{corollary} \label{cor:limit_of_single_band_approximation}
Let $t = t^* - \epsilon^\xi$. Then, for small enough $\epsilon > 0$, (\ref{eq:bound_on_eta_near_the_crossing}) implies that the corrector function $\eta^\epsilon(x,t)$ which appears in \eqref{eq:WP_and_corrector} satisfies:
\begin{equation} \label{eq:t_eps_xi}
	\sup_{t \in [0,t^* - \epsilon^{\xi}]} \| \eta^\epsilon(\cdot,t) \|_{L^2} \leq C \epsilon^{1 - \xi} 
\end{equation}
where $C > 0$ is a constant independent of $\epsilon, \xi, t$. In particular, if $0 < \xi < 1/2$, then:
\begin{equation} \label{eq:t_eps_xi_2}
	\sup_{t \in [0,t^* - \epsilon^{\xi}]} \| \eta^\epsilon(\cdot,t) \|_{L^2} = o(\sqrt{\epsilon}).
\end{equation}
It follows that $\eta^\epsilon(x,t)$ is negligible in $L^2(\field{R})$ compared with $\text{\emph{WP}}^{1,\epsilon}$ in the expansion \eqref{eq:WP_and_corrector} for $t \in [0,t^* - \epsilon^{\xi}]$. 
\end{corollary}
In order to describe the solution for $t \sim t^*$ and $t \geq t^*$, it is necessary to make a more general ansatz for the solution which accounts for the excitation of a wave associated with the other band involved in the crossing over the time-scale:
\begin{equation} \label{eq:nonadiabatic_timescale}
	s := \frac{t - t^*}{\sqrt{\epsilon}}\ =\ \mathcal{O}(1).
\end{equation}

The following proposition, which is analogous to Proposition \ref{prop:smooth_continuation}, is required to construct this excited wave:
\begin{proposition} \label{prop:smooth_continuation_again}
Assume the Band Crossing Scenario (Property \ref{classical_path_assumption}). Then for sufficiently small positive $\delta'$ the equations of motion of the classical Hamiltonian $\mathcal{H}_-({q},{p}) := E_-({p}) + W({q})$ with data specified at $t^*$:
\begin{align} \label{eq:-_band_classical_system}
	&\dot{{q}}_-(t) = \de_{{p}} E_-({p}_-(t)) &&\dot{{p}}_-(t) = - \de_{{q}} W({q}_-(t)) 	\\
	&{q}_-(t^*) = q^* &&{p}_-(t^*) = p^*
\end{align}
have a unique smooth solution $({q}_-(t),{p}_-(t)) \subset \field{R} \times U$ over the interval $t \in [t^* - \delta',t^* + \delta']$. Furthermore, for sufficiently small $T' \geq t^* + \delta' > 0$, there exists a solution $(q_{n}(t),p_{n}(t)) \subset \field{R} \times \mathcal{B}$ of the equations of motion of the classical Hamiltonian $\mathcal{H}_{n}(q,p) := E_{n}(p) + W(q)$ over the interval $t \in (t^*,T']$ satisfying the limits:
\begin{align} \label{eq:n_system}
	&\dot{q}_{n}(t) = \de_{p} E_{n}(p_{n}(t)) &&\dot{p}_{n}(t) = - \de_{q} W(q_{n}(t)) 	\\
\nonumber	&\lim_{t \downarrow t^*} q_{n}(t) = q^* &&\lim_{t \downarrow t^*} p_{n}(t) = p^*
\end{align}
such that $G(E_{n}(p_{n}(t))) > 0$ for all $t \in (t^*,T']$. This solution satisfies: 
\begin{equation} \label{eq:smooth_continuation_again}
	\text{for all } t \in (t^*,t^* + \delta'], \text{ } q_-(t) = {q}_{n}(t) \text{, } p_-(t) = {p}_{n}(t).
\end{equation}
It follows from \eqref{eq:smooth_continuation_again} that the map:
\begin{equation} \label{eq:extended_map_-}
	(\mathfrak{q}_-(t),\mathfrak{p}_-(t)) := \begin{cases} ({q}_-(t),{p}_-(t)) &\text{for } t \in [t^* - \delta',t^*+\delta'] \\ (q_{n}(t),p_{n}(t)) &\text{for } t \in [t^*+\delta',T]. \end{cases}
\end{equation}
is smooth over the interval $t \in [t^* - \delta',T']$. 
\end{proposition}
We again define, as in \eqref{eq:X_+}: 
\begin{equation} \label{eq:X_-}
	t \mapsto \mathfrak{X}_-(z;\mathfrak{p}_-(t)) := \begin{cases} \chi_-(z;p_-(t)) &\text{for } t \in [t^* - \delta',t^*+\delta'] \\ \chi_{n}(z;p_{n}(t)) &\text{for } t \in [t^*+\delta',T'] \end{cases}.
\end{equation}
The precise form of the wave `excited' at the crossing time is derived from a rigorous multiscale analysis on the emergent nonadiabatic time-scale \eqref{eq:nonadiabatic_timescale} (see Section \ref{sec:inner_solution}). The following definition is the result of this calculation: 
\begin{definition}[Parameters of the excited wave-packet] \label{def:-_band_functions}
We let $S_-(t)$ and $a^0_-(y,t)$ be defined for $t \in [t^* - \delta',t^* + \delta']$ by \eqref{eq:action_integral} and \eqref{eq:envelope_equation} with $t_0 = t^*$, in which all dependence on $p(t)$, $q(t)$, $E_n(p(t))$, $\chi_n(z;p(t))$, and $W(q(t))$ replaced by dependence on $p_-(t), q_-(t)$, $E_-(p_-(t))$, $\chi_-(z;p_-(t))$, and $W(q_-(t))$ respectively. 

Moreover, $S_-(t^*) = S^*$ and the initial data for $a^0_-(y,t)$, generated by the incoming `$+$ band' wave-packet is given by:
%\begin{equation} 	
%\end{equation}
\begin{equation} \label{eq:L_Z_connection}
\begin{split}
&	a^0_-(y,t^*) =\  \de_q W(q^*)\times \ip{\chi_-(\cdot;p^*)}{\de_p\chi_+(\cdot;p^*)} \\
	&\qquad  \times \int_{-\infty}^{\infty} e^{ i [\de_q W(q^*)][\de_p E_+(p^*) - \de_p E_-(p^*)] \tau^2 / 2 } \times {a}^{0,*}(y - [\de_p E_+(p^*) - \de_p E_-(p^*)] \tau) \, \mathrm{\emph{d}}\tau\ .
\end{split}
\end{equation}
Here, $S^* = \lim_{t \uparrow t^*} S(t)$ and $a^{0,*}(y) = \lim_{t \uparrow t^*} a^0(y,t)$ \eqref{eq:defs_of_stars}.

 Recall that $\de_q W(q^*)$ is assumed to be non-zero (see \eqref{eq:driven_through_pstar}) and that $\de_p E_+(p^*) - \de_p E_-(p^*)=2\de_p E_+(p^*)$ is always nonzero at band crossings (Theorem \ref{th:all_crossings_smooth}, Property \ref{band_crossing_assumption} (A4)) and hence the integral in \eqref{eq:L_Z_connection} is well-defined since $ a^{0,*}(y)$ is localized. We then define $S_n(t)$, $a^0_n(y,t)$ for $t \in (t^*,T']$ by replacing dependence on $p(t)$, $q(t)$, $E_n(p(t))$, $\chi_n(z;p(t))$, and $W(q(t))$ by dependence on $p_n(t), q_n(t)$, $E_n(p_n(t))$, $\chi_n(z;p_n(t))$, and $W(q_n(t))$ respectively, and by the limits:
\begin{equation}
	\lim_{t \downarrow t^*} S_n(t) = S^*, \text{ } \lim_{t \downarrow t^*} a^0_n(y,t) = a^{0,*}(y). 
\end{equation}
We denote by $\mathfrak{S}_-(t)$, $\mathfrak{a}^0_-(y,t)$ the smooth maps defined over the whole interval $t \in [t^* - \delta',T']$ in analogy with the definitions of  $\mathfrak{S}_+(t)$ and $ \mathfrak{a}_+(y,t)$ in \eqref{eq:a_+_extension}. 
\end{definition}
\begin{remark} 
Equation \eqref{eq:L_Z_connection} gives the precise form of the envelope of the wavepacket `excited' at the crossing time $t^*$. In Appendix \ref{app:consistency} we prove using this expression that the main statement \eqref{eq:statement_of_main_theorem} of Theorem \ref{th:band_crossing_theorem} (below) is consistent with the solution of an appropriate simplified `Landau-Zener' model. 
\end{remark}
We now define the wave-packet associated with the band $E_-$ which is `excited' at the crossing time $t^*$ by: 
\begin{equation} \label{eq:WP_1_-}
\begin{split}
	&\text{WP}^{0,\epsilon}[\mathfrak{S}_-(t),\mathfrak{q}_-(t),\mathfrak{p}_-(t),\mathfrak{a}^0_-(y,t),\mathfrak{X}_-(z;\mathfrak{p}_-(t))](x,t)	:= 	\\
	&:= \begin{cases} \text{{WP}}^{0,\epsilon}[S_-(t),q_-(t),p_-(t),a^0_-(y,t),\chi_-(z;p(t))](x,t) &\text{for } t \in [t^* - \delta',t^*+\delta'] \\ \text{{WP}}^{0,\epsilon}[S_{n}(t),q_{n}(t),p_{n}(t),a^0_{n}(y,t),\chi_{n}(z;p(t))](x,t) &\text{for } t \in [t^* + \delta',T] \end{cases}.
\end{split}
\end{equation}

\subsection{The main theorem}

Our main theorem is that a size $1$ incoming wave-packet associated with the `+ band', when encountering a band-crossing, generates a size $1$ `transmitted + band' wave-packet and a `reflected $-$ band' wave-packet of size $\sqrt\epsilon$. Moreover, when the wave-packet is in a neighborhood of the crossing, {\it i.e.} $t\approx t_*$ and hence $(p(t),q(t))\approx (p_*,q_*)$, 
the detailed dynamics are non-adiabatic and are described by an ansatz incorporating wave-packets from both bands with envelopes varying on an additional fast scale. The precise statement is the following: 
\begin{theorem} \label{th:band_crossing_theorem}
Assume the Band Crossing Scenario (Property \ref{classical_path_assumption}) in which the crossing time,
along the trajectory $(p(t),q(t))$ is $t=t_*$. Let $\xi, \xi'$ be fixed such that $3/8 < \xi' < \xi < 1/2$. Let $\tilde{T} > 0$, with $0<t_*<\tilde{T}$, be sufficiently small that Propositions \ref{prop:smooth_continuation} and \ref{prop:smooth_continuation_again} hold with $T = \tilde{T}$ and $T' = \tilde{T}$ respectively.
  
  Let $\psi^\epsilon(x,t)$ denote the unique solution of (\ref{eq:original_equation}) with `incident Bloch wavepacket' initial data (\ref{eq:bloch_wavepacket_initial_data}), defined for  $t \in [0,\tilde{T}]$.
  
   Then, there exists an $\epsilon_0 > 0$ such that for all $0 < \epsilon < \epsilon_0$ the following holds. 
\begin{enumerate}
\item For $t \in [0,t^* - \epsilon^\xi)$, $\psi^\epsilon(x,t)$ may be approximated up to errors of $o_{L^2_x}(\sqrt\epsilon)$ by a single-band ansatz (see Theorem \ref{th:limit_of_single_band_approximation} and Corollary \ref{cor:limit_of_single_band_approximation}):
\begin{equation} \label{eq:incoming_single_band_ansatz}
	\psi^\epsilon(x,t) = \text{\emph{WP}}^{1,\epsilon}[\mathfrak{S}_{+}(t),\mathfrak{q}_{+}(t),\mathfrak{p}_{+}(t),\mathfrak{a}^0_{+}(y,t),\mathfrak{a}^1_{+}(y,t),\mathfrak{X}_+(z;\mathfrak{p}_+(t))](x,t) + o_{L^2_x}(\sqrt\epsilon).
\end{equation}
\item For $t \in (t^* + \epsilon^{\xi},\tilde{T}]$, $\psi^\epsilon(x,t)$ is approximated up to errors of $o_{L^2_x}(\sqrt\epsilon)$ by the sum of two Bloch wave-packets, one associated with each band involved in the crossing (recall Definition \ref{def:-_band_functions}):
\begin{equation} \label{eq:statement_of_main_theorem}
\begin{split}
	&\psi^\epsilon(x,t) = \text{\emph{WP}}^{1,\epsilon}[\mathfrak{S}_{+}(t),\mathfrak{q}_{+}(t),\mathfrak{p}_{+}(t),\mathfrak{a}^0_{+}(y,t),\mathfrak{a}^1_{+}(y,t),\mathfrak{X}_+(z;\mathfrak{p}_+(t))](x,t) 	\\
	&+ \sqrt\epsilon \; \text{\emph{WP}}^{0,\epsilon}[\mathfrak{S}_-(t),\mathfrak{q}_-(t),\mathfrak{p}_-(t),\mathfrak{a}^0_-(y,t),\mathfrak{X}_-(z;\mathfrak{p}_-(t))](x,t) + o_{L^2_x}(\sqrt{\epsilon}). 
\end{split}
\end{equation}
\end{enumerate}

Over the interval $t \in (t^* - \epsilon^{\xi'},t^* + \epsilon^{\xi'})$, the solution $\psi^\epsilon(x,t)$ is expressible, with errors of size $o_{L^2_x}(\sqrt\epsilon)$, by a superposition of wave-packets from both $+$ and $-$ bands, whose amplitudes vary on an additional (fast / non-adiabatic) time scale:
\begin{equation}
	s := \frac{t - t^*}{\sqrt{\epsilon}}.
\end{equation}
The detailed construction appears in  Section \ref{sec:inner_solution}.
\end{theorem}
\begin{remark}
The restriction to sufficiently small $\tilde{T} > 0$ in Theorem \ref{th:band_crossing_theorem} is to ensure that neither the incident nor excited wavepacket encounter a second band crossing over the time interval $t \in [0,\tilde{T}]$. It is clear that this assumption may be relaxed and the analysis repeated each time a wavepacket is incident on a band crossing in order to obtain results valid over arbitrary finite time intervals, fixed independent of $\epsilon$.
\end{remark}
\begin{remark}
By construction: 
\begin{equation}
	\text{\emph{WP}}^{1,\epsilon}[\mathfrak{S}_{+}(t),\mathfrak{q}_{+}(t),\mathfrak{p}_{+}(t),\mathfrak{a}^0_{+}(y,t),\mathfrak{a}^1_{+}(y,t),\mathfrak{X}_+(z;\mathfrak{p}_+(t))](x,t)
\end{equation}
is a wavepacket with $L^2_x$-norm proportional to $1$ associated with the band $E_n$ for $t \in [0,t^* - \delta]$ and with the band $E_{n+1}$ for $t \in [t^* + \delta,\tilde{T}]$, and:
\begin{equation}
	\sqrt{\epsilon} \; \text{\emph{WP}}^{0,\epsilon}[\mathfrak{S}_-(t),\mathfrak{q}_-(t),\mathfrak{p}_-(t),\mathfrak{a}^0_-(y,t),\mathfrak{X}_-(z;\mathfrak{p}_-(t))](x,t)
\end{equation}
is a wavepacket with $L^2_x$-norm proportional to $\sqrt{\epsilon}$ associated with the band $E_n$ for $t \in [t^* + \delta',\tilde{T}]$. Hence the statement of Theorem \ref{th:band_crossing_theorem} is consistent with the description of our results given in Section \ref{sec:introduction}. 
\end{remark}
\begin{remark}
To leading order in $\sqrt{\epsilon}$, the center of mass of the wavepacket `excited' at the crossing is given by $\mathfrak{q}_-(t)$, which, for $|t - t^*|$ small enough, evolves according to \eqref{eq:-_band_classical_system}. The center of mass of the \emph{incoming} wavepacket is given by (again to leading order in $\sqrt{\epsilon}$) $\mathfrak{q}_+(t)$, which evolves (again for $|t - t^*|$ small enough) according to \eqref{eq:+_band_classical_system}. Since $\dot{\mathfrak{q}}_+(t^*) = \de_p E_+(p^*) > 0$ and $\dot{\mathfrak{q}}_-(t^*) = \de_p E_-(p^*) < 0$ we have that the velocities of the centers of mass of each wavepacket have \emph{opposite signs}: see Figure \ref{fig:centers_of_mass}.
\end{remark}
\begin{remark} \label{rem:leading_order} 
The $\mathcal{O}(\sqrt\epsilon)$ reflected wave is required only to describe the solution with an order 
of size $o(\sqrt\epsilon)$. Indeed, 
by dropping terms of $o(1)$ in $L^2_x$ in \eqref{eq:incoming_single_band_ansatz}, \eqref{eq:statement_of_main_theorem}, and in the asymptotic solution which we construct for $t \in (t^* - \epsilon^{\xi'},t^* + \epsilon^{\xi'})$, we have that, under the assumptions of Theorem \ref{th:band_crossing_theorem}, for all $t \in [0,\tilde{T}]$:
\begin{equation} \label{eq:statement}
	\psi^\epsilon(x,t) = \text{\emph{WP}}^{0,\epsilon}[\mathfrak{S}_+(t),\mathfrak{q}_+(t),\mathfrak{p}_+(t),\mathfrak{a}^0_+(y,t),\mathfrak{X}_+(z;\mathfrak{p}_+(t))](x,t) + o_{L^2_x}(1).
\end{equation}
Equation \eqref{eq:statement} may be summarized as follows: to leading order in $\sqrt{\epsilon}$, the wavepacket propagates as if associated with an isolated band, with band function given by the smooth extension of $E_n$ through the crossing. This is consistent with previous results obtained by Hagedorn \cite{hagedorn} and Jecko et al. \cite{2003Jecko,2009DuyckaertsFermanian-KammererJecko} on dynamics at codimension 1 eigenvalue band crossings.
\end{remark}

\section{Sketch of proof of Theorem \ref{th:limit_of_single_band_approximation} on blow-up of error in single-band approximation as $t$ approaches the crossing time $t^*$} \label{sec:proof_of_single_band_failure_theorem}
\subsection{Strategy for estimating the corrector}
In this section we recall the simple Lemma  which we use in the proofs of Theorem \ref{th:limit_of_single_band_approximation} and Theorem \ref{th:band_crossing_theorem} to estimate the corrector
to a wave-packet approximate solution. A similar strategy was followed in \cite{carles_sparber,2017WatsonWeinsteinLu}

\begin{lemma} \label{lem:key_lemma} 
For $0<T\le\infty$, let $\psi^\epsilon\in C^0([t_0,T); L^2(\mathbb{R}))$ denote the unique solution of the initial value problem \eqref{eq:original_equation} with initial data $\psi^\epsilon_0(x)$ given at $t = t_0$:
\begin{equation} 
\begin{split}
	&i \epsilon \de_t \psi^\epsilon = {H}^\epsilon \psi^\epsilon 	\\
	&\psi^\epsilon(x,t_0) = \psi^\epsilon_0(x) \ .
\end{split}
\end{equation}

Furthermore, let $\psi_{app}^\epsilon(x,t) \in C^0([t_0,T); L^2(\mathbb{R}))$, $r^\epsilon(x,t)$ be such that:
\begin{equation} \label{eq:definition_of_residual}
\begin{split}
	&i \epsilon \de_t \psi_{app}^\epsilon = {H}^\epsilon \psi_{app}^\epsilon + r^\epsilon 	\\
	&\psi_{app}^\epsilon(x,t_0) = \psi^\epsilon_{app,0}(x).
\end{split}
\end{equation}

Introduce $\eta^\epsilon(x,t)$ defined by: 
\begin{equation} \label{eq:def_of_eta}
	\eta^\epsilon(x,t) := \psi^\epsilon(x,t) - \psi^\epsilon_{app}(x,t).
\end{equation}
Then, 
\begin{equation} \label{eq:bound_on_eta}
	\| \eta^\epsilon(\cdot,t) \|_{L^2} \leq \| \psi^\epsilon_0(\cdot) - \psi^\epsilon_{app,0}(\cdot) \|_{L^2} + \frac{1}{\epsilon} \inty{t_0}{t}{ \| r^\epsilon(\cdot,t') \|_{L^2} }{t'}.
\end{equation}
\end{lemma}
\begin{remark}
We shall apply Lemma \ref{lem:key_lemma}  with  $\psi^\epsilon_{app}(x,t)$ equal to an \emph{approximate solution} of \eqref{eq:original_equation} and 
$r^\epsilon(x,t)=(i\partial_t-H^\epsilon)\psi^\epsilon_{app}$, a sufficiently high order in $\sqrt\epsilon$ \emph{residual}.
\end{remark}
\begin{proof}
The function $\eta^\epsilon(x,t)$ satisfies the initial value problem:
\begin{equation} \label{eq:equation_for_eta} 
\begin{split}
	&i \epsilon \de_t \eta^\epsilon = {H}^\epsilon \eta^\epsilon + r^\epsilon 	\\
	&\psi_{app}^\epsilon(x,t_0) = \psi^\epsilon_0(x) - \psi^\epsilon_{app,0}(x)
\end{split}	
\end{equation}
 Multiplying both sides of \eqref{eq:equation_for_eta} by $\overline{\eta^\epsilon}$, taking the imaginary part and using self-adjointness of $H^\epsilon$ yields:
$	\epsilon \partial_{t} \| \eta^\epsilon \|^2_{L^2} = - i \ip{\eta^\epsilon}{r^\epsilon}_{L^2} + i \ip{r^\epsilon}{\eta^\epsilon}_{L^2}. 
$
This implies, using the Cauchy-Schwarz inequality, that
$
	2 \epsilon \| \eta^\epsilon \|_{L^2}\ \partial_t \| \eta^\epsilon \|_{L^2} \leq 2 \| r^\epsilon \|_{L^2} \| \eta^\epsilon \|_{L^2}.
$
Cancelling common factors from both sides (note that the inequality is trivially true if $\| \eta^\epsilon \|_{L^2} = 0$) and integrating from $t_0$ to $t$ gives \eqref{eq:bound_on_eta}.
\end{proof}

We now estimate the error in the single-band approximation as $t \uparrow t^*$, as measured by the $L^2_x$-norm of the corrector function $\eta^\epsilon(x,t)$ which appears in \eqref{eq:WP_and_corrector}. We start by recalling the strategy of the proof of Theorem \ref{th:isolated_band_theorem_2}; the proof of Theorem \ref{th:isolated_band_theorem_1} is similar. Let $\psi^\epsilon(x,t)$ denote the exact solution of \eqref{eq:original_equation} with approximate `Bloch wavepacket' initial data \eqref{eq:approximate_bloch_wp_data} specified at $t = t_0$. Then by Lemma \ref{lem:key_lemma}, if we can find an approximate solution of \eqref{eq:original_equation}, $\psi^\epsilon_{app}(x,t)$, such that \eqref{eq:definition_of_residual} holds with:
\begin{align}
	&\| \psi^\epsilon_0(\cdot) - \psi^\epsilon_{app}(\cdot,t)\vert_{t=0} \|_{L^2} \leq C \epsilon	\\
	&\text{ and } \| r^\epsilon(\cdot,t) \|_{L^2} \leq C e^{c t} \epsilon^2, \label{eq:residual_bound}
\end{align}
where the constants $C > 0, c > 0$ are independent of $\epsilon, t$, then it follows from \eqref{eq:def_of_eta} and \eqref{eq:bound_on_eta} that:
\begin{equation}
	\| \psi^\epsilon(\cdot,t) - \psi^\epsilon_{app}(\cdot,t) \|_{L^2} \leq C \epsilon e^{c t}.
\end{equation}
If  in addition we have that:
\begin{equation}
	\|  \psi^\epsilon_{app}(\cdot,t)\ -\ \text{WP}^{1,\epsilon}[S(t),q(t),p(t),a^0(y,t),a^1(y,t),\chi_n(z;p(t))](\cdot,t)  \|_{L^2} \leq C e^{c t} \epsilon,
\end{equation}
where $q(t), p(t)$ and so on are as in the statement of Theorem \ref{th:isolated_band_theorem_2}, then the conclusions of Theorem \ref{th:isolated_band_theorem_2} follow immediately by the triangle inequality. The details of how to construct such a $\psi^\epsilon_{app}(x,t)$ were presented in \cite{2017WatsonWeinsteinLu}. 

Theorem \ref{th:isolated_band_theorem_2}  implies, in particular, that the solution $\psi^\epsilon(x,t)$ of \eqref{eq:original_equation} with initial data \eqref{eq:bloch_wavepacket_initial_data} satisfies:
\begin{equation}
\begin{split}
	&\text{for $t \in [0,t^* - \delta]$, } \\
	&\|  \psi^\epsilon(\cdot,t)\ -\ \text{WP}^{1,\epsilon}[S(t),q(t),p(t),a^0(y,t),a^1(y,t),\chi_n(z;p(t))](\cdot,t)  \|_{L^2} = O_{L^2_x}(\epsilon).
\end{split}
\end{equation}

Due to the band crossing at $p^*$, the Isolated Band Property \ref{isolated_band_assumption} does not hold as $t \uparrow t^*$. As a result
% even making proper use of the smooth band functions at the crossing,
  the proof of Theorem \ref{th:isolated_band_theorem_2} fails as follows:
\begin{enumerate}
\item As $t \uparrow t^*$, the $L^2_x$-norm of the residual $r^\epsilon(x,t)$ defined by \eqref{eq:definition_of_residual}, diverges.
\item The integral $\frac{1}{\epsilon} \inty{0}{t}{ \| r^\epsilon(\cdot,t') \|_{L^2} }{t'}$, and hence the bound \eqref{eq:bound_on_eta} on the $L^2$-norm of the corrector function $\eta^\epsilon(x,t)$ diverges as $t \uparrow t^*$. 
\end{enumerate}

Theorem \ref{th:limit_of_single_band_approximation} is proved by analyzing the  rates of blow-up of singular terms in $r^\epsilon(x,t)$ and then deducing the resulting rate of blow-up of the bound \eqref{eq:bound_on_eta}. In Section \ref{sec:estimate_of_representative_term} we explain the strategy by studying the term which contributes the first term on the right-hand side of the estimate \eqref{eq:bound_on_eta_near_the_crossing}, which is the dominant term when $t = t^* - \epsilon^\xi$ with $0 < \xi < 1/2$ and controls the time interval of validity of the single-band ansatz (Corollary \ref{cor:limit_of_single_band_approximation}). For a sketch of the general argument, see Appendix B.3 of \cite{thesis}. 

\subsection{Estimation of representative term demonstrating blow-up as $t \uparrow t^*$} \label{sec:estimate_of_representative_term}
Let $t \in [t^* - \delta,t^*]$ where $\delta > 0$ is as in Proposition \ref{prop:smooth_continuation} so that:
\begin{equation}
\begin{split}
	&\text{WP}^{1,\epsilon}[\mathfrak{S}_+(t),\mathfrak{q}_+(t),\mathfrak{p}_+(t),\mathfrak{a}^0_+(y,t),\mathfrak{a}^1_+(y,t),\mathfrak{X}_+(z;\mathfrak{p}_+(t))](x,t)	 	\\
	&= \text{{WP}}^{1,\epsilon}[S_+(t),q_+(t),p_+(t),a^0_+(y,t),a^1_+(y,t),\chi_+(z;p(t))](x,t)
\end{split}
\end{equation}
Here, $q_+(t), p_+(t)$ are as in \eqref{eq:+_band_classical_system}, $S_+(t), a^0_+(y,t), a^1_+(y,t)$ are as in Definition \ref{def:S_a_b_extensions}, and $E_+(p), \chi_+(z;p)$ are as in \eqref{eq:def_smooth_bands}. The representative term which appears in the residual $r^\epsilon(x,t)$ \eqref{eq:definition_of_residual} which we now consider is the following: 
\begin{equation} \label{eq:R}
\begin{split}
	&R^\epsilon(x,t) := \epsilon^{-1/4} e^{i \phi^\epsilon_+(y,t)/\epsilon} \left[ \vphantom{\epsilon^{-1/2}} \right. \\
	&\left. \left. \epsilon^2 (- i \de_t)\ \left( - i \de_q W(q_+(t))\ {a}^0_+(y,t) \mathcal{R}_+(p_+(t))\ P^\perp_+(p_+(t))\ \de_p \chi_+(z;p_+(t))\ \right) \vphantom{\epsilon^{-1/2}} \right] \right|_{y = \frac{x - q_+(t)}{\epsilon^{1/2}}, z = \frac{x}{\epsilon} }	\\
	&\text{where } \phi^\epsilon_+(y,t) := S_+(t) + \epsilon^{1/2} p_+(t) y 
\end{split}
\end{equation}
Here, $P^\perp_+(p)$ denotes the projection operator onto the orthogonal complement of the subspace of $L^2_{per}$ spanned by $\chi_+(z;p)$, and:
\begin{equation} \label{eq:resolvent}
	\mathcal{R}_+(p) := ( H(p) - E_+(p) )^{-1}
\end{equation}
denotes the resolvent operator where $H(p)$ is as in \eqref{eq:reduced_eigenvalue_problem}. Because of the band crossing at $p^*$, the operator $\mathcal{R}_+(p) P^\perp_+(p)$ is singular as $p \rightarrow p^*$ on the `resonant'  subspace of $L^2_{per}$ spanned by $\chi_-(z;p)$. 
The operator $\mathcal{R}_+(p) P^\perp_\pm(p)$ however, where $P^\perp_\pm(p)$ is defined as the projection onto the orthogonal complement of $\chi_+(z;p)$ and $\chi_-(z;p)$ in $L^2_{per}$ is regular for all $p \in U$ by \emph{\ref{other_bands_isolated_away}}  of Property \ref{band_crossing_assumption} (see Corollary \ref{cor:away_from_crossing_regular}). 

We isolate the singular part of \eqref{eq:R} as follows. 
Expressing  $\de_p \chi_+(z;p_+(t))$ in terms of its projections onto both $\chi_+(z;p_+(t))$ and $\chi_-(z;p_+(t))$, 
 and their orthogonal complement, the range of $P^\perp_\pm(p_+(t))$, we have
\begin{equation}\label{resolve-exp}
\begin{split}
	&\mathcal{R}_+(p_+(t)) P^\perp_+(p_+(t)) \de_p \chi_+(z;p_+(t)) = \mathcal{R}_+(p(t)) P^\perp_\pm(p_+(t)) \de_p \chi_+(z;p_+(t))   \\
	&+ (E_-(p_+(t)) - E_+(p_+(t)))^{-1} \ip{\chi_-(\cdot;p_+(t))}{\de_p \chi_+(\cdot ;p_+(t))} \chi_-(z;p(t))\ .
\end{split}
\end{equation}
Since $p_+(t)\to p^*$ as $t\uparrow t_*$, $E_+(p^*)=E_-(p^*)$, the singular behavior is isolated in the latter term of \eqref{resolve-exp}. 
We decompose  $R^\epsilon(x,t)$ into its corresponding regular and singular parts:
\begin{equation} \label{eq:R_again}
	R^\epsilon(x,t) := R^\epsilon_{regular}(x,t) + R^\epsilon_{singular}(x,t) 
\end{equation}
where:
\begin{equation}
\begin{split}
	&R^\epsilon_{regular}(x,t) = \epsilon^{-1/4} e^{i \phi^\epsilon_+(y,t) / \epsilon} \left[ \vphantom{\epsilon^{-1/4}} \right.  	\\
	&\left. \left. \epsilon^2 (- i \de_t) \left( - i \de_q W(q_+(t)) {a}^0_+(y,t) \mathcal{R}_+(p_+(t)) P^\perp_\pm(p_+(t)) \de_p \chi_+(z;p_+(t)) \vphantom{\epsilon^{-1/4}}  \right) \right] \right|_{y = \frac{x - q_+(t)}{\epsilon^{1/2}}, z = \frac{x}{\epsilon} }	\\
	&R^\epsilon_{singular}(x,t) = \epsilon^{-1/4} e^{i \phi^\epsilon_+(y,t) / \epsilon} \left[ \vphantom{\epsilon^{-1/4}} \epsilon^2 (- i \de_t) \left( \vphantom{\epsilon^{-1/4}}  - i \de_q W(q_+(t)) {a}^0_+(y,t) \right. \right. 	\\
	&\times \left. \left. \left. (E_-(p_+(t)) - E_+(p_+(t)))^{-1}\ \ip{\chi_-(z;p_+(t))}{\de_p \chi_+(z;p_+(t))} \vphantom{\epsilon^{-1/4}}\ \chi_-(z;p_+(t)) \right) \right] \right|_{y = \frac{x - q_+(t)}{\epsilon^{1/2}}, z = \frac{x}{\epsilon} }.
\end{split}
\end{equation}
It follows from the techniques detailed in \cite{2017WatsonWeinsteinLu} that:
\begin{equation} \label{eq:ok_term}
	R^\epsilon_{regular}(x,t) = O_{L^2_x}(\epsilon^2)
\end{equation}
uniformly as $t \uparrow t^*$. On the other hand, $R^\epsilon_{singular}(x,t)$ is explicitly singular, since it depends on $( E_-(p_+(t)) - E_+(p_+(t)) )^{-1}$ which is unbounded as $t \uparrow t^*$. The time derivative of $R^\epsilon_{singular}(x,t)$ yields two terms:
\begin{equation}
	R^\epsilon_{singular}(x,t) = R^{1,\epsilon}_{singular}(x,t) + R^{2,\epsilon}_{singular}(x,t),
\end{equation}
where:
\begin{equation}
\begin{split}
	&R^{1,\epsilon}_{singular}(x,t) := \epsilon^{-1/4} e^{i \phi^\epsilon_+(y,t)/\epsilon} \left[ \vphantom{\epsilon^{-1/2}} \epsilon^2 (E_-(p_+(t)) - E_+(p_+(t)))^{-1} \right. \\
	&\left. \left. \times (- i \de_t) \left( \vphantom{\epsilon^{-1/2}} - i \de_q W(q_+(t)) {a}^0_+(y,t) \ip{\chi_-(z;p_+(t))}{\de_p \chi_+(z;p_+(t))} \chi_-(z;p_+(t)) \right) \right] \right|_{y = \frac{x - q_+(t)}{\epsilon^{1/2}}, z = \frac{x}{\epsilon} }	\\
	&R^{2,\epsilon}_{singular}(x,t) := \epsilon^{-1/4} e^{i \phi^\epsilon_+(y,t)/\epsilon} \left[ \vphantom{\epsilon^{-1/2}} \epsilon^2 \de_t \left( \vphantom{\epsilon^{-1/4}} (E_-(p_+(t)) - E_+(p_+(t)))^{-1} \right)  \right. 	\\
	&\left. \left. \times (- i) \left( - i \de_q W(q_+(t)) {a}^0_+(y,t) \vphantom{\epsilon^{-1/2}} \ip{\chi_-(z;p_+(t))}{\de_p \chi_+(z;p_+(t))} \chi_-(z;p_+(t)) \right) \right] \right|_{y = \frac{x - q_+(t)}{\epsilon^{1/2}}, z = \frac{x}{\epsilon} }.
\end{split}
\end{equation}
We now concentrate on $R^{2,\epsilon}_{singular}(x,t)$ which will turn out to be the dominant term as $t \uparrow t^*$. We first evaluate the time-derivative:
\begin{equation}
\begin{split}
	&\de_t \left( \left( E_-(p_+(t)) - E_+(p_+(t)) \right)^{-1} \right) =	\\
	&\de_q W(q_+(t)) ( \de_p E_-(p_+(t)) - \de_p E_+(p_+(t)) ) \left( E_-(p_+(t)) - E_+(p_+(t)) \right)^{-2}.
\end{split}
\end{equation}
We then follow \cite{carles_sparber,2017WatsonWeinsteinLu} in estimating $R^{2,\epsilon}_{singular}$ in $L^2_x$ by taking the $L^\infty$ norm of all $z$-dependence and the $L^2$-norm of all $y$ dependence:
\begin{equation} \label{eq:R_again_again}
\begin{split}
	&\| R^{2,\epsilon}_{singular}(\cdot,t) \|_{L^2} \leq \epsilon^2 |\de_q W(q_+(t))|^2 \left| \de_p E_-(p_+(t)) - \de_p E_+(p_+(t)) \right|  \\
	&\times \left| E_-(p_+(t)) - E_+(p_+(t)) \right|^{-2} \| {a}^0_+(\cdot,t) \|_{L^2} \left| \ip{\chi_-(\cdot;p_+(t))}{\de_p \chi_+(\cdot;p_+(t))}_{L^2[0,1]} \right| \| \chi_+(\cdot,p_+(t)) \|_{L^\infty[0,1]}.
\end{split}
\end{equation}
Taylor-expanding in $t - t^*$, using the non-degeneracy conditions \eqref{eq:linear_splitting} and \eqref{eq:driven_through_pstar}, we have that as $t \uparrow t^*$:
\begin{equation} \label{eq:E_plus_minus_expanded}
\begin{split}
	&\left| \left( E_+(p_+(t)) - E_-(p_+(t)) \right)^{-2} \right| \leq  \\
	&\left| \frac{ 1 }{ \de_q W(q^*) \left( \de_p E_+(p^*) - \de_p E_-(p^*) \right) } \right|^2 \left( \frac{1}{ |t - t^*|^2 } \right) + O\left(\frac{1}{|t - t^*|}\right).
\end{split}
\end{equation}
Substituting \eqref{eq:E_plus_minus_expanded} into \eqref{eq:R_again_again} and Taylor-expanding all other terms gives:
\begin{equation} \label{eq:term}
\begin{split}
	&\| R^{2,\epsilon}_{singular}(\cdot,t) \|_{L^2} \leq 	\\
	&\left| \frac{ \ip{\chi_-(\cdot;p^*)}{\de_p \chi_+(\cdot;p^*)}_{L^2[0,1]} \| {a}^0_+(\cdot,t^*) \|_{L^2} \| \chi_+(\cdot,p^*) \|_{L^\infty[0,1]} }{ \de_p E_-(p^*) - \de_p E_+(p^*) }\right| \left( \frac{ \epsilon^2 }{ | t - t^* |^2 } \right)  + O\left(\frac{\epsilon^2}{|t - t^*|}\right)
\end{split}
\end{equation}
Similar analysis shows that:
\begin{equation} \label{eq:less_bad_term}
	\| R^{1,\epsilon}_{singular}(\cdot,t) \|_{L^2} = O\left(\frac{\epsilon^2}{|t - t^*|}\right).
\end{equation}
Recall the relationship between the residual $r^\epsilon(x,t)$ and the bound on the solution error $\eta^\epsilon(x,t) := \psi^\epsilon(x,t) - \psi^\epsilon_{app}(x,t)$ \eqref{eq:bound_on_eta}. Putting everything (\eqref{eq:R_again}, \eqref{eq:ok_term}, \eqref{eq:term}, and \eqref{eq:less_bad_term}) together, then integrating once in time and dividing by $\epsilon$, we see that the term contributed by $R^\epsilon(x,t)$ to the solution error $\eta^\epsilon(x,t)$ may be bounded by: 
\begin{equation}\label{total-error}
\begin{split}
	&\frac{1}{\epsilon} \inty{0}{t}{ \| R^\epsilon(\cdot,t') \|_{L^2} }{t'} \leq \\
	&\left| \frac{ \ip{\chi_-(\cdot;p^*)}{\de_p \chi_+(\cdot;p^*)}_{L^2[0,1]} \| a^0_+(\cdot,t^*) \|_{L^2} \| \chi_+(\cdot,p^*) \|_{L^\infty[0,1]} }{ \de_p E_-(p^*) - \de_p E_+(p^*) }\right| \left( \frac{ \epsilon }{ | t - t^* | } \right)\\
	& + O\left( \epsilon , \epsilon \ln |t - t^*| \right).
\end{split}
\end{equation}
It follows that the right hand side of \eqref{total-error}, which is the bound on the corrector to the wave-packet ansatz, is of the same size as the ${O}(\epsilon^{1/2})$ term in the wave-packet approximate solution for $|t-t_*|\sim \epsilon^{1/2}$. 

\section{Proof of Theorem \ref{th:band_crossing_theorem} on coupled band dynamics when $t \sim t^*$}
We now turn to the proof of Theorem \ref{th:band_crossing_theorem}, on the dynamics through the crossing time $t^*$. We follow the strategy adopted by Hagedorn, see in particular the proof of Theorem 5.1 in \cite{hagedorn}. Theorem \ref{th:limit_of_single_band_approximation} and Corollary \ref{cor:limit_of_single_band_approximation} give a description of the exact solution $\psi^\epsilon(x,t)$ of \eqref{eq:original_equation} with initial data given by \eqref{eq:bloch_wavepacket_initial_data} which is valid with errors of $o_{L^2_x}(\epsilon^{1/2})$ up to $t = t^* - \epsilon^\xi$ for any $\xi \in (0,1/2)$:
\begin{equation} \label{eq:description}
\begin{split}
	&\text{For all $t \in [0,t^* - \epsilon^\xi)$,} 	\\
	&\psi^\epsilon(x,t) = \text{WP}^{1,\epsilon}[\mathfrak{S}_+(t),\mathfrak{q}_+(t),\mathfrak{p}_+(t),\mathfrak{a}^0_+(y,t),\mathfrak{a}^1_+(y,t),\mathfrak{X}_+(z;\mathfrak{p}_+(t))](x,t) + o_{L^2_x}(\epsilon^{1/2}). 
\end{split}
\end{equation}
We seek to extend \eqref{eq:description} to a description of $\psi^\epsilon(x,t)$ up to errors of $o_{L^2_x}(\epsilon^{1/2})$ over the entire interval $t \in [0,\tilde{T}]$ where $\tilde{T}$ is chosen such that Propositions \ref{prop:smooth_continuation} and \ref{prop:smooth_continuation_again} hold with $T = \tilde{T}$ and $T' = \tilde{T}$. We first claim that the proof of Theorem \ref{th:band_crossing_theorem} may be reduced to (a) the construction of an `inner solution'  $\psi^\epsilon_{app,inner}(x,t)$ satisfying certain properties and (b) an application of Lemma \ref{lem:key_lemma}: 
\begin{proposition} \label{prop:idea_of_proof}
Let $\xi, \xi' \in (0,1/2)$ be such that $\xi' < \xi$ so that $(t^* - \epsilon^\xi,t^* + \epsilon^\xi) \subset (t^* - \epsilon^{\xi'},t^* + \epsilon^{\xi'})$. Assume \eqref{eq:description} for an incoming wave-packet. Consider an approximate solution $\psi^\epsilon_{app,inner}(x,t)$ which satisfies the following three properties:
\begin{enumerate}[label=\emph{(P\arabic*)}]
\item $\psi^\epsilon_{app,inner}(x,t)$ is equal to the single-band ansatz in the `incoming' overlap region $t \in (t^* - \epsilon^{\xi'},t^* - \epsilon^\xi)$ up to errors of $o(\epsilon^{1/2})$ in $L^2(\field{R})$. That is,
\begin{equation} \label{eq:incoming_match}
\begin{split}
	&\text{for all $t \in (t^* - \epsilon^{\xi'},t^* - \epsilon^\xi)$},	\\
	&\left\| \psi^\epsilon_{app,inner}(\cdot,t) - \text{\emph{WP}}^{1,\epsilon}[\mathfrak{S}_+(t),\mathfrak{q}_+(t),\mathfrak{p}_+(t),\mathfrak{a}^0_+(y,t),\mathfrak{a}^1_+(y,t),\mathfrak{X}_+(z;\mathfrak{p}_+(t))](\cdot,t) \right\|_{L^2} = o(\epsilon^{1/2}),
\end{split}
\end{equation}
\item $\psi^\epsilon_{app,inner}(x,t)$ is an approximate solution to \eqref{eq:original_equation} $(i \epsilon \de_t - H^\epsilon) \psi^\epsilon = 0$:
\begin{equation} \label{eq:equation_for_psi_inner}
	i \epsilon \de_t \psi^\epsilon_{app,inner} - H^\epsilon \psi^\epsilon_{app,inner} = r^\epsilon_{inner} 	
\end{equation}
with residual satisfying the bound: 
\begin{equation} \label{eq:inner_residual}
	\frac{1}{\epsilon} \inty{t^* - \epsilon^{\xi'}}{t^* + \epsilon^{\xi'}}{ \| r_{inner}^\epsilon(\cdot,t') \|_{L^2} }{t'} = o(\epsilon^{1/2}),
\end{equation}
\item $\psi^\epsilon_{app,inner}(x,t)$ matches the `two-band' ansatz of \eqref{eq:statement_of_main_theorem} in the `outgoing' overlap region $t \in (t^* + \epsilon^{\xi},t^* + \epsilon^{\xi'})$ up to errors of $o(\epsilon^{1/2})$ in $L^2(\field{R})$. That is,
\begin{equation} \label{eq:outgoing_match}
\begin{split}
	&\text{for all $t \in (t^* + \epsilon^{\xi},t^* + \epsilon^{\xi'})$},	\\
	&\left\| \psi^\epsilon_{app,inner}(\cdot,t) - \text{\emph{WP}}^{1,\epsilon}[\mathfrak{S}_+(t),\mathfrak{q}_+(t),\mathfrak{p}_+(t),\mathfrak{a}^0_+(y,t),\mathfrak{a}^1_+(y,t),\mathfrak{X}_+(z;\mathfrak{p}_+(t))](\cdot,t) \right.	\\
	&\left. - \epsilon^{1/2} \text{\emph{WP}}^{0,\epsilon}[\mathfrak{S}_-(t),\mathfrak{q}_-(t),\mathfrak{p}_-(t),\mathfrak{a}^0_-(y,t),\mathfrak{X}_-(z;\mathfrak{p}_-(t))](\cdot,t) \right\|_{L^2} = o(\epsilon^{1/2}).
\end{split}
\end{equation}
\end{enumerate}
Then, under conditions (P1), (P2), and (P3), Theorem \ref{th:band_crossing_theorem} holds. 
\end{proposition}
\begin{proof}
We apply Lemma \ref{lem:key_lemma} with $t_0 = t^* - \epsilon^{\xi'}$, $t_1 = t^* + \epsilon^{\xi'}$, $\psi^\epsilon_{app}(x,t) = \psi^\epsilon_{app,inner}(x,t)$, and $r^\epsilon(x,t) = r^\epsilon_{app,inner}(x,t)$. It then follows from (P1) and (P2) that:
\begin{equation} \label{eq:psi_app_and_exact}
\begin{split}
	&\text{For all $t \in (t^* - \epsilon^{\xi'}, t^* + \epsilon^{\xi'})$},	\\
	&\| \psi^\epsilon(\cdot,t) - \psi^\epsilon_{app,inner}(\cdot,t) \|_{L^2} = o(\epsilon^{1/2}).
\end{split}
\end{equation}
Combining \eqref{eq:psi_app_and_exact} with (P3), we have that:
\begin{equation} \label{eq:outgoing_match_again}
\begin{split}
	&\text{For all $t \in (t^* + \epsilon^{\xi},t^* + \epsilon^{\xi'})$},	\\
	&\left\| \psi^\epsilon(\cdot,t) - \text{{WP}}^{1,\epsilon}[\mathfrak{S}_+(t),\mathfrak{q}_+(t),\mathfrak{p}_+(t),\mathfrak{a}^0_+(y,t),\mathfrak{a}^1_+(y,t),\mathfrak{X}_+(z;\mathfrak{p}_+(t))](\cdot,t) \right.	\\
	&\left. + \epsilon^{1/2} \text{{WP}}^{0,\epsilon}[\mathfrak{S}_-(t),\mathfrak{q}_-(t),\mathfrak{p}_-(t),\mathfrak{a}^0_-(y,t),\mathfrak{X}_-(z;\mathfrak{p}_-(t))](x,t) \right\|_{L^2} = o(\epsilon^{1/2}).
\end{split}
\end{equation}
We claim that the main statement \eqref{eq:statement_of_main_theorem} of Theorem \ref{th:band_crossing_theorem} then follows from the isolated band theory. For any $\tilde{T}_0$ fixed independent of $\epsilon$ such that $t^* < \tilde{T}_0 < \tilde{T}$, the Isolated Band Property \ref{isolated_band_assumption} holds for the bands $E_n(p)$ and $E_{n+1}(p)$ and trajectories $p_n(t)$ and $p_{n+1}(t)$ defined by \eqref{eq:n_system}, \eqref{eq:n+1_system} with $t_0 = \tilde{T}_0, t_1 = \tilde{T}$. By linearity, the two-band wavepacket ansatz agrees (modulo errors of $o_{L^2_x}(\epsilon^{1/2})$) with the exact solution $\psi^\epsilon_{outgoing}(x,t)$ of the full equation \eqref{eq:original_equation} over the interval $t \in [\tilde{T}_0,\tilde{T}]$ with initial data given at $\tilde{T}_0$ by:
\begin{equation}
\begin{split}
	&\psi^\epsilon_{outgoing}(x,\tilde{T}_0) = \\
	&\text{{WP}}^{1,\epsilon}[\mathfrak{S}_+(\tilde{T}_0),\mathfrak{q}_+(\tilde{T}_0),\mathfrak{p}_+(\tilde{T}_0),\mathfrak{a}^0_+(y,\tilde{T}_0),\mathfrak{a}^1_+(y,\tilde{T}_0),\mathfrak{X}_+(z;\mathfrak{p}_+(\tilde{T}_0))](x,\tilde{T}_0) 	\\
	&+ \epsilon^{1/2} \text{{WP}}^{0,\epsilon}[\mathfrak{S}_-(\tilde{T}_0),\mathfrak{q}_-(\tilde{T}_0),\mathfrak{p}_-(\tilde{T}_0),\mathfrak{a}^0_-(y,\tilde{T}_0),\mathfrak{X}_-(z;\mathfrak{p}_-(\tilde{T}_0))](x,\tilde{T}_0).
\end{split}
\end{equation}
By performing the same analysis as in the proof of Theorem \ref{th:limit_of_single_band_approximation} backwards in time towards $t^*$, we have that:
\begin{equation} \label{eq:outgoing_agrees_with_2_band_ansatz}
\begin{split}
	&\text{For all $t \in (t^* + \epsilon^{\xi},\tilde{T}]$},	\\
	&\left\| \psi_{outgoing}^\epsilon(\cdot,t) - \text{{WP}}^{1,\epsilon}[\mathfrak{S}_+(t),\mathfrak{q}_+(t),\mathfrak{p}_+(t),\mathfrak{a}^0_+(y,t),\mathfrak{a}^1_+(y,t),\mathfrak{X}_+(z;\mathfrak{p}_+(t))](\cdot,t) \right.	\\
	&\left. + \epsilon^{1/2} \text{{WP}}^{0,\epsilon}[\mathfrak{S}_-(t),\mathfrak{q}_-(t),\mathfrak{p}_-(t),\mathfrak{a}^0_-(y,t),\mathfrak{X}_-(z;\mathfrak{p}_-(t))](x,t) \right\|_{L^2} = o(\epsilon^{1/2}).
\end{split}
\end{equation}
But now combining the triangle inequality with \eqref{eq:outgoing_match_again}, we have that:
\begin{equation} \label{eq:outgoing_match_again_again}
\begin{split}
	&\text{For all $t \in (t^* + \epsilon^{\xi},t^* + \epsilon^{\xi'})$},	\\
	&\left\| \psi^\epsilon(\cdot,t) - \psi^\epsilon_{outgoing}(\cdot,t) \right\|_{L^2} = o(\epsilon^{1/2}).
\end{split}
\end{equation}
Since $\psi^\epsilon(x,t)$ and $\psi^\epsilon_{outgoing}(x,t)$ are both exact solutions of \eqref{eq:original_equation}, applying Lemma \ref{lem:key_lemma} one more time with $\psi^\epsilon_{app}(x,t) = \psi^\epsilon_{outgoing}(x,t)$ gives that:
\begin{equation} \label{eq:outgoing}
\begin{split}
	&\text{For all $t \in (t^* + \epsilon^{\xi},\tilde{T}]$},	\\
	&\left\| \psi^\epsilon(\cdot,t) - \psi^\epsilon_{outgoing}(\cdot,t) \right\|_{L^2} = o(\epsilon^{1/2}).
\end{split}
\end{equation}
The main statement of Theorem \ref{th:band_crossing_theorem} \eqref{eq:statement_of_main_theorem} then follows from combining \eqref{eq:outgoing} and \eqref{eq:outgoing_agrees_with_2_band_ansatz}. 
\end{proof}

This brings us to the core construction of the paper.

\subsection{Derivation of $\psi^\epsilon_{app,inner}$ satisfying hypotheses of Proposition \ref{prop:idea_of_proof}} \label{sec:inner_solution}
We make the following ansatz for $\psi^\epsilon_{app,inner}(x,t)$, which incorporates both $+$ and $-$ bands, and a new `fast' timescale: 
\begin{equation}
	s = \frac{ t - t^* }{ \epsilon^{1/2} },
\end{equation}
which was motivated by the preceding  single-band analysis:
\begin{equation} \label{eq:psi_inner_ansatz}
	\psi^\epsilon_{app,inner}(x,t) = \epsilon^{-1/4} \sum_{\sigma = \pm} \left. e^{i \{ S_\sigma(t) + \epsilon^{1/2} p_\sigma(t) y_\sigma \} / \epsilon } f^\epsilon_{\sigma,inner}\left(y_\sigma,z,t,s\right) \right|_{y_\sigma = \frac{x - q_\sigma(t)}{\epsilon^{1/2}}, z = \frac{x}{\epsilon}, s = \frac{t - t^*}{\epsilon^{1/2}} }. 
\end{equation}
The new time scale has been introduced in the envelope functions $f^\epsilon_{\sigma,inner}$. 
We take $(q_\sigma(t),p_\sigma(t)), \sigma = \pm$ as in \eqref{eq:+_band_classical_system} and \eqref{eq:-_band_classical_system}, $S_\sigma(t), \sigma = \pm$, as in Definitions \ref{def:S_a_b_extensions} and \ref{def:-_band_functions}, and assume that $f^\epsilon_{\sigma,inner}(y_\sigma,z,t,s)$ may be expanded in powers of $\epsilon^{1/2}$:
\begin{equation} \label{eq:f_as_a_series}
	f^\epsilon_{\sigma,inner}(y_\sigma,z,t,s) = f^0_{\sigma,inner}(y_\sigma,z,t,s) + \epsilon^{1/2} f^1_{\sigma,inner}(y_\sigma,z,t,s) + ...
\end{equation}

Then,  $\psi^\epsilon_{app,inner}$, given by \eqref{eq:psi_inner_ansatz}, satisfies the non-homogeneous Schroedinger equation \eqref{eq:equation_for_psi_inner} with residual:
{\footnotesize{
\begin{equation} \label{eq:equation_for_r_inner}
\begin{split}
	&r^\epsilon_{inner}(x,t) = \epsilon^{-1/4} \sum_{\sigma = \pm} e^{i \left\{ S_\sigma(t) + \epsilon^{1/2} p_\sigma(t) y_\sigma \right\} / \epsilon}\left\{\vphantom{\frac{1}{2}} \epsilon^{3/2} \left[ y_\sigma^3 \inty{0}{1}{ \frac{(\tau - 1)^3}{3!} \de_x^3W\left(q_\sigma(t) + \tau \epsilon^{1/2} y_\sigma \right)}{\tau}\right] \right.  \\
	&+ \epsilon \left[ \frac{1}{2} (- i \de_{y_\sigma})^2 + \frac{1}{2} \de_x^2 W(q_\sigma(t)) y_\sigma^2 - i \de_t \right] + \epsilon^{1/2} \left[ \vphantom{\frac{1}{2}} \left( \vphantom{\epsilon^{-1/2}} p_\sigma(t) - i \de_z - \de_p E_\sigma(p_\sigma(t)) \right)(- i \de_{y_\sigma}) - i \de_s \right]   \\
	&\left. \left. + \left[ \vphantom{\frac{1}{2}} H(p_\sigma(t)) - E_\sigma(p_\sigma(t)) \right] \vphantom{\frac{1}{2}} \right\} \left\{ f^0_{\sigma,inner}(y_\sigma,z,t,s) + \epsilon^{1/2} f^1_{\sigma,inner}(y_\sigma,z,t,s) + ... \right\} \right|_{y_\sigma = \frac{x - q_\sigma(t)}{\epsilon^{1/2}}, z = \frac{x}{\epsilon}, s = \frac{t - t^*}{\epsilon^{1/2}}} \\
	&=\ r^\epsilon_{inner,0}(x,t)\ +\ \epsilon^{1/2}\ r^\epsilon_{inner,1}(x,t)\ +\ (\epsilon^{1/2})^2\ r^\epsilon_{inner,2}(x,t)\ +\ \dots\ +\ (\epsilon^{1/2})^m\ r^\epsilon_{inner,3}(x,t)\ +\ \dots
\end{split}
\end{equation}
}}
Here, $H(p)=-\frac12(p-i\partial_z)^2+V(z)$; see \eqref{eq:reduced_eigenvalue_problem}. 
\medskip

In the coming subsections we construct the functions $f^j_{\sigma,inner}$ so that $\psi^\epsilon_{app,inner}(x,t)$ satisfies the properties (P1), (P2) and (P3) of Proposition \ref{prop:idea_of_proof}.  

\subsubsection{Terms in $r^\epsilon_{inner}(x,t)$ with $L^2_x$ norm of order $\epsilon^0$}
The terms with $L^2_x$ norm proportional to $\epsilon^0=1$ in \eqref{eq:equation_for_r_inner} are of the form:
{\footnotesize{
\begin{equation} \label{eq:order_1_terms}
\begin{split}
 r^\epsilon_{inner,0}(x,t)\ =\ 	\epsilon^{-1/4} \sum_{\sigma = \pm} e^{i \left\{ S_\sigma(t) + \epsilon^{1/2} p_\sigma(t) y_\sigma \right\} / \epsilon} \left. \left[ \vphantom{\frac{1}{2}} H(p_\sigma(t)) - E_\sigma(p_\sigma(t)) \right] \vphantom{\frac{1}{2}} f^0_{\sigma,inner}(y_\sigma,z,t,s) \right|_{y_\sigma = \frac{x - q_\sigma(t)}{\epsilon^{1/2}}, z = \frac{x}{\epsilon}, s = \frac{t - t^*}{\epsilon^{1/2}}}.
\end{split}
\end{equation}
}}
We may set these two terms individually to zero by defining:
\begin{equation} \label{eq:f_0_inner}
	f^0_{\sigma,inner}(y_\sigma,z,t,s) = a^0_{\sigma,inner}(y_\sigma,t,s)\ \chi_\sigma(z;p_\sigma(t)),\quad
	\sigma = \pm\qquad .
\end{equation}
The functions $a^0_{\sigma,inner}(y_\sigma,t,s)$,\ $\sigma=\pm$ are left arbitrary for now and will be determined at a later stage. 

\subsubsection{Terms in $r^\epsilon_{inner}$ with $L^2_x$ norm of order $\epsilon^{1/2}$}
The terms with $L^2_x$ norm proportional to $\epsilon^{1/2}$ in \eqref{eq:equation_for_r_inner} are of the form $\epsilon^{1/2}$ times the following expression which is $O_{L^2}(1)$: 
\begin{equation} \label{eq:order_eps_one_half_terms}
\begin{split}
	& r^\epsilon_{inner,1}(x,t) = \epsilon^{-1/4} \sum_{\sigma = \pm} e^{i \left\{ S_\sigma(t) + \epsilon^{1/2} p_\sigma(t) y_\sigma \right\} / \epsilon}\left\{ \vphantom{\frac{1}{2}} \right.	\\
	&\left[ \vphantom{\frac{1}{2}} \left( \vphantom{\epsilon^{-1/2}} p_\sigma(t) - i \de_z - \de_p E_\sigma(p_\sigma(t)) \right)(- i \de_{y_\sigma}) - i \de_s \right] f^0_{\sigma,inner}(y_\sigma,z,t,s)   \\
	&\left. \left. + \left[ \vphantom{\frac{1}{2}} H(p_\sigma(t)) - E_\sigma(p_\sigma(t)) \right] \vphantom{\frac{1}{2}} f^1_{\sigma,inner}(y_\sigma,z,t,s) \right\} \right|_{y_\sigma = \frac{x - q_\sigma(t)}{\epsilon^{1/2}}, z = \frac{x}{\epsilon}, s = \frac{t - t^*}{\epsilon^{1/2}}}.
\end{split}
\end{equation}

\begin{proposition}
Substituting the expression \eqref{eq:f_0_inner} for $f^0_{\sigma,inner}(y_\sigma,z,t,s)$   into \eqref{eq:order_eps_one_half_terms} yields the following equivalent expression for  \eqref{eq:order_eps_one_half_terms}: 
\begin{equation} \label{eq:better_form}
\begin{split}
	&\epsilon^{-1/4} \sum_{\sigma = \pm} e^{i \left\{ S_\sigma(t) + \epsilon^{1/2} p_\sigma(t) y_\sigma \right\} / \epsilon}\left\{ \vphantom{\frac{1}{2}} \left[ \vphantom{\frac{1}{2}} H(p_\sigma(t)) - E_\sigma(p_\sigma(t)) \right] \left( \vphantom{\frac{1}{2}} f^1_{\sigma,inner}(y_\sigma,z,t,s) \right. \right.   \\
	&\left. \vphantom{\frac{1}{2}} - (- i \de_{y_\sigma}) a^0_{\sigma,inner}(y_\sigma,t,s) \de_{p_\sigma} \chi_\sigma(z;p_\sigma(t)) \right) \left. \left. \vphantom{\frac{1}{2}} - i \de_s a^0_{\sigma,inner}(y_\sigma,s) \chi_\sigma(z;p_\sigma(t)) \right\} \right|_{y_\sigma = \frac{x - q_\sigma(t)}{\epsilon^{1/2}}, z = \frac{x}{\epsilon}, s = \frac{t - t^*}{\epsilon^{1/2}}}.
\end{split}
\end{equation}
\end{proposition}
\begin{proof} Differentiating the eigenvalue problem \eqref{eq:reduced_eigenvalue_problem} satisfied by $(E_\sigma, \chi_\sigma)$ with respect to $p$, we obtain the following pair of identities for $\sigma =\pm$:
\begin{equation} \label{eq:derivative_identity}
\begin{split}
	 &\left( p_\sigma(t) - i \de_z - \de_{p_\sigma} E_\sigma(p_\sigma(t)) \right) \chi_\sigma(z;p_\sigma(t)) 	\\
	&= - \left( H(p_\sigma(t)) - E_\sigma(p_\sigma(t)) \right) \de_p \chi_\sigma(z;p_\sigma(t)).
\end{split}
\end{equation}
Relation \eqref{eq:better_form} now follows from substituting \eqref{eq:f_0_inner} into \eqref{eq:order_eps_one_half_terms} and using \eqref{eq:derivative_identity}.
\end{proof}
We may therefore set the expression in \eqref{eq:order_eps_one_half_terms} equal to zero by setting each term in the sum individually to zero. To do this, we first take:
\begin{equation} \label{eq:a_0_ind_of_s}
	\de_s a^0_{\sigma,inner}(y_\sigma,t,s) = 0,\qquad \sigma = \pm\quad .
\end{equation}
We then require,  for $\sigma \in \pm$,  that $f^1_{\sigma,inner}(y_\sigma,z,t,s)$ satisfy:
\begin{equation}
	\left[ \vphantom{\frac{1}{2}} H(p_\sigma(t)) - E_\sigma(p_\sigma(t)) \right] \left( \vphantom{\frac{1}{2}} f^1_{\sigma,inner}(y_\sigma,z,t,s) - (- i \de_{y_\sigma}) a^0_{\sigma,inner}(y_\sigma,t,s) \de_{p_\sigma} \chi_\sigma(z;p_\sigma(t)) \right) = 0
\end{equation}
Therefore, for $\sigma = \pm$, 
\begin{equation} \label{eq:f_1_inner}
	 f^1_{\sigma,inner}(y_\sigma,z,t,s) = a^1_{\sigma,inner}(y_\sigma,t,s) \chi_\sigma(z;p_\sigma(t)) + (- i \de_{y_\sigma}) a^0_{\sigma,inner}(y_\sigma,t) \de_p \chi_\sigma(z;p_\sigma(t)),
\end{equation}
where the functions $a^1_{\sigma,inner}(y_\sigma,t,s)$ are thus far arbitrary and to be determined. 

\subsubsection{Terms in $r^\epsilon_{inner}$ with $L^2_x$ norm of order $\epsilon^1$}
The terms in $r^\epsilon_{inner}$ with $L^2_x$ norm proportional to $\epsilon^1$ in \eqref{eq:equation_for_r_inner} are of the form:
$\epsilon$ times the following expression which is $O_{L^2}(1)$:
\begin{equation} \label{eq:epsilon_terms}
\begin{split}
	& r^\epsilon_{inner,2}(x,t) = \epsilon^{-1/4} \sum_{\sigma = \pm} e^{i \left\{ S_\sigma(t) + \epsilon^{1/2} p_\sigma(t) y_\sigma \right\} / \epsilon}\left\{\vphantom{\frac{1}{2}} \right.  \\
	&+ \left[ \frac{1}{2} (- i \de_{y_\sigma})^2 + \frac{1}{2} \de_x^2 W(q_\sigma(t)) y_\sigma^2 - i \de_t \right] f^0_{\sigma,inner}(y_\sigma,z,t,s) \\
	&+ \left[ \vphantom{\frac{1}{2}} \left( \vphantom{\epsilon^{-1/2}} p_\sigma(t) - i \de_z - \de_p E_\sigma(p_\sigma(t)) \right)(- i \de_{y_\sigma}) - i \de_s \right] f^1_{\sigma,inner}(y_\sigma,z,t,s)   \\
	&\left. \left. + \left[ \vphantom{\frac{1}{2}} H(p_\sigma(t)) - E_\sigma(p_\sigma(t)) \right] \vphantom{\frac{1}{2}} f^2_{\sigma,inner}(y_\sigma,z,t,s) \right\} \right|_{y_\sigma = \frac{x - q_\sigma(t)}{\epsilon^{1/2}}, z = \frac{x}{\epsilon}, s = \frac{t - t^*}{\epsilon^{1/2}}}.
\end{split}
\end{equation}
Recall (Proposition \ref{prop:idea_of_proof}, (P2)) that we must choose the $f^j_{\sigma,inner}$ in \eqref{eq:equation_for_r_inner} such that $\frac{1}{\epsilon} \inty{t^* - \epsilon^{\xi'}}{t^* + \epsilon^{\xi'}}{ \| r^\epsilon_{inner}(\cdot,t') \|_{L^2} }{t'} =o(\epsilon^{1/2})$ for all $t = (t^* - \epsilon^{\xi'},t^* + \epsilon^{\xi'})$. It follows that we need to choose the undetermined functions so that $r^\epsilon_{inner,2}(x,t)$ in \eqref{eq:epsilon_terms} satisfies:
\begin{equation}\label{eq:smallness_condition}
	\inty{t^* - \epsilon^{\xi}}{t^* + \epsilon^{\xi'}}{\| r^\epsilon_{inner,2}(\cdot,t)\|_{L^2}}{t} = o(\epsilon^{1/2}).
\end{equation}

In contrast to considerations at previous orders in $\epsilon^{1/2}$, we will not be able to satisfy \eqref{eq:smallness_condition} by choosing each summand of \eqref{eq:epsilon_terms} to satisfy the smallness condition \eqref{eq:smallness_condition}. To see this and to see how to proceed, we first simplify the expression \eqref{eq:epsilon_terms} using the expressions for $f^0_{\sigma,inner}(y_\sigma,z,t,s)$ \eqref{eq:f_0_inner} and $f^1_{\sigma,inner}(y_\sigma,z,t,s)$ \eqref{eq:f_1_inner} derived above. 

\begin{proposition}\label{yet-better}
The expression  \eqref{eq:epsilon_terms} may be written in the following form: 
\begin{equation} \label{eq:better_form_again}
\begin{split}
	& r^\epsilon_{inner,2}(x,t) \\
	& =\  \epsilon^{-1/4} \sum_{\sigma = \pm} e^{i \left\{ S_\sigma(t) + \epsilon^{1/2} p_\sigma(t) y_\sigma \right\} / \epsilon}\left\{ \vphantom{\frac{1}{2}} \right. \left[ \vphantom{\frac{1}{2}} H(p_\sigma(t)) - E_\sigma(p_\sigma(t)) \right] \left( \vphantom{\frac{1}{2}} f^2_{\sigma,inner}(y_\sigma,z,t,s) \right. \\
	&\left. - (- i \de_{y_\sigma}) a^1_{\sigma,inner}(y_\sigma,t,s) \de_{p_\sigma} \chi_\sigma(z;p_\sigma(t)) - \frac{1}{2} (- i \de_{y_\sigma})^2 a^0_{\sigma,inner}(y_\sigma,t) \de^2_{p_\sigma} \chi_\sigma(z;p_\sigma(t)) \vphantom{\frac{1}{2}} \right) \\
	&- i \de_s a^1_{\sigma,inner}(y_\sigma,t,s) \chi_\sigma(z;p_\sigma(t))
	 -\left[\ i \de_t  -\  \mathscr{H}_\sigma(t)\ \right]\ a^0_{\sigma,inner}(y_\sigma,t) \chi_\sigma(z;p_\sigma(t)) 	\\
	&+ i \de_{q_\sigma} W(q_\sigma(t)) a^0_{\sigma,inner}(y_\sigma,t,s) \ip{ \chi_{-\sigma}(\cdot;p_\sigma(t))}{\de_{p_\sigma} \chi_\sigma(\cdot;p_\sigma(t)) } \chi_{-\sigma}(z;p_\sigma(t)) \\
	&\left. \left. + i \de_{q_\sigma} W(q_\sigma(t)) a^0_{\sigma,inner}(y_\sigma,t,s) P^\perp_{\pm}(p_\sigma(t)) \de_{p_\sigma} \chi_\sigma(z;p_\sigma(t)) \vphantom{\frac{1}{2}} \right\} \right|_{y_\sigma = \frac{x - q_\sigma(t)}{\epsilon^{1/2}}, z = \frac{x}{\epsilon}, s = \frac{t - t^*}{\epsilon^{1/2}}}.
\end{split}
\end{equation}

Here, we recall that $H(p)=-\frac12(p-i\partial_z)^2+V(z)$ and $\mathscr{H}_\sigma(t)$  denotes the time-dependent harmonic oscillator Hamiltonian defined in \eqref{eq:envelope_equation}, where we replace $p(t), q(t), E_n, \chi_n, y$, respectively,  by $p_\sigma(t), q_\sigma(t), E_\sigma, \chi_\sigma, y_\sigma$. Finally,  $P^\perp_\pm(p_\sigma(t))$ denotes the orthogonal projection operator given by:
\begin{equation}
	P^\perp_\pm(p_\sigma(t)) f(z) := f(z) - \sum_{\sigma' = \pm} \ip{\chi_{\sigma'}(\cdot;p_\sigma(t))}{f(\cdot)} \chi_{\sigma'}(z;p_\sigma(t)).
\end{equation}
\end{proposition}

\begin{proof}
We begin with the identity, obtained by differentiating the eigenvalue problem \eqref{eq:reduced_eigenvalue_problem}, satisfied by the eigenpair $(E_\sigma, \chi_\sigma)$, \emph{twice} with respect to $p$: 
\begin{equation} \label{eq:second_derivative_identity}
\begin{split}
	&\frac{1}{2} \left( 1 - \de_p^2 E_\sigma(p_\sigma(t)) \right) \chi_\sigma(z;p_\sigma(t)) + \left( p_\sigma(t) - i \de_z - \de_{p_\sigma} E_\sigma(p_\sigma(t)) \right) \de_{p_\sigma} \chi_\sigma(z;p_\sigma(t)) 	\\
	&= - \frac{1}{2} \left( H(p_\sigma(t)) - E_\sigma(p_\sigma(t)) \right) \de^2_{p_\sigma} \chi_\sigma(z;p_\sigma(t)),
	\qquad \sigma = \pm\ .
\end{split}
\end{equation}
To obtain the expression \eqref{eq:better_form_again}, we  first  substitute  expression \eqref{eq:f_0_inner} for $f^0_{\sigma,inner}$  and expression \eqref{eq:f_1_inner} for $f^1_{\sigma,inner}$  into \eqref{eq:epsilon_terms}. We then  simplify using the identity \eqref{eq:second_derivative_identity} and the expansion of $\de_{p_\sigma} \chi_\sigma(z;p_\sigma(t))$ in terms of  its orthogonal components:
\begin{equation}
\begin{split}
	&\de_p \chi_\sigma(z;p_\sigma(t)) = 	\\
	&\sum_{\sigma' = \pm} \ip{\chi_{\sigma'}(\cdot;p_\sigma(t))}{\de_p \chi_\sigma(\cdot;p_\sigma(t))} \chi_{\sigma'}(z;p_\sigma(t)) + P^\perp_\pm(p_\sigma(t)) \chi_\sigma(z;p_\sigma(t)).
\end{split}
\end{equation}
\end{proof}
\medskip

By Proposition \ref{yet-better} the smallness condition \eqref{eq:smallness_condition} may be studied with the expression 
\eqref{eq:better_form_again} in place of \eqref{eq:epsilon_terms}.  We proceed in two steps. 
\begin{itemize}
\item[(A)] We first use certain degrees of freedom to eliminate `in-band' contributions to \eqref{eq:better_form_again}. \item[(B)] We will then be left with contributions which relate to the coupling of bands revealed in analysis of the breakdown of the single-band approximation.
\end{itemize}

{\it Step A:}  We first choose $a^0_{\sigma,inner}(y_\sigma,t)$  so that:
\begin{equation} \label{eq:envelope_equations}
i \de_t a^0_{\sigma,inner}(y_\sigma,t) = \mathscr{H}_\sigma(t) a^0_{\sigma,inner}(y_\sigma,t),\qquad \sigma=\pm\ .
\end{equation}

We also require that  $z\mapsto f^2_{\sigma,inner}(y_\sigma,z,t,s)$ be a $1-$ periodic solution of:
\begin{equation}\label{f2-eqn}
\begin{split}
	&\left[ \vphantom{\frac{1}{2}} H(p_\sigma(t)) - E_\sigma(p_\sigma(t)) \right] \left( \vphantom{\frac{1}{2}} f^2_{\sigma,inner}(y_\sigma,z,t,s) \right. \\
	&\left. - (- i \de_{y_\sigma}) a^1_{\sigma,inner}(y_\sigma,t,s) \de_{p_\sigma} \chi_\sigma(z;p_\sigma(t)) - \frac{1}{2} (- i \de_{y_\sigma})^2 a^0_{\sigma,inner}(y_\sigma,t) \de^2_{p_\sigma} \chi_\sigma(z;p_\sigma(t)) \vphantom{\frac{1}{2}} \right) \\
	&= - i \de_{q_\sigma} W(q_\sigma(t)) a^0_{\sigma,inner}(y_\sigma,t,s) P^\perp_{\pm}(p_\sigma(t)) \de_{p_\sigma} \chi_\sigma(z;p_\sigma(t))\ .
\end{split}
\end{equation}
Equation \eqref{f2-eqn} is solvable, with a uniform bound in time on the inverse, for all $t$ near $t^*$ by Corollary \ref{cor:away_from_crossing_regular}. Hence we have for $\sigma = \pm$:
\begin{equation} \label{eq:f_pm_inner}
\begin{split}
	&\ f^2_{\sigma,inner}(y_\sigma,z,t,s) = a^2_{\sigma,inner}(y_\sigma,t,s) \chi_{\sigma}(z;p_\sigma(t)) \\
	&+ (- i \de_{y_\sigma}) a^1_{\sigma,inner}(y_\sigma,s,t) \de_p \chi_\sigma(z;p_\sigma(t)) + \frac{1}{2} (- i \de_{y_\sigma})^2 a^0_{\sigma,inner}(y_\sigma,t) \de_p^2 \chi_\sigma(z;p_\sigma(t))	\\
	&- i \de_q W(q_\sigma(t))\ a^0_{\sigma,inner}(y_\sigma,t,s)\ \mathcal{R}_\sigma(p_\sigma(t)) P^\perp_\pm(p_\sigma(t)) \de_p \chi_\sigma(z;p_\sigma(t)).
\end{split}
\end{equation}
Here, $a^2_{\sigma,inner}(y_\sigma,t,s)$ is presently arbitrary and can be determined at higher order in $\epsilon^{1/2}$. \medskip

The initial data for equations \eqref{eq:envelope_equations} is fixed by the requirement that $\psi^\epsilon_{app,inner}$ satisfy (P1) of Proposition \ref{prop:idea_of_proof}. By inspection of the incoming solution, we see that this is equivalent to requiring that:
\begin{equation} \label{eq:requirement_for_P1}
	\text{for $t \in (t^* - \epsilon^{\xi'},t^* - \epsilon^\xi)$}: \text{ } a_{+,inner}^0(y,t) = a_+^0(y,t) \text{ and } a_{-,inner}(y,t) = 0.
\end{equation}
The only choice of initial data $a^0_{+,inner}(y,t^*)$, $a^0_{-,inner}(y,t^*)$ for \eqref{eq:envelope_equations} consistent with \eqref{eq:requirement_for_P1} are :
\[ a^0_{+,inner}(y,t^*) = a^0_+(y,t^*) = a^{0,*}(y),\quad a^0_{-,inner}(y,t^*) = 0\ .\]
Here, $a^{0,*}(y) := \lim_{t \uparrow t^*} a^0(y,t)$ \eqref{eq:defs_of_stars}. Indeed, for all $t \in (t^* - \epsilon^{\xi'},t^* + \epsilon^{\xi'})$: 
\begin{equation} \label{eq:inner_envelopes}
	a^0_{+,inner}(y_+,t) = a_+(y_+,t), \text{ and } a^0_{-,inner}(y_-,t) = 0.
\end{equation}
The choices \eqref{eq:envelope_equations}, \eqref{eq:f_pm_inner}, and \eqref{eq:inner_envelopes} simplify  \eqref{eq:better_form_again} to:
\begin{equation} \label{eq:better_form_again_simplified}
\begin{split}
	& r^\epsilon_{inner,2}(x,t) = \epsilon^{-1/4} e^{i \left\{ S_+(t) + \epsilon^{1/2} p_+(t) y_+ \right\} / \epsilon}\left\{ \vphantom{\frac{1}{2}} - i \de_s a^1_{+,inner}(y_+,t,s) \chi_+(z;p_+(t)) \right.	\\
	&\left. + i \de_{q_+} W(q_+(t)) a^0_{+,inner}(y_+,t,s) \ip{ \chi_{-}(\cdot;p_+(t))}{\de_{p_+} \chi_+(\cdot;p_+(t)) } \chi_{-}(z;p_+(t)) \vphantom{\frac{1}{2}} \right\}	\\ 
	&\left. \epsilon^{-1/4} e^{i \left\{ S_-(t) + \epsilon^{1/2} p_-(t) y_- \right\} / \epsilon}\left\{ \vphantom{\frac{1}{2}} - i \de_s a^1_{-,inner}(y_-,t,s) \chi_-(z;p_-(t)) \right\} \right|_{y_\sigma = \frac{x - q_\sigma(t)}{\epsilon^{1/2}}, z = \frac{x}{\epsilon}, s = \frac{t - t^*}{\epsilon^{1/2}}}.
\end{split}
\end{equation}
 
 We find that at this order in $\epsilon^{1/2}$ that there is no loss in taking the functions $a^1_{\sigma,inner}$, $\sigma = \pm$ to be independent of $t$:
 \[ a^1_{+,inner}(y_+,t,s)=a^1_{+,inner}(y_+,s),\qquad a^1_{-,inner}(y_-,t,s)=a^1_{-,inner}(y_-,s).\]
 From \eqref{eq:better_form_again_simplified} it is natural to set:
\begin{equation} \label{eq:a_+}
\begin{split}
	&\de_s a^1_{+,inner}(y_+,s) = 0 	\\
	&a^1_{+,inner}(y_+,0) = a^1_{+,inner,0}(y_+).
\end{split}
\end{equation}
and to choose $a^1_{-,inner}(y_+,s)$ to eliminate the projection of  $r^\epsilon_{inner,2}(x,t)$ onto  the vector $\chi_{-}(z;p_+(t))$. The function $a^1_{+,inner,0}(y_+)$ is at this point arbitrary, it will be fixed below by enforcing (P1) of Proposition \ref{prop:idea_of_proof}. 

Taking $a^1_{+,inner}(y_+,s)$ to satisfy \eqref{eq:a_+} reduces \eqref{eq:better_form_again_simplified} to the following:
\begin{equation} \label{eq:reduced_epsilon_terms}
\begin{split}
	& r^\epsilon_{inner,2}(x,t) = \epsilon^{-1/4} e^{i \left\{ S_+(t) + \epsilon^{1/2} p_+(t) y_+ \right\} / \epsilon} \left[ \vphantom{\ip{}{}_{L^2_x}} \vphantom{a^0_+} \right. \\
	&\left. i \de_q W(q_+(t)) a_{+}(y_+,t) \ip{\chi_-(\cdot;p_+(t))}{\de_p \chi_{+}(\cdot;p_+(t))} \chi_-(z;p_+(t)) \right]  \\
	&\left. + \epsilon^{-1/4} e^{i \left\{ S_-(t) + \epsilon^{1/2} p_-(t) y_- \right\} / \epsilon} \left[ - i \de_s a^1_{-,inner}(y_-,s) \chi_-(z;p_-(t)) \right] \vphantom{\frac{1}{2}} \right|_{y_\sigma = \frac{x - q_\sigma(t)}{\epsilon^{1/2}}, z = \frac{x}{\epsilon}, s = \frac{t - t^*}{\epsilon^{1/2}}}.
\end{split}
\end{equation}

We next determine the evolution of $a^1_{-,inner}(y_-,s)$ to satisfy the smallness condition \eqref{eq:smallness_condition}. 
We find it useful at this point to re-express functions of $t$ and $ y_+$ in terms of the variables $y_-$ and $ s$ using the relations:
\begin{equation}y_+ = y_- + \frac{ q_-(t) - q_+(t) }{\epsilon^{1/2}},\qquad t = t^* + \epsilon^{1/2} s.
\end{equation}
This yields:
\begin{equation} \label{eq:remaining_evaluated}
\begin{split}
	&r^\epsilon_{inner,2}(x,t^*+\epsilon^{1/2}s)\\
	& =\ \epsilon^{-1/4} e^{i \left\{ S_+(t^* + \epsilon^{1/2}s) + \epsilon^{1/2} p_+(t^* + \epsilon^{1/2}s) y_- + p_+(t^* + \epsilon^{1/2}s) ( q_-(t^* + \epsilon^{1/2}s) - q_+(t^* + \epsilon^{1/2} s) ) \right\} / \epsilon} \left[ \vphantom{\frac{1}{2}} \vphantom{ a^0_{+} } \right. \\
	&i \de_q W(q_+(t^* + \epsilon^{1/2} s)) a_{+}\left(y_- + \frac{ q_-(t^* + \epsilon^{1/2}s) - q_+(t^* + \epsilon^{1/2}s) }{\epsilon^{1/2}} ,t^* + \epsilon^{1/2}s\right) \\
	&\left. \times \ip{\chi_-(\cdot;p_+(t^* + \epsilon^{1/2} s))}{\de_p \chi_{+}(\cdot;p_+(t^* + \epsilon^{1/2}s))} \chi_-(z;p_+(t^* + \epsilon^{1/2}s)) \right]  \\
	&+ \epsilon^{-1/4} e^{i \left\{ S_-(t^* + \epsilon^{1/2}s) + \epsilon^{1/2} p_-(t^* + \epsilon^{1/2}s) y_- \right\} / \epsilon} \left[ \vphantom{a^1_{inner} \epsilon^{-1/2}} \right.	\\
	&\left. \left. - i \de_s a^1_{-,inner}(y_-,s) \chi_-(z;p_-(t^* + \epsilon^{1/2}s)) \right] \vphantom{\frac{1}{2}} \right|_{y_- = \frac{x - q_-(t^* + \epsilon^{1/2} s)}{\epsilon^{1/2}}, z = \frac{x}{\epsilon}}.
\end{split}
\end{equation}
In terms of $s$, the condition \eqref{eq:smallness_condition} reads: 
\begin{equation}\label{eq:smallness_condition_again}
	\int_{- \epsilon^{\xi' - 1/2}}^{\epsilon^{\xi' - 1/2}} \| r^\epsilon_{inner,2}(\cdot,t^*+\epsilon^{1/2}s)\|_{_{L^2}}  = o(1).
\end{equation}
We proceed with the construction of $a^1_{-,inner}(y_-,s)$ by seeking the  expression 
 in $r^\epsilon_{inner,2}(x,t^*+\epsilon^{1/2}s)$ which, to leading order, will be balanced (indeed cancelled out by) the term proportional to 
 $\de_s a^1_{-,inner}(y_-,s)$,  for $- \epsilon^{\xi' - 1/2}<s<\epsilon^{\xi' - 1/2}$\ ($0<\xi'<1/2$).

Thus we expand the expression for $r^\epsilon_{inner,2}(x,t^*+\epsilon^{1/2}s)$ in powers of $\epsilon^{1/2}s$,
 making use of the equations governing $(q_\pm(t),p_\pm(t))$ and $S_\pm(t)$ \eqref{eq:+_band_classical_system}, \eqref{eq:-_band_classical_system}, Definitions \ref{def:S_a_b_extensions} and \ref{def:-_band_functions} to compute their derivatives. We first Taylor-expand the expression within square brackets in \eqref{eq:remaining_evaluated}:
 \begin{equation}
\begin{split}
	&r^\epsilon_{inner,2}(x,t^*+\epsilon^{1/2}s)\\ 
	&=\ \epsilon^{-1/4} e^{i \left\{ S_+(t^* + \epsilon^{1/2}s) + \epsilon^{1/2} p_+(t^* + \epsilon^{1/2}s) y_- + p_+(t^* + \epsilon^{1/2}s) ( q_-(t^* + \epsilon^{1/2}s) - q_+(t^* + \epsilon^{1/2} s) ) \right\} / \epsilon} \left[ \vphantom{\frac{1}{2}} \vphantom{ a^0_{+} } \right. \\
	&\left. \vphantom{\frac{1}{2}} i \de_q W(q^*)\ a_{+}\left(y_- + [ \de_p E_-(p^*) - \de_p E_+(p^*) ] s ,t^*\right)\ \ip{\chi_-(\cdot;p^*)}{\de_p \chi_{+}(\cdot;p^*)} \chi_-(z;p^*) \right]  \\
	&+ \epsilon^{-1/4} e^{i \left\{ S_-(t^* + \epsilon^{1/2}s) + \epsilon^{1/2} p_-(t^* + \epsilon^{1/2}s) y_- \right\} / \epsilon} \left[ - i \de_s a^1_{-,inner}(y_-,s) \chi_-(z;p^*) \right] \left. \vphantom{\frac{1}{2}} \right|_{y_- = \frac{x - q_-(t^* + \epsilon^{1/2} s)}{\epsilon^{1/2}}, z = \frac{x}{\epsilon}} \\
	&+ O_{L^2_x}(\epsilon^{1/2} s, \epsilon^{1/2} s^2).
\end{split}
\end{equation}
Rearranging terms, we obtain:
\begin{equation}
\begin{split}
	&r^\epsilon_{inner,2}(x,t^*+\epsilon^{1/2}s)\\
	&=\ \epsilon^{-1/4} e^{i \left\{ S_-(t^* + \epsilon^{1/2}s) + \epsilon^{1/2} p_-(t^* + \epsilon^{1/2}s) y_- \right\} / \epsilon} i \chi_-(z;p^*)   \\
	&\times \left\{ \vphantom{\frac{1}{2}} e^{i \left\{ S_+(t^* + \epsilon^{1/2}s) - S_-(t^* + \epsilon^{1/2}s) + \epsilon^{1/2} ( p_+(t^* + \epsilon^{1/2}s) - p_-(t^* + \epsilon^{1/2}s) ) y_- + p_+(t^* + \epsilon^{1/2}s) ( q_-(t^* + \epsilon^{1/2}s) - q_+(t^* + \epsilon^{1/2} s) ) \right\} / \epsilon} \right. \\
	&\times \left[ \de_q W(q^*)\ a_{+}\left(y_- + [ \de_p E_-(p^*) - \de_p E_+(p^*) ] s ,t^*\right)\ \ip{\chi_-(\cdot;p^*)}{\de_p \chi_{+}(\cdot;p^*)} \right]  \\
	&\left. - \de_s a^1_{-,inner}(y_-,s) \vphantom{\frac{1}{2}} \right\} \left. \vphantom{\frac{1}{2}} \right|_{y_- = \frac{x - q_-(t^* + \epsilon^{1/2}s)}{\epsilon^{1/2}}, z = \frac{x}{\epsilon}}	\\
	&+ O_{L^2_x}(\epsilon^{1/2} s, \epsilon^{1/2} s^2).
\end{split}
\end{equation}
We next Taylor-expand the exponential:
\begin{equation}
\begin{split}
	&S_+(t^* + \epsilon^{1/2}s) - S_-(t^* + \epsilon^{1/2}s) 	\\
	&= ( \epsilon^{1/2} s ) p^* ( \de_p E_+(p^*) - \de_p E_-(p^*) ) 	\\
	&+ \frac{1}{2} ( \epsilon^{1/2} s )^2 \left( \de_q W(q^*) p^* ( \de_p^2 E_-(p^*) - \de_p^2 E_+(p^*) ) + \de_q W(q^*) ( \de_p E_-(p^*) - \de_p E_+(p^*) ) \right)	\\
	&+ O(\epsilon^{3/2} s^3) 	\\
	&\epsilon^{1/2} ( p_+(t^* + \epsilon^{1/2}s) - p_-(t^* + \epsilon^{1/2}s) ) y_- 	\\
	&= O(\epsilon^{3/2} s^2 y_-) 	\\
	&p_+(t^* + \epsilon^{1/2}s) ( q_-(t^* + \epsilon^{1/2}s) - q_+(t^* + \epsilon^{1/2} s) )	\\
	&= ( \epsilon^{1/2} s )( p^* ( \de_p E_-(p^*) - \de_p E_+(p^*) ) )	\\
	&+ \frac{1}{2} (\epsilon^{1/2} s)^2 \left( - 2 \de_q W(q^*) ( \de_p E_-(p^*) - \de_p E_+(p^*) ) + \de_q W(q^*) p^* (\de_p^2 E_+(p^*) - \de_p^2 E_-(p^*) ) \right)	\\
	&+ O(\epsilon^{3/2} s^3).
\end{split}
\end{equation}
Substituting these expressions and using the fact that $a_+(y,t) \in \mathcal{S}(\field{R})$ gives:
\begin{equation}
\begin{split}
	&r^\epsilon_{inner,2}(x,t^*+\epsilon^{1/2}s)\ =\ \epsilon^{-1/4} e^{i \left\{ S_-(t^* + \epsilon^{1/2}s) + \epsilon^{1/2} p_-(t^* + \epsilon^{1/2}s) y_- \right\} / \epsilon} i \chi_-(z;p^*)   \\
	&\times \left\{ e^{i \frac{1}{2} \de_q W(q^*) ( \de_p E_+(p^*) - \de_p E_-(p^*) ) s^2 } \right. \\
	&\times \left[ \vphantom{\frac{1}{2}} \de_q W(q^*)\ \ip{\chi_-(\cdot;p^*)}{\de_p \chi_{+}(\cdot;p^*)}\ a_{+}\left(y_- + [ \de_p E_-(p^*) - \de_p E_+(p^*) ] s ,t^*\right)\ 
	 \right]  \\
	&\left. - \de_s a^1_{-,inner}(y_-,s) \vphantom{\frac{1}{2}} \right\} \left. \vphantom{\frac{1}{2}} + O_{L^2_x}(\epsilon^{1/2} s, \epsilon^{1/2} s^2, \epsilon^{1/2} s^3) \right|_{y_- = \frac{x - q_-(t)}{\epsilon^{1/2}}, z = \frac{x}{\epsilon}, s = \frac{t - t^*}{\epsilon^{1/2}}}.
\end{split}
\end{equation}
It follows that by taking $a^1_{-,inner}(y_-,s)$ to satisfy:
\begin{equation}
\begin{split}
	&\de_s a^1_{-,inner}(y_-,s) = e^{i \frac{1}{2} \de_q W(q^*) ( \de_p E_+(p^*) - \de_p E_-(p^*) ) s^2 }  \\
	&\times \de_q W(q^*) a_{+}\left(y_- + [ \de_p E_-(p^*) - \de_p E_+(p^*) ] s ,t^*\right) \ip{\chi_-(\cdot;p^*)}{\de_p \chi_{+}(\cdot;p^*)}	\\
	&a^1_{-,inner}(y_-,0) = a^1_{-,inner,0}(y_-).
\end{split}
\end{equation}
We have that $\psi^\epsilon_{app,inner}(x,t)$ satisfies \eqref{eq:smallness_condition_again}, and therefore (P2) of Proposition \ref{prop:idea_of_proof}, provided  $3/8 < \xi'< 1/2$. That is, for  $3/8 < \xi' < 1/2$, we have
\begin{equation}
\begin{split}
	&\inty{- \epsilon^{\xi' - 1/2}}{\epsilon^{\xi' - 1/2}}{ \| O_{L^2_x}(\epsilon^{1/2}s, \epsilon^{1/2}s^2, \epsilon^{1/2}s^3) \|_{L^2_x} }{s} = O(\epsilon^{2 \xi'-1/2},\epsilon^{3\xi' - 1},\epsilon^{4\xi'-3/2})\ =\ o(1)\ .
\end{split}
\end{equation}

The initial data choices $a^1_{+,inner,0}(y_+)$ and $a^1_{-,inner,0}(y_-)$ are forced by the requirement that $\psi^\epsilon_{app,inner}(x,t)$ satisfies the matching condition (P1) of Proposition \ref{prop:idea_of_proof}. Since these terms appear at order $\epsilon^{1/2}$ in the asymptotic expansion, for (P1) to hold it is sufficient that for $s \in (- \epsilon^{\xi' - 1/2},t^* - \epsilon^{\xi - 1/2})$: $a^1_{+,inner}(y_+,s) - a^1_+(y_+,t^* + \epsilon^{1/2}s) = o_{L^2_{y_+}}(1)$ and $a^1_{-,inner}(y_-,s) = o_{L^2_{y_-}}(1)$. 

We claim that we may ensure this by taking:
\begin{equation}
	a^1_{+,inner,0}(y_+) = a^1_+(y_+,t^*) = a^{1,*}(y_+),
\end{equation}
(recall: $a^{1,*}(y) := \lim_{t \uparrow t^*} a^1(y,t)$ \eqref{eq:defs_of_stars}), and:
\begin{equation} \label{eq:choice_for_a_-_0}
\begin{split}
	&a^1_{-,inner,0}(y_-) =\   \de_q W(q^*)\times \ip{\chi_-(\cdot;p^*)}{\de_p \chi_{+}(\cdot;p^*)}\\
&	\times  \int_{-\infty}^{0} e^{i \frac{1}{2} \de_q W(q^*) ( \de_p E_+(p^*) - \de_p E_-(p^*) ) (s')^2 } 	 a_{+}\left(y_- + [ \de_p E_-(p^*) - \de_p E_+(p^*) ] s' ,t^*\right) \text{d}s'.
\end{split}
\end{equation}
This claim follows from Taylor-expansion: 
\begin{equation}
\begin{split}
	&\text{for all $s \in (- \epsilon^{\xi' - 1/2},t^* - \epsilon^{\xi - 1/2})$}: 	\\
	&a^1_+(y_+,t^* + \epsilon^{1/2}s) - a^1_+(y_+,t^*) = O_{L^2_{y_+}}(\epsilon^{1/2}s) = O_{L^2_{y_+}}(\epsilon^{\xi'}) = o_{L^2_{y_+}}(1) 
\end{split}
\end{equation}
since $3/8 < \xi' < 1/2$ and from integration by parts, which shows that for $s \in (- \epsilon^{\xi'-1/2},-\epsilon^{\xi-1/2})$: 
\begin{equation}
\begin{split}
	&a^1_{-,inner}(y_-,s) = \int_{-\infty}^{s} e^{i \frac{1}{2} \de_q W(q^*) ( \de_p E_+(p^*) - \de_p E_-(p^*) ) (s')^2 } 	\\
	&\times \de_q W(q^*) a_{+}\left(y_- + [ \de_p E_-(p^*) - \de_p E_+(p^*) ] s' ,t^*\right) \ip{\chi_-(\cdot;p^*)}{\de_p \chi_{+}(\cdot;p^*)} \text{d}s'	\\
	&= O_{L^2_{y_-}}(\epsilon^{1/2 - \xi}).
\end{split}
\end{equation}
Since $\xi < 1/2$ by assumption, we are done.

 It remains to show (P3) of Proposition \ref{prop:idea_of_proof}. But by an identical argument, for $s \in (\epsilon^{\xi-1/2},\epsilon^{\xi'-1/2})$:
\begin{equation}
\begin{split}
	&a^1_{-,inner}(y_-,s) = \de_q W(q^*) \times \ip{\chi_-(\cdot;p^*)}{\de_p \chi_{+}(\cdot;p^*)} \\
	& \int_{-\infty}^{\infty} e^{i \frac{1}{2} \de_q W(q^*) ( \de_p E_+(p^*) - \de_p E_-(p^*) ) (s')^2 } \  a_{+}\left(y_- + [ \de_p E_-(p^*) - \de_p E_+(p^*) ] s' ,t^*\right) \text{d}s'	\\
	&+ O(\epsilon^{\xi' - 1/2}),
\end{split}
\end{equation}
so that for $\sigma = \pm$: 
\begin{equation}
	a^1_{\sigma,inner}(y_\sigma,t^* + \epsilon^{1/2} s) - a^1_\sigma(y_\sigma,t^*) = O_{L^2_{y_\sigma}}(\epsilon^{\xi'}) = o_{L^2_{y_\sigma}}(1) 
\end{equation}
for $t-t^*=\epsilon^{1/2}s \in (\epsilon^{\xi},\epsilon^{\xi'})$. It follows that $\psi^\epsilon_{app,inner}(x,t)$ so constructed satisfies all hypotheses of Proposition \ref{prop:idea_of_proof}, and so the proof of Theorem \ref{th:band_crossing_theorem} is complete.

\begin{appendices}
\section{Consistency of Theorem \ref{th:band_crossing_theorem} with `Landau-Zener' theory} \label{app:consistency} 
In this Appendix we:
\begin{enumerate}
\item Derive the precise magnitude (in $L^2$) of the `excited' wavepacket from the main statement of Theorem \ref{th:band_crossing_theorem} \eqref{eq:statement_of_main_theorem} in terms of that of the initial, `incident', wavepacket (Section \ref{sec:consistency_part1})
\item Show how this result is consistent with the solution of a simplified `Landau-Zener'-type model (Section \ref{sec:consistency_part2}). 
\end{enumerate}
\subsection{Magnitude of wave `excited' at the crossing} \label{sec:consistency_part1}
Let $t \in (t^*,\tilde{T}]$. According to the main statement of Theorem \ref{th:band_crossing_theorem} \eqref{eq:statement_of_main_theorem}, the wave `excited' at the crossing time $t^*$ is of the form: 
\begin{equation} \label{eq:def_excited_WP}
\begin{split}
	&\text{WP}^\epsilon_{\text{excited}}(x,t) := 	\\
	&\sqrt\epsilon \; \text{WP}^{0,\epsilon}[\mathfrak{S}_-(t),\mathfrak{q}_-(t),\mathfrak{p}_-(t),\mathfrak{a}^0_-(y,t),\mathfrak{X}_-(z;\mathfrak{p}_-(t))](x,t) + o_{L^2_x}(\sqrt{\epsilon})
\end{split}
\end{equation}
where the parameters of the excited wavepacket: $\mathfrak{S}_-(t)$, $\mathfrak{q}_-(t)$, $\mathfrak{p}_-(t)$, $\mathfrak{a}^0_-(y,t)$, $\mathfrak{X}_-(z;\mathfrak{p}_-(t))$ were defined in Proposition \ref{prop:smooth_continuation_again} and Definition \ref{def:-_band_functions}. The notation $\text{WP}^{0,\epsilon}[...](x,t)$ was defined in \eqref{eq:WP_notation_0}. 

%We seek an expression for the magnitude (in $L^2$) of $\text{WP}^\epsilon_{\text{excited}}$ in terms of that of the `incident' wavepacket, which has the form:
%\begin{equation} \label{eq:def_incident_WP}
%\begin{split}
	%&\text{WP}^\epsilon_{\text{incident}}(x,t) := 	\\
	%&\text{WP}^{1,\epsilon}[\mathfrak{S}_+(t),\mathfrak{q}_+(t),\mathfrak{p}_+(t),\mathfrak{a}^0_+(y,t),\mathfrak{X}_+(z;\mathfrak{p}_+(t))](x,t) + o_{L^2_x}(\epsilon)
%\end{split}
%\end{equation}
%where the parameters of the incident wavepacket: $\mathfrak{S}_+(t)$, $\mathfrak{q}_+(t)$, $\mathfrak{p}_+(t)$, $\mathfrak{a}^0_+(y,t)$, $\mathfrak{X}_+(z;\mathfrak{p}_+(t))$ were defined in Proposition \ref{prop:smooth_continuation} and Definition \ref{def:S_a_b_extensions}, and the notation $\text{WP}^{1,\epsilon}[...](x,t)$ was definited in \eqref{eq:WP_notation_1}. 

Expanding the short-hand notation, \eqref{eq:def_excited_WP} becomes:
\begin{equation} \label{eq:expanded_excited_WP}
\begin{split}
	&\text{WP}_{\text{excited}}^\epsilon(x,t) =	\\
	&\sqrt{\epsilon} \times \epsilon^{-1/4} e^{ i \mathfrak{S}_-(t) / \epsilon } e^{i \mathfrak{p}_-(t) ( x - \mathfrak{q}_-(t) ) / \epsilon} \mathfrak{a}^0_-\left( \frac{x - \mathfrak{q}_-(t)}{\sqrt\epsilon}, t \right) \mathfrak{X}_-\left(\frac{x}{\epsilon};\mathfrak{p}_-(t)\right) + o_{L^2_x}(\sqrt{\epsilon}). 
\end{split}
\end{equation}
Using 1-periodicity in $z$ of the function $\mathfrak{X}_-(z;\mathfrak{p}_-(t))$ for every $t$ and Lemma 4.1 of \cite{2017WatsonWeinsteinLu} we have that:
\begin{equation}
	\| \text{WP}_{\text{excited}}^\epsilon(\cdot,t) \|^2_{L^2} = \epsilon \; \| \mathfrak{a}^0_-(\cdot,t) \|^2_{L^2} \| \mathfrak{X}_-(\cdot;\mathfrak{p}_-(t)) \|^2 + o(\epsilon).  
\end{equation}
Using normalization of $\mathfrak{X}_-(\cdot;\mathfrak{p}_-(t))$ ($\| \mathfrak{X}_-(\cdot;\mathfrak{p}_-(t)) \| = 1$ for every $t$), and preservation of the $L^2_y$-norm by the equations defining $\mathfrak{a}_-(y,t)$ (recall Definition \ref{def:S_a_b_extensions} and \eqref{eq:envelope_equation}), we have that: 
\begin{equation} \label{eq:excited_and_incident_waves_in_terms_envelopes}
	\| \text{WP}_{\text{excited}}^\epsilon(\cdot,t) \|^2_{L^2} = \epsilon \; \| a^{0}_-(\cdot,t^*) \|^2_{L^2} + o(\epsilon). 
\end{equation}
%A similar manipulation yields that:
%\begin{equation}
	%\| \text{WP}_{\text{incident}}^\epsilon(\cdot,t) \|^2_{L^2} = \| \mathfrak{a}^0_+(\cdot,t) \|^2_{L^2} + o(1).
%\end{equation}
%Using preservation of the $L^2_y$-norm by the equations defining $\mathfrak{a}_+(y,t)$ and $\mathfrak{a}_-(y,t)$ (recall Definitions \ref{def:S_a_b_extensions} and \ref{def:-_band_functions} and \eqref{eq:envelope_equation}) we have that: 
%\begin{equation} \label{eq:excited_and_incident_waves_in_terms_envelopes}
	%\| \text{WP}_{\text{excited}}^\epsilon(\cdot,t) \|^2_{L^2} = \epsilon \; \| a^{0}_-(\cdot,t^*) \|^2_{L^2} + o(\epsilon), \qquad \| \text{WP}_{\text{incident}}^\epsilon(\cdot,t) \|^2_{L^2} = \| a^0_0(\cdot) \|^2_{L^2} + o(1).
%\end{equation}
%We now derive an expression for $\| \text{WP}_{\text{excited}}^\epsilon(\cdot,t) \|_{L^2}^2$ in terms of $\| \text{WP}_{\text{incident}}^\epsilon(\cdot,t) \|^2_{L^2}$ using the precise for of $a^0_-(y,t^*)$ given in \eqref{eq:L_Z_connection}. 
The precise form of $a^0_-(y,t^*)$ was derived in Section \ref{sec:inner_solution} and is displayed in \eqref{eq:L_Z_connection}. Taking the $L^2$-norm and squaring both sides of this expression we have:
\begin{equation} \label{eq:estimate_L_Z_connection}
\begin{split}
	&\| {a}^0_-(\cdot,t^*) \|^2_{L^2} = | \de_q W(q^*) |^2 \times |\ip{\chi_-(\cdot;p^*)}{\de_p \chi_+(\cdot;p^*)} |^2  	\\
	&\times \left\| \int_{-\infty}^{\infty} e^{ i [\de_q W(q^*)][\de_p E_+(p^*) - \de_p E_-(p^*)] \tau^2 / 2 }
	\times  {a}^{0,*}(y - [\de_p E_+(p^*) - \de_p E_-(p^*)] \tau) \, \mathrm{{d}}\tau\ \right\|^2_{L^2_y}.
\end{split}
\end{equation}
By changing variables (recall Property \ref{band_crossing_assumption} (A4) that $\de_p E_+(p^*) - \de_p E_-(p^*) \neq 0$) we can write the integral as a convolution: 
\begin{equation} \label{eq:estimate_L_Z_connection_again}
\begin{split}
	&\| {a}^0_-(\cdot,t^*) \|^2_{L^2} = \frac{ | \de_q W(q^*) |^2 \times |\ip{\chi_-(\cdot;p^*)}{\de_p \chi_+(\cdot;p^*)} |^2 }{ | \de_p E_+(p^*) - \de_p E_-(p^*) |^2 }  	\\
	&\times \left\| \int_{-\infty}^\infty e^{ \frac{ i [\de_q W(q^*)] \tau^2 }{ 2 [\de_p E_+(p^*) - \de_p E_-(p^*)] } } \times {a}^{0,*}(y - \tau) \, \mathrm{{d}}\tau\ \right\|^2_{L^2_y}.
\end{split}
\end{equation}
Applying Plancharel's theorem and then the convolution theorem yields that:
\begin{equation} \label{eq:estimate_L_Z_connection_again_again}
\begin{split}
	&\| {a}^0_-(\cdot,t^*) \|^2_{L^2} = \frac{ | \de_q W(q^*) |^2 \times |\ip{\chi_-(\cdot;p^*)}{\de_p \chi_+(\cdot;p^*)} |^2 }{ | \de_p E_+(p^*) - \de_p E_-(p^*) |^2 }  	\\
	&\times \left\| \mathcal{F}_\tau \left\{ e^{ \frac{ i [\de_q W(q^*)] \tau^2 }{ 2 [\de_p E_+(p^*) - \de_p E_-(p^*)] } } \right\}(\xi) \times \mathcal{F}_y \left\{ {a}^{0,*}(y) \right\}(\xi) \right\|^2_{L^2_\xi},
\end{split}
\end{equation}
where $\mathcal{F}_\tau$ denotes the Fourier transform with respect to the variable $\tau$:
\begin{equation}
	\mathcal{F}_\tau \{ f(\tau) \}(\xi) := \inty{- \infty}{\infty}{ e^{- 2 \pi i \xi \tau} f(\tau) }{\tau}.
\end{equation}
A standard calculation shows that:
\begin{equation} \label{eq:gaussian_FT}
\begin{split}
	&\mathcal{F}_\tau \left\{ e^{ \frac{ i [\de_q W(q^*)] \tau^2 }{ 2 [\de_p E_+(p^*) - \de_p E_-(p^*)] } } \right\}(\xi) =  \\
	&\frac{ \sqrt{\pi} (1 + i) \sqrt{ [\de_p E_+(p^*) - \de_p E_-(p^*)] }}{ \sqrt{ [\de_q W(q^*)] } } e^{\frac{ - i 2 \pi^2 [\de_p E_+(p^*) - \de_p E_-(p^*)] \xi^2 }{ [\de_q W(q^*)] } }. 
\end{split}
\end{equation}
Substituting \eqref{eq:gaussian_FT} into \eqref{eq:estimate_L_Z_connection_again_again} and again using Plancharel's theorem yields:
\begin{equation} \label{eq:L_Z_con}
	\| {a}^0_-(\cdot,t^*) \|^2_{L^2} = \frac{ 2 \pi | \de_q W(q^*) | |\ip{\chi_-(\cdot;p^*)}{\de_p \chi_+(\cdot;p^*)} |^2 \| a^{0,*}(\cdot) \|^2_{L^2} }{ | \de_p E_+(p^*) - \de_p E_-(p^*) | }.
\end{equation}
Recall that $a^{0,*}(y)$ is defined as the limit: $a^{0,*}(y)$ $:= \lim_{t \uparrow t^*} a^0(y,t)$ where $a^0(y,t)$ is the leading-order envelope of the incident wavepacket and satisfies \eqref{eq:envelope_equation}. By the $L^2_y$-norm conserving property of \eqref{eq:envelope_equation}, we have that:
\begin{equation} \label{eq:a_star_equal_a_0}
	\| a^{0,*}(\cdot) \|^2_{L^2} = \lim_{t \uparrow t^*} \| a^0(\cdot,t) \|_{L^2} = \| a^0(\cdot,0) \|^2_{L^2} = \| a^0_0(\cdot) \|^2_{L^2}.
\end{equation}
Substituting \eqref{eq:a_star_equal_a_0} into \eqref{eq:L_Z_con} we obtain:
\begin{equation} \label{eq:excited_wv_magnitude}
	\| a^{0}_-(\cdot,t^*) \|^2_{L^2} = \frac{ 2 \pi | \de_q W(q^*) | |\ip{\chi_-(\cdot;p^*)}{\de_p \chi_+(\cdot;p^*)} |^2 \| a_0^0(\cdot) \|^2_{L^2} }{ | \de_p E_+(p^*) - \de_p E_-(p^*) | }.
\end{equation}
Combining \eqref{eq:excited_wv_magnitude} with \eqref{eq:excited_and_incident_waves_in_terms_envelopes} we have that:
\begin{equation} \label{eq:LZ}
\begin{split}
	&\| \text{WP}_{\text{excited}}^\epsilon(\cdot,t) \|^2_{L^2} = 	\\
	&\qquad \frac{ 2 \pi | \de_q W(q^*) | |\ip{\chi_-(\cdot;p^*)}{\de_p \chi_+(\cdot;p^*)} |^2 \| a^0_0(\cdot) \|^2_{L^2} }{ | \de_p E_+(p^*) - \de_p E_-(p^*) | } \times \epsilon + o(\epsilon).
\end{split}
\end{equation}
For the purposes of comparing our result with `Landau-Zener' theory, we seek an expression for the magnitude of the excited wavepacket \eqref{eq:LZ} in terms of that of the initial wavepacket. First, recall that the initial wavepacket has the form \eqref{eq:bloch_wavepacket_initial_data}:
\begin{equation}
	\text{WP}_{\text{initial}}^\epsilon(x) := \text{WP}^{1,\epsilon}[S_0,q_0,p_0,a^0_0(y),a^1_0(y),\chi_n(z;p_0)](x). 
\end{equation}
By an identical calculation to that given above \eqref{eq:expanded_excited_WP}-\eqref{eq:excited_and_incident_waves_in_terms_envelopes}, we see that:
\begin{equation}
	\| \text{WP}_{\text{initial}}^\epsilon(\cdot) \|^2_{L^2} = \| a^0_0(\cdot) \|^2_{L^2} + o(\sqrt\epsilon),
\end{equation}
and hence:
\begin{equation} \label{eq:LZ_again}
\begin{split}
	&\| \text{WP}_{\text{excited}}^\epsilon(\cdot,t) \|^2_{L^2} = 	\\
	&\qquad \frac{ 2 \pi | \de_q W(q^*) | |\ip{\chi_-(\cdot;p^*)}{\de_p \chi_+(\cdot;p^*)} |^2 }{ | \de_p E_+(p^*) - \de_p E_-(p^*) | } \times \epsilon \times \| \text{WP}^\epsilon_{\text{initial}} (\cdot) \|^2_{L^2} + o(\epsilon).
\end{split}
\end{equation}
Using $\dot{p}_+(t) = - \de_q W(q_+(t))$ \eqref{eq:+_band_classical_system}, we may re-write \eqref{eq:LZ_again} in the equivalent form: 
\begin{equation} \label{eq:LZ_time_derivatives}
\begin{split}
	&\| \text{WP}_{\text{excited}}^\epsilon(\cdot,t) \|^2_{L^2} = 	\\
	&\qquad \left. \frac{ 2 \pi |\ip{\chi_-(\cdot;p_+(\tau))}{\fdf{t} \chi_+(\cdot;p_+(\tau))} |^2 }{ | \fdf{\tau} \left( E_+(p_+(\tau)) - E_-(p_+(\tau)) \right) | } \right|_{\tau = t^*} \times \epsilon \times \| \text{WP}_{\text{initial}}^\epsilon(\cdot) \|^2_{L^2} + o(\epsilon)\ .
\end{split}
\end{equation}
In the next subsection we shall compare the expression \eqref{eq:LZ_time_derivatives} with 
 \eqref{eq:inter_band_transition_prob} of a Landau-Zener type model.

\subsection{Agreement with `Landau-Zener' theory} \label{sec:consistency_part2} 
In this section we show that the expression \eqref{eq:LZ_time_derivatives} may be derived from a Landau-Zener type model, a simple finite dimensional model which incorporates key spectral features of the exact problem.
 We consider the following Schr\"{o}dinger equation: 
\begin{equation} \label{eq:abstract_schro}
	i \epsilon \dot{\psi}^\epsilon = H(t) \psi^\epsilon, \quad t \geq 0. 
\end{equation}
Here, $H(t)$ denotes a time-dependent Hamiltonian whose spectrum consists of two time-dependent eigenvalue `bands' which are degenerate at some $t^* > 0$. Throughout this section, dots will denote derivatives with respect to $t$ hence $\dot{f}(t) := df(t)/dt$. More precisely, we assume the following on the operator $H(t)$:
\begin{equation}
	\text{Spec}[H(t)] = E_+(t) \cup E_-(t),
\end{equation}
where $E_+(t), E_-(t)$ denote real, smooth functions which satisfy:
\begin{equation}
\begin{split}
	&E_+(t) < E_-(t) \text{ for } t \in [0,t^*), \quad E_-(t) < E_+(t) \text{ for } t > t^*,	\\
	&\text{and } E_+(t^*) = E_-(t^*).
\end{split}
\end{equation}
We assume further that the crossing is \emph{linear} in the sense that: 
\begin{equation} \label{eq:linear_assumption}
	\left. \left( \dot{E}_+(t) - \dot{E}_-(t) \right) \right|_{t = t^*} > 0. 
\end{equation}
Let $\chi_\sigma(t)$, $\sigma = \pm$ denote smoothly varying and orthonormal eigenfunctions corresponding to the eigenvalues $E_\sigma(t)$, so that:
\begin{equation}
	H(t) \chi_\sigma(t) = E_\sigma(t) \chi_\sigma(t), \quad t \geq 0, 
\end{equation}
and $\ip{\chi_\sigma(t)}{\chi_{\sigma'}(t)} = \delta_{\sigma \sigma'}$. We may assume without loss of generality (by making an appropriate gauge transformation) that:
\begin{equation}
	\ip{ \chi_\sigma(t) }{ \dot{\chi}_\sigma(t) } = 0, \quad t \geq 0.
\end{equation}
This choice of gauge is sometimes known as the \emph{adiabatic gauge} (see Remark \ref{rem:adiabatic_gauge}). We now seek a solution of \eqref{eq:abstract_schro} of the form:
\begin{equation} \label{eq:ansatz}
	\psi^\epsilon(t) = \sum_{\sigma = \pm} c_\sigma(t) \chi_\sigma(t) e^{ \frac{ - i \inty{t^*}{t}{ E_\sigma(\tau) }{\tau} }{ \epsilon } }.
\end{equation}
Substituting \eqref{eq:ansatz} into \eqref{eq:abstract_schro} and projecting onto the subspaces spanned by each eigenfunction yields a system for the coefficients $c_\sigma(t)$, $\sigma = \pm$, which is equivalent to the original PDE \eqref{eq:abstract_schro}: 
\begin{equation} \label{eq:c_system}
\begin{split}
	&\fd{c_+}{t}(t) = \ip{\chi_+(t)}{\dot{\chi}_-(t)} e^{ \frac{ i \inty{t^*}{t}{ E_+(\tau) - E_-(\tau) }{\tau} }{ \epsilon } } c_-(t)	\\
	&\fd{c_-}{t}(t) = \ip{\chi_-(t)}{\dot{\chi}_+(t)} e^{ \frac{ i \inty{t^*}{t}{ E_-(\tau) - E_+(\tau) }{\tau} }{ \epsilon } } c_+(t).
\end{split}
\end{equation}
By differentiating $\ip{\chi_+(t)}{\chi_-(t)} = 0$ we obtain the identity:
\begin{equation}
	\ip{\chi_+(t)}{\dot{\chi}_-(t)} = - \overline{ \ip{\chi_-(t)}{ \dot{\chi}_+(t) } }, \quad t \geq 0.
\end{equation}
We may then write the system \eqref{eq:c_system} in the compact form:
\begin{equation} \label{eq:system_with_F}
\begin{split}
	&\fd{c_+}{t}(t) = F^\epsilon(t) c_-(t), \quad \fd{c_-}{t}(t) = - \overline{ F^\epsilon(t) } c_+(t),	\\
	&F^\epsilon(t) := - \overline{ \ip{ \chi_-(t) }{ \dot{\chi}_+(t) } } e^{ \frac{ i \inty{t^*}{t}{ E_+(\tau) - E_-(\tau) }{ \tau } }{ \epsilon } }.  
\end{split}
\end{equation}
Integrating \eqref{eq:system_with_F} in time gives:
\begin{equation} \label{eq:c_system_again}
\begin{split}
	&c_+(t) = c_+(0) + \inty{0}{t}{ F^\epsilon(\tau) c_-(\tau) }{\tau}	\\
	&c_-(t) = c_-(0) - \inty{0}{t}{ \overline{ F^\epsilon(\tau) } c_+(\tau) }{\tau}.
\end{split}
\end{equation}
At this point, fix $t > t^*$. A standard calculation using oscillations in $F^\epsilon$ \eqref{eq:system_with_F} and the linear crossing assumption \eqref{eq:linear_assumption} implies that: 
\begin{equation}
\begin{split}
	&c_+(t) = c_+(0) - \left. \frac{ \sqrt{\pi} (1 + i) \overline{ \ip{ \chi_-(\tau) }{ \dot{\chi}_+(\tau) } } c_-(\tau) }{ \sqrt{ \left( \dot{E}_+(\tau) - \dot{E}_-(\tau) \right) } } \right|_{\tau = t^*} \times \sqrt{\epsilon} + o(\sqrt{\epsilon})	\\
	&c_-(t) = c_-(0) + \left. \frac{ \sqrt{\pi} (1 - i) \ip{ \chi_-(\tau) }{ \dot{\chi}_+(\tau) } c_+(\tau) }{ \sqrt{ \left( \dot{E}_+(\tau) - \dot{E}_-(\tau) \right) } } \right|_{\tau = t^*} \times \sqrt{\epsilon} + o(\sqrt{\epsilon}).
\end{split}
\end{equation}
A similar calculation yields that $| c_\sigma(t^*) - c_\sigma(0) | = o(1)$, $\sigma = \pm$, and hence:
\begin{equation}
\begin{split}
	&c_+(t) = c_+(0) - \left. \frac{ \sqrt{\pi} (1 + i) \overline{ \ip{ \chi_-(\tau) }{ \dot{ \chi }_+(\tau) } } }{ \sqrt{ \left( \dot{E}_+(\tau) - \dot{E}_-(\tau) \right) } } \right|_{\tau = t^*} \times \sqrt{\epsilon} \times c_-(0) + o(\sqrt{\epsilon})	\\
	&c_-(t) = c_-(0) + \left. \frac{ \sqrt{\pi} (1 - i) \ip{ \chi_-(\tau) }{ \dot{\chi}_+(\tau) } }{ \sqrt{ \left( \dot{E}_+(\tau) - \dot{E}_-(\tau) \right) } } \right|_{\tau = t^*} \times \sqrt{\epsilon} \times c_+(0) + o(\sqrt{\epsilon}).
\end{split}
\end{equation}
To compare with \eqref{eq:LZ_time_derivatives}, we now set $c_-(0) = 0$ and square the absolute value of the expression for $c_-(t)$ to obtain:
\begin{equation} \label{eq:inter_band_transition_prob}
	| c_-(t) |^2 = \left. \frac{ 2 \pi | \ip{ \chi_-(\tau) }{ \dot{\chi}_+(\tau) } |^2 }{ \left| \left( \dot{E}_+(\tau) - \dot{E}_-(\tau) \right) \right| } \right|_{\tau = t^*} \times \epsilon \times | c_+(0) |^2 + o(\epsilon).
\end{equation}
When the initial data is normalized: $|c_+(0)|^2 = 1$, expressions \eqref{eq:inter_band_transition_prob} and \eqref{eq:LZ_time_derivatives} may be interpreted as the `inter-band transition probability'.

\section{Proof that the `inter-band coupling coefficient' vanishes for trivial crossings} \label{app:coupling_coefficient_zero}

\subsection{Formula for $\ip{\chi_-(\cdot;p^*)}{\de_p \chi_+(\cdot;p^*)}$ from symmetry of Bloch band}
Let $E(p)$, $\chi(z;p)$ denote an eigenpair of \eqref{eq:reduced_eigenvalue_problem}. Then:
\begin{equation}
\begin{split}
	&H(2\pi - p) e^{- 2 \pi i z} \overline{ {\chi}(z;p) } = e^{- 2 \pi i z} H(- p) \overline{ {\chi}(z;p) }  	\\
	&= e^{- 2 \pi i z} \overline{ H(p) {\chi}(z;p) } = {E}(p) e^{- 2 \pi i z} \overline{\chi(z;p)}	\\
	&e^{- 2 \pi i (z + 1)} \overline{ {\chi}(z + 1;p) } = e^{- 2 \pi i z} \overline{ {\chi}(z;p) }.
\end{split}
\end{equation}
Hence, for $p \in \mathcal{B}$ such that the eigenvalue $E(p)$ is non-degenerate, ${E}(p)$ and ${\chi}(z;p)$ obey the symmetry (after possibly multiplying $\chi(z;p)$ by a constant): 
\begin{align} \label{eq:symmetree}
	&{E}(2 \pi - p) = {E}(p) && {\chi}(z;2\pi - p) = e^{- 2 \pi i z} \overline{{\chi}(z;p)}.
\end{align}
Note further that the symmetry \eqref{eq:symmetree} implies that the eigenvalue $E(2\pi - p)$ is non-degenerate if and only if $E(p)$ is; if it were not, we could use \eqref{eq:symmetree} to generate two linearly independent eigenfunctions with eigenvalue $E(p)$ from those with eigenvalue $E(2\pi - p)$. 

Now, let $E_{n}(p), E_{n+1}(p)$ denote eigenvalue bands of \eqref{eq:reduced_eigenvalue_problem} which cross at $p = \pi$ (Without loss of generality: the case where the crossing takes place at $p = 0$ is similar and these are the only possibilities (Theorem \ref{th:all_crossings_smooth})), and fix the Brillouin zone: $\mathcal{B} = [0,2\pi]$. Let $E_+(p)$, $E_-(p)$ and $\chi_+(z;p)$, $\chi_-(z;p)$ denote the smooth eigenpairs defined in a neighborhood $U$ of $\pi$ by \eqref{eq:def_smooth_bands}. It follows from \eqref{eq:symmetree} that for $p$ away from the degeneracy at $\pi$, $\chi_+(z;p)$ and $\chi_-(z;p)$ obey the symmetry: 
\begin{equation} \label{eq:pm_symmetree}
	\chi_-(z;p) = e^{- 2 \pi i z} \overline{ \chi_+(z;p) }, \quad p \in U \setminus \{ \pi \}.
\end{equation}
But now recall that the maps $\chi_+(z;p), \chi_-(z;p)$ are smooth at $p = \pi$, hence:
\begin{equation}
	\chi_-(z;\pi) = \lim_{p \uparrow \pi} \chi_-(z;p) = \lim_{p \uparrow \pi} e^{- 2 \pi i z} \overline{ \chi_+(z;p) } = e^{- 2 \pi i z} \overline{ \chi_+(z;\pi) }.
\end{equation}
It follows that \eqref{eq:pm_symmetree} holds for every $p \in U$:
\begin{equation} \label{eq:chi_-}
	\chi_-(z;p) = e^{- 2 \pi i z} \overline{ \chi_+(z;p) }, \quad p \in U.
\end{equation}
Substituting \eqref{eq:chi_-} into the formula for the `inter-band coupling coefficient' \eqref{eq:coupling_coefficient} gives:
\begin{equation} \label{eq:formula_for_coupling_coeff}
	\ip{\chi_-(z;\pi)}{\de_p \chi_+(z;\pi)} = \ip{e^{- 2 \pi i z} \overline{ \chi_+(z;p) }}{\de_p \chi_+(z;\pi)} = \inty{0}{1}{ e^{2 \pi i z} \chi_+(z;\pi) \de_p \chi_+(z;\pi) }{z}.
\end{equation}

\subsection{Proof that coefficient vanishes for trivial crossings} 
Now, suppose that $E_n(p)$ and $E_{n+1}(p)$ cross \emph{trivially} in the sense that $V(z) = V_{1/2}(z)$, where $V_{1/2}(z)$ denotes a $1/2$-periodic function, and the smooth band functions $E_+(p), E_-(p)$ and associated eigenfunctions $\chi_+(z;p)$, $\chi_-(z;p)$ defined in a neighborhood of $p = \pi$ satisfy (all equality of eigenfunctions understood as holding up to a constant phase):
\begin{align}
	&E_+(p) = \tilde{E}(p) &&\chi_+(z;p) = \tilde{\chi}(z;p) \nonumber	\\
	&E_-(p) = \tilde{E}(2\pi + p) &&\chi_-(z;p) = \tilde{\chi}(z;2\pi + p) \label{eq:smooth_bands_from_folding}
\end{align}
where $\tilde{E}(p)$ is an eigenvalue band of the Bloch eigenvalue problem \eqref{eq:reduced_eigenvalue_problem} with potential $V(z) = V_{1/2}(z)$ and \underline{$1/2$ -periodic} boundary conditions, considered on the Brillouin zone $[0,4\pi]$ (see Figure \ref{fig:folded_bands}). 

\eqref{eq:smooth_bands_from_folding} in particular implies that $\chi_+(z;p)$, $\de_p \chi_+(z;p)$, and the function:
\begin{equation} \label{eq:one_half_periodic}
	\chi_+(z;p) \de_p \chi_+(z;p),	
\end{equation}
are all $1/2$ -periodic for all $p \in U$. It follows that the function \eqref{eq:one_half_periodic} has for all $p \in U$ a convergent Fourier series with only \emph{even index} modes:
\begin{equation} \label{eq:even_index_FS}
\begin{split}
	&\chi_+(z;p) \de_p \chi_+(z;p) = \sum_{m \in 2 \field{Z}} c_m e^{2\pi i m z},	\\
	&c_m := \inty{0}{1}{ e^{- 2 \pi i m z'} \chi_+(z';p) \de_p \chi_+(z';p) }{z'}.
\end{split}
\end{equation}
Substituting \eqref{eq:even_index_FS} into \eqref{eq:formula_for_coupling_coeff} we have that: 
\begin{equation}
	\ip{\chi_-(z;\pi)}{\de_p \chi_+(z;\pi)} = \sum_{m \in 2 \field{Z}} c_m \inty{0}{1}{ e^{2 \pi i z} e^{2 \pi i m z} }{z},
\end{equation}
where the $c_m$ are as in \eqref{eq:even_index_FS}. But:
\begin{equation}
	\inty{0}{1}{ e^{2 \pi i (m + 1) z} }{z} = 0, \quad m \in 2 \field{Z} \, .
\end{equation}
Hence $\ip{\chi_-(z;\pi)}{\de_p \chi_+(z;\pi)} = 0$ as required. 

\end{appendices}

\printbibliography

\end{document}